\documentclass[11pt]{article}
\usepackage{amsmath,amsthm,amssymb,amscd,epsfig,url}
\usepackage{fullpage}
\usepackage{times}
\newtheorem{theorem}{Theorem}[section]
\newtheorem{lemma}[theorem]{Lemma}
\newtheorem{proposition}[theorem]{Proposition}
\newtheorem{corollary}[theorem]{Corollary}

\newtheorem{definition}[theorem]{Definition}

\newtheorem{remark}[theorem]{Remark}

\def\Vmst{{V(G) \setminus \{s,t\}}}

\def\<{{\langle}} \def\>{{\rangle}}       
\begin{document}

\author{Aaron Potechin \\
MIT\\
potechin@mit.edu}

\title{Improved upper and lower bound techniques for monotone switching networks for directed connectivity}

\maketitle
\begin{abstract}
\noindent In this paper, we analyze the monotone space of complexity of directed connectivity for a large class of 
input graphs $G$ using the switching network model. The upper and lower bounds we obtain are a significant generalization 
of previous results and the proofs involve several completely new techniques and ideas.
\end{abstract}

\thispagestyle{empty}
\noindent \textbf{Acknowledgement:}\\
This material is based on work supported by the National Science Foundation Graduate Research Fellowship 
under Grant No. 0645960.
\newpage
\section{Introduction}\label{intro}
\setcounter{page}{1}

$L$ versus $NL$, the problem of whether non-determinism helps in logarithmic space 
bounded computation, is a longstanding open question in computational complexity. At present, only 
a few results are known. It is known that the problem is equivalent to the question of whether there is 
a log-space algorithm for the \textit{directed connectivity} problem, namely given an $n$ vertex directed 
graph $G$ and pair of vertices $s,t$, find out if there is a directed path from $s$ to $t$ in $G$. 
Savitch \cite{savitch} gave an $O(\log^{2}n)$-space deterministic algorithm for 
directed connectivity, thus proving that $NSPACE(g(n)) \subseteq DSPACE((g(n)^2))$ for 
every space constructable function $g$. Immerman \cite{nlconlone} and 
Szelepcs\'{e}nyi \cite{nlconltwo} independently gave an $O(\log n)$-space non-deterministic algorithm 
for directed \textit{non-connectivity}, thus proving that $NL = co$-$NL$. For the problem of 
\textit{undirected connectivity} (i.e. where the input graph $G$ is undirected), a probabilistic 
algorithm was shown using random walks by Aleliunas, Karp, Lipton, Lov\'{a}sz, and Rackoff \cite{randomwalk}, 
and Reingold \cite{undirectedgraph} gave a deterministic 
$O(\log n)$-space algorithm for the same problem, showing that undirected connectivity is in $L$. 
Trifonov \cite{trifonov} independently gave an $O(\lg{n}\lg{\lg{n}})$ space algorithm for undirected connectivity.

In our previous work \cite{potechin}, we separated monotone analogues of $L$ and $NL$ using the 
switching network model. However, stronger results are needed before we have any hope of extending this 
approach to the non-monotone case. The reason is that in \cite{potechin} we 
analyzed input graphs $G$ consisting only of a path from $s$ to $t$ and isolated vertices. This type of graph 
was a natural place to start, as it is the hardest type of graph for monotone models to solve. However, it is very 
easy for non-monotone algorithms to solve directed connectivity on this type of graph, as we can just follow the path from $s$ to $t$. 
The reason for this gap in difficulty is that in following the path from $s$ to $t$ we are using the fact that 
at each vertex $v$ in the path, there is only one vertex to go to next. This uses the information that the other edges going out 
from $v$ are NOT in $G$. Monotone models can only use the information of which edges are in $G$, not which edges are NOT in $G$, so 
monotone models cannot use this idea.

To have any hope of extending monotone lower space bounds to non-monotone lower space bounds, we must be able to anaylze 
input graphs $G$ which we believe are hard even for non-monotone algorithms to solve. In this paper, we do this, analyzing 
a much wider class of input graphs $G$. While the overall idea is the same as before, the analysis is extremely different, 
requiring completely new and considerably more sophisticated techniques. Our bounds and the techniques we use are 
interesting on their own and these techniques are more robust and thus more likely to be generalizable to non-monotone analysis.

This paper does not assume prior knowledge of switching networks and the techniques used to analyze them. That said, the paper builds 
on intuition from previous work, so it is recommended that a reader who is learning about this approach for the first time read either \cite{potechin}, \cite{potechinsiuman}, or \cite{averagecase} before reading this paper.
\subsection{Definitions}\label{definitions}
We now give definitions which will be used throughout the paper and will allow us to give an overview of the paper and 
precisely state our results. Many of these definitions are from \cite{potechin}.
\begin{definition}\label{modifiedswitchingdefinition}
A switching network for directed connectivity on a set of vertices $V(G)$ with distinguished vertices 
$s,t$ is a tuple $< G', s', t', {\mu}' >$ where $G'$ is an undirected multi-graph with distinguished 
vertices $s'$,$t'$ and ${\mu}'$ is a labeling function giving each edge $e' \in E(G')$ a label of the form 
$v_1 \to v_2$ or $\neg(v_1 \to v_2)$ for some vertices $v_1, v_2 \in V(G)$ with $v_1 \neq v_2$.
\begin{enumerate}
\item We say that $G'$ accepts an input graph $G$ with vertex set $V(G)$ if there is a path $P'$ in $G'$ from 
$s'$ to $t'$ such that for each edge $e' \in E(P')$, $\mu'(e')$ is consistent with the input graph $G$ 
(i.e. of the form $e$ for some edge $e \in E(G)$ or $\neg{e}$ for some $e \notin E(G)$).
\item We say that $G'$ is sound if it does not accept any input graphs $G$ on 
the set of vertices $V(G)$ which do not have a path from $s$ to $t$.
\item We say that $G'$ is complete if it accepts all input graphs $G$ on the 
set of vertices $V(G)$ which have a path from $s$ to $t$.
\item We say that $G'$ solves directed connectivity on $V(G)$ if $G'$ is both complete and sound.
\item We take the size of $G'$ to be $n' = |V(G') \setminus \{s',t'\}|$.
\item We say that $G'$ is monotone if it has no labels of the form $\neg(v_1 \to v_2)$.
\end{enumerate}
\end{definition}
\begin{definition}\label{inputdifficulty}
Given a set $I$ of input graphs on a set $V(G)$ of vertices with distinguished vertices $s,t$ where each graph 
in $I$ contains a path from $s$ to $t$, let $m(I)$ be the size of the smallest sound monotone switching network 
for directed connectivity on $V(G)$ which accepts all of the input graphs in $I$.
\end{definition}
Just as in \cite{potechin}, we analyze switching networks sound monotone switching networks $G'$ by 
looking at cuts of the vertices of the input graph $G$.
\begin{definition}
We define an s-t cut (below we use cut for short) of $V(G)$ to be a partition of $V(G)$ into subsets $L(C),R(C)$ such that 
$s \in L(C)$ and $t \in R(C)$. We say an edge $v_1 \to v_2$ crosses $C$ if $v_1 \in L(C)$ and $v_2 \in R(C)$.
Let $\mathcal{C}$ denote the set of all cuts $C$ of $V(G)$.
\end{definition}
\begin{definition}
For a cut $C$, define the input graph $G(C)$ to be the graph with vertex set $V(G)$ and edge set \\
$E(G(C)) = \{e: e$ does not cross $C\}$
\end{definition}
We have the following vector space, dot product, and Fourier basis for this space:
\begin{definition}
Given two functions 
$f,g: \mathcal{C} \to \mathbb{R}$, $f \cdot g = 2^{-n}\sum_{C \in \mathcal{C}}{f(C)g(C)}$ where 
$n = |\Vmst|$
\end{definition}
\begin{definition}
Given a set of vertices $V \subseteq V(G) \backslash \{s,t\}$, define $e_V: \mathcal{C} \to \mathbb{R}$ by 
$e_V(C) = {(-1)^{|V \cap L(C)|}}$.
\end{definition}
\begin{proposition}
The set $\{e_V, V \subseteq V(G) \backslash \{s,t\}\}$ is an orthonormal basis for the vector space of functions from 
$\mathcal{C}$ to $\mathbb{R}$. 
\end{proposition}
\begin{definition}
Given a function $f: \mathcal{C} \to \mathbb{R}$ and a set of vertices $V \subseteq V(G) \backslash \{s,t\}$, 
define $\hat{f_V} = f \cdot e_V$.
\end{definition}
\begin{proposition}[Parseval's Theorem]\label{parsevals}
For any functions $f,g: \mathcal{C} \to \mathbb{R}$, $f \cdot g = \sum_{V \subseteq \Vmst}{\hat{f}_{V}\hat{g}_{V}}$
\end{proposition}
\begin{definition}
We say a function $g: \mathcal{C} \to \mathbb{R}$ is $e$-invariant for some edge $e$ if $g(C) = 0$ for 
any cut $C$ which $e$ crosses.
\end{definition}
In this paper, we will take a random permutation of the vertices $\Vmst$ and analyze what happens.
\begin{definition}
For each permutation $\sigma \in S_{\Vmst}$,
\begin{enumerate}
\item Given an edge $e = v \to w$, define $\sigma(e) = \sigma(v) \to \sigma(w)$.
\item Given an input graph $G$ with vertices $V(G)$, define $\sigma(G)$ to be the 
graph with vertices $V(G)$ and edges $\{\sigma(e), e \in E(G)\}$.
\item Given a cut $C \in \mathcal{C}$, define $\sigma(C)$ to be the cut such that  
$L(\sigma(C)) = \sigma(L(C))$ and $R(\sigma(C)) = \sigma(R(C))$.
\item Given a function $g: \mathcal{C} \to \mathbb{R}$, define the function 
$\sigma(g): \mathcal{C} \to \mathbb{R}$
so that $\sigma(g)(\sigma(C)) = g(C)$, i.e. $\sigma(g)(C) = g(\sigma^{-1}(C))$.
\end{enumerate}
\end{definition}
\begin{proposition}\label{permutationproperties}
For all $\sigma \in S_{\Vmst}$, 
\begin{enumerate}
\item An edge $e$ crosses a cut $C$ if and only if $\sigma(e)$ crosses $\sigma(C)$.
\item $g: \mathcal{C} \to \mathbb{R}$ is $e$-invariant for some edge $e$ if and only if $\sigma(g)$ is $\sigma(e)$-invariant.
\item For all $f,g: \mathcal{C} \to \mathbb{R}$, $\sigma(f) \cdot \sigma(g) = f \cdot g$
\item For all $V \subseteq \Vmst, \sigma(e_V) = e_{\sigma(V)}$
\item For all $g: \mathcal{C} \to \mathbb{R}$, for all $V \subseteq \Vmst$, $\hat{\sigma(g)}_{\sigma(A)} = \hat{g}_A$
\end{enumerate}
\end{proposition}
\begin{proof}
Statements 1 and 3 are clear, statement 2 follows from statement 1, and statement 5 follows from statements 3 and 4. Thus 
we only need to show statement 4. To see statment 4, note that \\
$\sigma(e_V)(C) = e_V(\sigma^{-1}(C)) = (-1)^{|L(\sigma^{-1}(C)) \cap V|} = (-1)^{|\sigma^{-1}(L(C)) \cap V|} = 
(-1)^{|L(C) \cap \sigma(V)|} = e_{\sigma(V)}(C)$
\end{proof}
\begin{definition}
For an input graph $G$ containing a path from $s$ to $t$, define \\
$m(G) = m(\{\sigma(G): \sigma \in S_{\Vmst}\})$
\end{definition}
\begin{remark}
The function $m(G)$ is a complexity measure on $G$ which measures how hard it is for sound monotone switching networks 
to solve directed connectivity on $G$.
\end{remark}
\subsection{Notation and conventions}
We use the same notation and conventions as \cite{potechin}. For the remainder of the paper, 
we will assume without explicitly stating it that $V(G)$ is a set of 
vertices with distinguished vertices $s,t$ and $n = |\Vmst|$. Throughout the paper, we use lower case 
letters (e.g $v, e, f$) to denote vertices, edges, and functions. We use upper case letters (e.g $G, V, E$) to 
denote graphs and sets of vertices, edges, or other objects. Uppercase script letters (e.g. $\mathcal{C}$) 
are often used to denote a family or set of objects which are themselves graphs or sets. We use 
unprimed symbols to denote vertices, edges, etc. in the directed graph $G$, and we use primed symbols to 
denote vertices, edges, etc. in the switching network $G'$.
\subsection{Technical comparison with previous work}
The main result of \cite{potechin} is
\begin{theorem}
If $G$ is an input graph consisting of just a path from $s$ to $t$ of length $l$ and isolated vertices then
$(\frac{n}{64(l-1)^2})^{\frac{\lceil{\lg{l}}\rceil}{2}} \leq m(G) 
\leq n^{\lceil{\lg{l}}\rceil}$
\end{theorem} 
The main lower bound result of this paper is 
\begin{theorem}\label{lowerboundforlowconnectivity}
If $z$ and $m$ are constants such that $m \leq \frac{n}{2000z^{4}}$ and $G$ is a directed acyclic input graph such that 
\begin{enumerate}
\item There is no path of length at most $2^{z-1}$ from $s$ to $t$.
\item For any vertex $v \in V(G)$, there are at most $m$ vertices $w \in V(G)$ such that either there is a path of length at 
most $2^{z-2}$ from $v$ to $w$ in $G$ or there is a path of length at 
most $2^{z-2}$ from $w$ to $v$ in $G$
\end{enumerate}
then $$m(G) \geq \frac{(9mn)^{\frac{1}{4}}}{20|E(G)|(z+1)\sqrt{2^{z}z!}}(\frac{n}{9m})^\frac{z}{4}$$
\end{theorem}
For both papers, the general idea is as follows. We aim to construct a good set of functions 
$\{g_e: e \in E(G)\}$ where $g_e$ is $e$-invariant for any $e$. For any switching network $G'$, we will assign each 
vertex $v' \in V(G')$ a function from $\mathcal{C} \to \mathbb{R}$ (which we identify with the vertex for 
convenience) where $s'(C) = -1$, $t'(C) = 1$ for all $C$, and whenever there is an edge with label $e$ between 
vertices $v'$ and $w'$ in $G'$, $v'(C) = w'(C)$ for all cuts $C$ which are not crossed by $e$. This implies that 
$v' \cdot g_e$ = $w' \cdot g_e'$ whenever there is an edge with label $e$ between $v'$ and $w'$ in $G'$

Now consider the following game. There are $|E(G)|$ players, each of which has an edge $e$. At any point, the player with 
edge $e$ has value $g_e \cdot v'$ where $v'$ is the vertex in $G'$ which we are currently at. We will choose the functions 
$\{g_e: e \in E(G)\}$ so that for all $e$, $s' \cdot g_e = -1$ and $t' \cdot g_e = 1$. Thus, all players start with value $-1$ 
and must end at value $1$. However, whenever we go along an edge with label $e$ in $G'$, the player who has $e$ cannot change his/her 
value while every other player can change values. This means that if $G'$ is small then at some point there 
will be a large discrepency in the values of the players. This corresponds to a vertex $v' \in V(G')$ and 
edges $e_1,e_2 \in E(G)$ such that $v' \cdot (g_{e_2} - g_{e_1})$ is large.

Both papers then use this to prove a lower size bound on the size of $G'$. In \cite{potechin} we take a small input graph 
$G_0$ consisting of just a path of length $l$ from $s$ to $t$ and find functions $\{g_e: e \in E(G_0)\}$ for this graph where for 
any $e_2,e_1 \in E(G_0)$
$(g_{e_2} - g_{e_1}) \cdot e_V = 0$ unless $V \subseteq V(G_0)$ and $|V| \geq \lceil{\lg{l}}\rceil$. We then add isolated vertices to 
the graph, keeping the same set of functions $\{g_e: e \in E(G_0)\}$ (as expressed in terms of their Fourier coefficients). 
Now if we choose permutations $\{\sigma_i\}$ such that 
$\sigma_{i_2}(V(G_0)) \cap \sigma_{i_1}(V(G_0)) < \lceil{\lg{l}}\rceil$ for any distinct $i_1,i_2$ then for any distinct 
$i_1$ and $i_2$, for any $e_1,e_2,e_3,e_4 \in E(G_0)$, 
$\sigma_{i_1}(g_{e_2} - g_{e_2})$ is orthogonal to $\sigma_{i_2}(g_{e_4} - g_{e_3})$. Using this orthogonality and 
the discrepency in progress idea above, we prove a good lower size bound on $G'$.

The trouble with this type of argument is that it requires the input graph to mostly consist of isolated vertices, which is 
extremely restrictive. In this paper, we instead choose the functions $\{g_e: e \in E(G)\}$ for the input graph directly. 
We choose these functions so that for any vertex $v' \in V(G')$, for any edges $e_2,e_1 \in E(G)$, 
for a random permutation $\sigma$ of the vertices $\Vmst$ the expected value of $v' \cdot \sigma(g_{e_2} - g_{e_1})$ is small. 
We then use this to prove our lower size bound on $G'$. As shown in Section \ref{lowerbounds}, this technique works well
for a much wider class of input graphs $G$.

Note that for both papers, our lower bound on $m(G)$ is $n^{\Omega(\lg{l})}$ where $l$ is the length of the shortest path 
from $s$ to $t$ and there are some conditions on $G$. We may ask what the conditions on $G$ should be for this bound to hold. 
Does it hold for all directed acyclic input graphs $G$ or are there directed input graphs 
$G$ for which $m(G)$ is significantly smaller than $n^{c(\lg{l})}$ for any $c > 0$? In \cite{potechin} we 
only get the upper bound of $n^{O(\lg{l})}$ which comes from Savitch's algorithm. In section \ref{upperbounds} of 
this paper we give a class of input graphs $G$ for which we can obtain significantly better upper bounds on $m(G)$. This 
shows that we do indeed need additional conditions on $G$ to get the lower bound of $n^{\Omega(\lg{l})}$ on $m(G)$. Also, 
while it is not shown here, for these graphs $G$ we still have a lower bound of $n^{\Omega(\lg{l})}$ for all 
certain-knowledge switching networks (defined in Section 3 of \cite{potechin}) which accept all of the 
inputs $\{\sigma(G): \sigma \in S_{\Vmst}\}$. This shows that monotone switching networks are strictly more powerful than 
certain knowledge switching networks.
\section{Lower bounds}\label{lowerbounds}
In this section, we prove bounds on $m(G)$ for a large class of directed 
acyclic input graphs $G$ by carefully constructing a set of functions $\{g_e: e \in E(G)\}$ corresponding to $G$. 
In subsections \ref{permutations} and \ref{lowerboundsfrompermutationaverages} we show what properties our 
set of functions should have to give us good lower bounds. In subsection \ref{einvarianceconditions} 
we explore what these properties say about our set of functions. The remaining subsections are devoted to showing how to 
construct our set of functions for the input graph $G$.
\subsection{Calculating Permutation Averages}\label{permutations}
For our lower bounds, we will need bounds on expressions of the form $E_{\sigma \in S_{\Vmst}}[f \cdot \sigma(g)]^2$ 
where $f,g$ are functions from $\mathcal{C}$ to $\mathbb{R}$ and $\sigma(g)$ is a permutation of $g$. 
Here we obtain bounds on $E_{\sigma \in S_{\Vmst}}[f \cdot \sigma(g)]^2$ 
in terms of the norm $||f||$ of $f$ and certain sums associated to the function $g$. This will be extremely useful 
because the function $f$ will correspond to a vertex in an arbitrary sound monotone 
switching network $G'$ so we will have no control over anything except for $||f||$. However, we will have a 
large degree of control over $g$ and will thus be able to adjust the values of many of the sums associated with 
$g$ to give us the bounds we need.
We now state our bound on $E_{\sigma \in S_{\Vmst}}[(f \cdot \sigma(g))^2]$.
\begin{definition}
Given functions $g_1,g_2: \mathcal{C} \to \mathbb{R}$, define 
$$s_{k,u_1,u_2}(g_1,g_2) = 
\sum_{A,B,C \subseteq \Vmst: |A| = k, |B| = u_1, |C| = u_2, \atop A \cap B = A \cap C = B \cap C = \emptyset}
{\hat{g_1}_{A \cup B}\hat{g_2}_{A \cup C}}$$
\end{definition}
\begin{theorem}\label{allequalisbesttheorem}
For any function $g: \mathcal{C} \to \mathbb{R}$ such that $\hat{g}_{V}$ is only nonzero when $|V| \leq \frac{\sqrt{n}}{2} - 1$, 
for any function $f: \mathcal{C} \to \mathbb{R}$, \\
$$E_{\sigma \in S_{\Vmst}}[(f \cdot \sigma(g))^2] \leq 2\sum_{k,u_1,u_2}{\frac{\sqrt{(k+u_1)!(k+u_2)!}}{n^{k + \frac{u_1+u_2}{2}}}}|s_{k,u_1,u_2}(g,g)|\cdot||f||$$
\end{theorem}
\begin{proof}
\begin{lemma}\label{expectationfromsums}
For any functions $f,g: \mathcal{C} \to \mathbb{R}$, 
$$E_{\sigma \in S_{\Vmst}}[(f \cdot \sigma(g))^2] = 
\sum_{k,u_1,u_2}{\frac{k!{u_1}!{u_2}!(n - k - u_1 - u_2)!}{n!}s_{k,u_1,u_2}(f,f)s_{k,u_1,u_2}(g,g)}$$
\end{lemma}
\begin{proof}
$$E_{\sigma \in S_{\Vmst}}[(f \cdot \sigma(g))^2] = E_{\sigma \in S_{\Vmst}}\left[(\sum_{V_1 \subseteq \Vmst}{\hat{f}_{V_1}}{\hat{\sigma(g)}_{V_1}})
(\sum_{V_2 \subseteq \Vmst}{\hat{f}_{V_2}}{\hat{\sigma(g)}_{V_2}})\right]$$
Taking $A = V_1 \cap V_2, B = V_1 \setminus V_2, C = V_2 \setminus V_1$, this is equal to 
$$E_{\sigma \in S_{\Vmst}}\left[\sum_{k,u_1,u_2}
{\sum_{A,B,C \subseteq \Vmst: |A| = k, |B| = u_1, |C| = u_2, \atop A \cap B = A \cap C = B \cap C = \emptyset}
{\hat{f}_{A \cup B}\hat{f}_{A \cup C}\hat{\sigma(g)}_{A \cup B}\hat{\sigma(g)}_{A \cup C}}}\right]$$ 
which is equal to 
\begin{equation}\label{firstequation}
\sum_{k,u_1,u_2}{\sum_{A,B,C \subseteq \Vmst: |A| = k, |B| = u_1, |C| = u_2, 
\atop A \cap B = A \cap C = B \cap C = \emptyset}}
{\hat{f}_{A \cup B}\hat{f}_{A \cup C}E_{\sigma \in S_{\Vmst}}[\hat{\sigma(g)}_{A \cup B}\hat{\sigma(g)}_{A \cup C}]}
\end{equation}
Now by statement 5 of Proposition \ref{permutationproperties}, for any $A,B,C \subseteq \Vmst$ such that $|A| = k$, $|B| = u_1$, 
$|C| = u_2$, and $A \cap B = A \cap C = B \cap C = \emptyset$,
\begin{align*}
E_{\sigma \in S_{\Vmst}}[\hat{\sigma(g)}_{A \cup B}\hat{\sigma(g)}_{A \cup C}] &= 
E_{\sigma \in S_{\Vmst}}[\hat{g}_{\sigma^{-1}(A \cup B)}\hat{g}_{\sigma^{-1}(A \cup C)}] \\
&= \frac{k!{u_1}!{u_2}!(n - k - u_1 - u_2)!}{n!}s_{k,u_1,u_2}(g,g)
\end{align*}
Plugging this into \eqref{firstequation} we deduce that 
$$E_{\sigma \in S_{\Vmst}}[(f \cdot \sigma(g))^2] = 
\sum_{k,u_1,u_2}{\frac{k!{u_1}!{u_2}!(n - k - u_1 - u_2)!}{n!}s_{k,u_1,u_2}(f,f)s_{k,u_1,u_2}(g,g)}$$
which completes the proof of Lemma \ref{expectationfromsums}.
\end{proof}
We now wish to bound how large $|s_{k,u_1,u_2}(f,f)|$ can be in terms of $||f||$.
\begin{definition}
Given a function $g: \mathcal{C} \to \mathbb{R}$ and a subset $A \subseteq \Vmst$, define 
$$s_{A,u_1}(g) = \sum_{B \subseteq \Vmst: |B| = u_1, A \cap B = \emptyset}{\hat{g}_{A \cup B}}$$
\end{definition}
\begin{definition}
Given functions $g_1,g_2: \mathcal{C} \to \mathbb{R}$ and a subset $A \subseteq \Vmst$, define 
$$s_{A,u_1,u_2}(g_1,g_2) = \sum_{B,C \subseteq \Vmst: |B| = u_1, |C| = u_2, 
\atop A \cap B = A \cap C = B \cap C = \emptyset}{\hat{g_1}_{A \cup B}\hat{g_2}_{A \cup C}}$$
\end{definition}
\begin{proposition}\label{pairsfromsingles}
$\forall A, s_{A,u_1}(g_1)s_{A,u_2}(g_2) = \sum_{u \geq 0}{\sum_{A \subseteq B \subseteq \Vmst, |B| = |A| + u}s_{B,u_1-u,u_2-u}(g_1,g_2)}$
\end{proposition}
\begin{corollary}\label{pairsfromsinglescorollary}
$\forall k,u_1,u_2, \sum_{A: |A| = k}{s_{A,u_1}(g_1)s_{A,u_2}(g_2)} = \sum_{u \geq 0}{{{k + u} \choose u}s_{k+u,u_1 - u, u_2 - u}(g_1,g_2)}$
\end{corollary}
\begin{proof}
This follows immediately when we sum Proposition \ref{pairsfromsingles} over all $A$ of size $k$.
\end{proof}
\noindent Surprisingly, Corollary \ref{pairsfromsinglescorollary} has an almost identical inverse formula.
\begin{lemma}\label{pairsfromsingleslemma}
$\forall k,u_1,u_2, s_{k,u_1,u_2}(g_1,g_2) = \sum_{u \geq 0}{(-1)^{u}{{k + u} \choose u}\sum_{A: |A| = k + u}{s_{A,u_1-u}(g_1)s_{A,u_2-u}(g_2)}}$
\end{lemma}
\begin{proof}
This lemma corresponds to the fact that the inverse of a matrix like\\
\[ \left( \begin{array}{cccc}
1 & {{k+1} \choose 1} & {{k+2} \choose 2} & {{k+3} \choose 3} \\
0 & 1 & {{k+2} \choose 1} & {{k+3} \choose 2} \\
0 & 0 & 1  & {{k+3} \choose 1} \\
0 & 0 & 0 & 1 \end{array} \right)\]
is 
\[ \left( \begin{array}{cccc}
1 & -{{k+1} \choose 1} & {{k+2} \choose 2} & -{{k+3} \choose 3} \\
0 & 1 & -{{k+2} \choose 1} & {{k+3} \choose 2} \\
0 & 0 & 1  & -{{k+3} \choose 1} \\
0 & 0 & 0 & 1 \end{array} \right)\]
Calculating directly, every entry in the product of these matrices has the form 
$\sum_{u=0}^{m}{(-1)^{m-u}{{j+u} \choose u}{{j+m} \choose {m-u}}}$ 
where $j$ is $k$ plus the row number and $m$ is the column number minus the row number. Now if $m \geq 1$, \\
$\sum_{u=0}^{m}{(-1)^{m-u}{{j+u} \choose u}{{j+m} \choose {m-u}}} = 
\frac{(j+m)!}{j!}\sum_{u=0}^{m}{(-1)^k\frac{1}{u!(m-u)!}} = \frac{(j+m)!}{j!m!}(1 + (-1))^m = 0$. If $m < 0$ then 
$\sum_{u=0}^{m}{(-1)^{m-u}{{j+u} \choose u}{{j+m} \choose {m-u}}} = 0$ and if $m = 0$ then 
$\sum_{u=0}^{m}{(-1)^{m-u}{{j+u} \choose u}{{j+m} \choose {m-u}}} = 1$.
\end{proof}
To bound how large $|s_{k,u_1,u_2}(f,f)|$ can be in terms of $||f||$, we just need bounds for how large expressions 
of the form $\sum_{A:|A| = k}{s_{A,u_1}(g)s_{A,u_2}(g)}$ can be.
\begin{proposition}\label{boundingbysquares}
$\forall k,u_1,u_2, |\sum_{A:|A| = k}{s_{A,u_1}(g)s_{A,u_2}(g)}| \leq 
\sqrt{\sum_{A:|A| = k}{(s_{A,u_1}(g))^2}\sum_{A:|A| = k}{(s_{A,u_2}(g))^2}}$
\end{proposition}
\begin{proof}
This follows immediately from the Cauchy-Schwarz inequality.
\end{proof}
\begin{proposition}\label{allequalisbest}
For all $k,u_1$ and all $A$ with $|A| = k$, 
$$(s_{A,u_1}(g))^2 \leq {{n - k} \choose {u_1}}
\sum_{B \subseteq \Vmst: |B| = u_1, \atop A \cap B = \emptyset}{(\hat{g}_{A \cup B})^2}$$
\end{proposition}
\begin{proof}
This follows from the Cauchy-Schwarz inequality 
$$\left(\sum_{B \subseteq \Vmst: |B| = u_1, \atop A \cap B = \emptyset}{f(B)h(B)}\right)^2 \leq 
\sum_{B \subseteq \Vmst: |B| = u_1, \atop A \cap B = \emptyset}{(f(B))^2}
\sum_{B \subseteq \Vmst: |B| = u_1, \atop A \cap B = \emptyset}{(h(B))^2}$$
with $f(B) = 1$ and $h(B) = \hat{g}_{A \cup B}$.
\end{proof}
\begin{corollary}\label{allequalisbestcorollary}
$\forall k,u_1, \sum_{A:|A| = k}{(s_{A,u_1}(g))^2} \leq 
{{n - k} \choose {u_1}}{{k + u_1} \choose {u_1}}\sum_{B:|B| = k + u_1}{(\hat{g}_B)^2} \leq 
{{n - k} \choose {u_1}}{{k + u_1} \choose {u_1}}||g||^2$
\end{corollary}
\begin{proof}
This follows immediately when we sum Proposition \ref{allequalisbest} over all $A$ of size $k$.
\end{proof}
\begin{corollary}\label{allequalisbestcorollarytwo}
$\forall k,u_1,u_2, |\sum_{A:|A| = k}{s_{A,u_1}(g)s_{A,u_2}(g)}| \leq \sqrt{{{n - k} \choose {u_1}}{{n - k} \choose {u_2}}
{{k + u_1} \choose {u_1}}{{k + u_2} \choose {u_2}}} \cdot ||g||$
\end{corollary}
\begin{proof}
This follows immediately from Proposition \ref{boundingbysquares} and Corollary \ref{allequalisbestcorollary}
\end{proof}
\begin{corollary}\label{originalallequalisbestcorollarythree}
If $f$ is a function $f: \mathcal{C} \to \mathbb{R}$, we have that for all $k, u_1, u_2$, 
$$|s_{k,u_1,u_2}(f,f)| \leq \sum_{u \geq 0}{{{k + u} \choose u}}
\sqrt{{{n - k - u} \choose {u_1 - u}}{{n - k - u} \choose {u_2 - u}}{{k + u_1} \choose {u_1 - u}}{{k + u_2} \choose {u_2 - u}}}||f||$$
\end{corollary}
\begin{proof}
This follows immediately from Corollary \ref{allequalisbestcorollarytwo} and Lemma \ref{pairsfromsingleslemma}.
\end{proof}
\begin{corollary}\label{allequalisbestcorollarythree}
If $n \geq k + \max\{u_1,u_2\} + 2(\max\{u_1,u_2\})^2$ then 
$$|s_{k,u_1,u_2}(f,f)| \leq 2\sqrt{{{n} \choose {u_1}}{{n} \choose {u_2}}{{k + u_1} \choose {u_1}}{{k + u_2} \choose {u_2}}}||f||$$
\end{corollary}
\begin{proof}
The idea is to show that each term in the sum in \ref{originalallequalisbestcorollarythree} is at most half of the 
previous term. This follows from the equations
\begin{enumerate}
\item $\frac{{{k + u + 1} \choose {u+1}}}{{{k + u} \choose u}} = \frac{k+u+1}{u+1}$
\item $\frac{{{n - k - u - 1} \choose {u_1 - u - 1}}{{n - k - u - 1} \choose {u_2 - u - 1}}}
{{{n - k - u} \choose {u_1 - u}}{{n - k - u} \choose {u_2 - u}}} = \frac{(u_1 - u)(u_2 - u)}{(n-k-u)^2}$
\item $\frac{{{k + u_1} \choose {u_1 - u - 1}}{{k + u_2} \choose {u_2 - u - 1}}}
{{{k + u_1} \choose {u_1 - u}}{{k + u_2} \choose {u_2 - u}}} = \frac{(u_1 - u)(u_2 - u)}{(k+u+1)^2}$
\end{enumerate}
\end{proof}
Theorem \ref{allequalisbesttheorem} now follows easily. By Lemma \ref{expectationfromsums}, 
$$E_{\sigma \in S_{\Vmst}}[(f \cdot \sigma(g))^2] = 
\sum_{k,u_1,u_2}{\frac{k!{u_1}!{u_2}!(n - k - u_1 - u_2)!}{n!}s_{k,u_1,u_2}(f,f)s_{k,u_1,u_2}(g,g)}$$
Plugging Corollary \ref{allequalisbestcorollarythree} into this gives 
$$E_{\sigma \in S_{\Vmst}}[(f \cdot \sigma(g))^2] \leq 2\sum_{k,u_1,u_2}{\frac{\sqrt{(k+u_1)!(k+u_2)!}}{n^{k + \frac{u_1+u_2}{2}}}}|s_{k,u_1,u_2}(g,g)| \cdot ||f||$$ 
as needed. To check that $n \geq k + \max\{u_1,u_2\} + 2(\max\{u_1,u_2\})^2$ holds when needed, note that by our assumption that 
$\hat{g}_{V}$ is only nonzero when $|V| \leq \frac{\sqrt{n}}{2}-1$, we may ignore all terms where 
$k + \max{\{u_1,u_2\}} > \frac{\sqrt{n}}{2}-1$. Thus, for all of our terms, 
$n \geq 2(k + \max{\{u_1,u_2\}}+1)^2 \geq k + \max\{u_1,u_2\} + 2(\max\{u_1,u_2\})^2$, as needed.
\end{proof}
\noindent For the functions $g$ we will be looking at, it is difficult to bound $|s_{k,u_1,u_2}(g,g)|$ directly. We would like a 
bound in terms of the sums $s_{A,u}(g)$. To obtain such a bound, we apply Lemma \ref{pairsfromsingleslemma} and 
Proposition \ref{boundingbysquares} to Theorem \ref{allequalisbesttheorem}.
\begin{corollary}\label{keycorollary}
For any function $g: \mathcal{C} \to \mathbb{R}$ such that $\hat{g}_{V}$ is only nonzero when $|V| \leq z$ for 
some $z \leq \frac{\sqrt{n}}{2}-1$, for any function $f: \mathcal{C} \to \mathbb{R}$, 
$$E_{\sigma \in S_{\Vmst}}[(f \cdot \sigma(g))^2] \leq 
2(z+1)||f||\sum_{k,u}{\frac{2^{k}(k+u)!}{n^{k+u}}\sum_{A:|A| = k}{(s_{A,u}(g))^2}}$$
\end{corollary}
\begin{proof}
By Lemma \ref{pairsfromsingleslemma} and Proposition \ref{boundingbysquares}, 
\begin{align*}
2\sum_{k,u_1,u_2}&{\frac{\sqrt{(k+u_1)!(k+u_2)!}}{n^{k + \frac{u_1+u_2}{2}}}}|s_{k,u_1,u_2}(g,g)| \\
&\leq 2\sum_{k,u_1,u_2,u}{{{k+u} \choose u}\frac{\sqrt{(k+u_1)!(k+u_2)!}}{n^{k + \frac{u_1+u_2}{2}}}}
\sqrt{\sum_{A:|A| = k + u}{(s_{A,u_1-u}(g))^2}\sum_{A:|A| = k + u}{(s_{A,u_2-u}(g))^2}}
\end{align*}
Replacing $k+u$ with $k$, $u_1 - u$ with $u_1$, and $u_2 - u$ with $u_2$ on the right hand side, 
\begin{align*}
 2\sum_{k,u_1,u_2}&{\frac{\sqrt{(k+u_1)!(k+u_2)!}}{n^{k + \frac{u_1+u_2}{2}}}}\,|s_{k,u_1,u_2}(g,g)| \\
& \leq 
2\sum_{k,u_1,u_2}\left(\sum_{u}{k \choose u}\right)
\frac{\sqrt{(k+u_1)!(k+u_2)!}}{n^{k + \frac{u_1+u_2}{2}}}
\sqrt{\sum_{A:|A| = k}{(s_{A,u_1}(g))^2}\sum_{A:|A| = k}{(s_{A,u_2}(g))^2}} \\
&= 
2\sum_{k}{\left(\sum_{u_1}{\frac{2^{\frac{k}{2}}\sqrt{(k+u_1)!}}{n^{\frac{k+u_1}{2}}}
\sqrt{\sum_{A:|A| = k}{(s_{A,u_1}(g))^2}}}\right)^2}\\
& \leq 
2(z+1)\sum_{k}{\sum_{u_1}{\left(\frac{2^{\frac{k}{2}}\sqrt{(k+u_1)!}}{n^{\frac{k+u_1}{2}}}
\sqrt{\sum_{A:|A| = k}{(s_{A,u_1}(g))^2}}\right)^2}} \\
&= 
2(z+1)\sum_{k,u_1}{\frac{2^{k}(k+u_1)!}{n^{k+u_1}}\sum_{A:|A| = k}{(s_{A,u_1}(g))^2}}
\end{align*}
Plugging this in to Theorem \ref{allequalisbesttheorem} gives the desired result.
\end{proof}
\subsection{Lower bounds from permutation averages}\label{lowerboundsfrompermutationaverages}
In this subsection we show how we can obtain lower size bounds on monotone switching networks using 
Corollary \ref{keycorollary}. 
\begin{definition}
Given a monotone switching network $G'$, for each vertex $v' \in V(G')$, assign $v'$ the function
$v': \mathcal{C} \to \mathbb{R}$ so that 
$v'(C) = -1$ if there is a path from $s'$ to $t'$ in $G'$ wohse edge labels are all in $E(G(C))$ 
and $v'(C) = 1$ otherwise.
\end{definition}
\begin{remark}
This is the reachability function description of $G'$ from \cite{potechin}.
\end{remark}
\begin{proposition}
For any monotone switching network $G'$,
\begin{enumerate}
\item For all $C \in \mathcal{C}, s'(C) = -1$
\item If $G'$ is sound then for all $C \in \mathcal{C}, t'(C) = 1$
\item For all $v' \in V(G')$, $|v'| = 1$
\item If there is an edge $e' \in G'$ with label $e$ between vertices $v'$ and $w'$ in $G'$, 
$C \in \mathcal{C}$, and $e$ does not cross $C$, then $v'(C) = w'(C)$.
\end{enumerate}
\end{proposition}
Property 4 is extremely useful, as for carefully chosen functions $g$ it gives us information 
about the dot products $\{v' \cdot g, v' \in V(G')\}$, which will give us our lower bounds. 
\begin{proposition}
If there is an edge with label $e$ between vertices $v',w' \in V(G')$ and $g$ is an $e$-invariant function 
then $v' \cdot g = w' \cdot g$.
\end{proposition}
\begin{definition}
For an input graph $G$ with a path from $s$ to $t$, we say that $F_G = \{g_e: e \in E(G)\}$ is a set of 
invariant functions for $G$ if 
\begin{enumerate}
\item For all $e \in E(G)$, $g_e$ is $e$-invariant.
\item For all $e \in E(G)$, $g_e \cdot e_{\{\}} = 1$
\end{enumerate}
\end{definition}
\begin{theorem}\label{lowerboundtheorem}
Let $G$ be an input graph containing a path from $s$ to $t$ and let $F_G = \{g_e: e \in E(G)\}$ be a set of invariant 
functions for $G$. If for all $e \in E(G)$ we have that $\hat{g_e}_{V}$ is only nonzero when $|V| \leq z$ 
for some $z \leq \frac{\sqrt{n}}{2}-1$, then for any edge $e_0 \in E(G)$,
$$m(G) \geq \frac{2}{|E(G)| - 1}{\left(\max_{e \in E(G) \setminus \{e_0\}}
\left\{2(z+1)\sum_{k,u}{\frac{2^{k}(k+u)!}{n^{k+u}}\sum_{A:|A| = k}{(s_{A,u}(g_e - g_{e_0}))^2}}\right\}\right)}^{-\frac{1}{2}}$$
\end{theorem}
\begin{proof}
\begin{definition}
For a sound monotone switching network $G'$ and a path $P'$ from $s'$ to $t'$, 
for each edge $e' \in E(P')$, define $\Delta(P',e') = v'_{end} - v'_{start}$ where 
$e'$ goes from $v'_{start}$ to $v'_{end}$ in $P'$
\end{definition}
\begin{lemma}\label{discrepancyinprogress}
Let $G$ be an input graph containing a path from $s$ to $t$ and let $F_G = \{g_e: e \in E(G)\}$ be a set of invariant 
functions for $G$. For any edge $e_0 \in E(G)$, for any sound monotone switching network $G'$, for any path 
$P'$ in $G'$ from $s'$ to $t'$ whose edge labels are all in $E(G)$, \\
$\sum_{e \in E(G) \setminus \{e_0\}}{\sum_{e' \in E(P'): \mu'(e') \neq e}{\Delta(P',e') \cdot (g_{e} - g_{e_0})}} = 2$
\end{lemma}
\begin{proof}
Let $P'$ be a walk from $s'$ to $t'$ in $G'$ whose edge labels are all in $E(G)$. Since $g_{e}$ is $e$-invariant, \\
$\forall e \in E(G) \setminus \{e_0\}, \sum_{e' \in E(P'): \mu'(e') \neq e}{\Delta(P',e') \cdot g_{e}} = 
\sum_{e' \in E(P')}{\Delta(P',e') \cdot g_{e}} = g_{e} \cdot t' - g_{e} \cdot s' = 2$\\
Since $g_{e_0}$ is $e_0$-invariant, 
\begin{align*}
&\sum_{e \in E(G) \setminus \{e_0\}}{\sum_{e' \in E(P'): \mu'(e') \neq e}{\Delta(P',e') \cdot g_{e_0}}} \\
&= \left((|E(G)|-2)\sum_{e \in E(G) \setminus \{e_0\}}{\sum_{e' \in E(P'): \mu'(e') = e}{\Delta(P',e')}} + 
(|E(G)|-1)\sum_{e' \in E(P'): \mu'(e') = e_0}{\Delta(P',e')}\right) \cdot g_{e_0} \\
&= (k-2)(g_{e_0} \cdot t' - g_{g_0} \cdot s') = 2(|E(G)|-2)
\end{align*}
Putting all of these equations together gives the needed equality.
\end{proof}
\begin{corollary}\label{acorollary}
Let $G$ be an input graph containing a path from $s$ to $t$ and let $F_G = \{g_e: e \in E(G)\}$ be a set of invariant 
functions for $G$. For any edge $e_0 \in E(G)$, for any sound monotone switching network $G'$, for any path 
$P'$ in $G'$ from $s'$ to $t'$ whose edge labels are all in $E(G)$, 
\begin{enumerate}
\item $\sum_{v' \in V(P') \setminus \{s',t'\}}{\sum_{e \in E(G) \setminus \{e_0\}}{|v' \cdot (g_e - g_{e_0})|}} = 
\sum_{e \in E(G) \setminus \{e_0\}}{\sum_{v' \in V(P')}{|v' \cdot (g_e - g_{e_0})|}} \geq 2$
\item $E_{v' \in V(G') \setminus \{s',t'\}, e \in E(G) \setminus \{e_0\}}[|v' \cdot (g_e - g_{e_0})|] \geq 
\frac{2}{(|E(G)| - 1)|V(G') \setminus \{s',t'\}|}$
\item $E_{v' \in V(G') \setminus \{s',t'\}, e \in E(G) \setminus \{e_0\}}[|v' \cdot (g_e - g_{e_0})|^2] \geq 
\frac{4}{(|E(G)| - 1)^2|V(G') \setminus \{s',t'\}|^2}$
\end{enumerate}
\end{corollary}
\begin{proof}
The first statement follows immediately from Lemma \ref{discrepancyinprogress} and the 
fact that for all $e \in E(G) \setminus \{e_0\}$, 
$\sum_{e' \in E(P'): \mu'(e') \neq e}{\Delta(P',e')}$ is a linear combination of 
the vertices of $P'$ where each vertex has coefficient $-1$, $0$, or $1$. The second statement follows 
immediately from the first statement. The third statement follows immediately from the second statement and 
the fact that for any $X$, $E[X^2] \geq (E[X])^2$.
\end{proof}
Theorem \ref{lowerboundtheorem} now follows easily. If $G'$ is a monotone switching network accepting 
all of the input graphs $\{\sigma(G): \sigma \in S_{\Vmst}\}$ then by statement 3 of Corollary \ref{acorollary}, 
$$E_{\sigma \in S_{\Vmst}}\left[E_{v' \in V(G') \setminus \{s',t'\}, e \in E(G) \setminus \{e_0\}}
[|v' \cdot \sigma(g_e - g_{e_0})|^2]\right] \geq \frac{4}{(|E(G)| - 1)^2|V(G') \setminus \{s',t'\}|^2}$$
Applying Corollary \ref{keycorollary} to $g_{e} - g_{e_0}$ and $v'$ for each $e \in E(G) \setminus \{e_0\}$ and 
$v' \in V(G') \setminus \{s',t'\}$, since $||v'|| = 1$ for all $v' \in V(G') \setminus \{s',t'\}$,
$$|V(G') \setminus \{s',t'\}| \geq \frac{2}{|E(G)| - 1}{\left(E_{e \in E(G) \setminus \{e_0\}}
\left[2(z+1)\sum_{k,u}{\frac{2^{k}(k+u)!}{n^{k+u}}\sum_{A:|A| = k}{(s_{A,u}(g_e - g_{e_0}))^2}}\right]\right)}^{-\frac{1}{2}}$$
The result now follows immediately.
\end{proof}
Theorem \ref{lowerboundtheorem} says that for our lower bound, we want to find a set of invariant functions 
$F_G = \{g_e: e \in E(G)\}$ such that for some $e_0 \in E(G)$, for all $e \in E(G) \setminus \{e_0\}$ 
the sums $\sum_{A:|A| = k}{(s_{A,u}(g_e - g_{e_0}))^2}$ are as small as possible.
\subsection{Equations on sum vectors}\label{einvarianceconditions}
Rather than choosing the functions $\{g_e: e \in E(G)\}$ directly, it is more convenient to choose 
the sums $\{s_{A,u}(g_e): A \subseteq \Vmst, u \geq 0, e \in E(G)\}$ and have these sums determine the 
functions $\{g_e: e \in E(G)\}$. Also, it is convenient to group these sums into vectors.
\begin{definition}
For a function $g$ and $k,u \geq 0$, define $\vec{s}_{k,u,g}$ to be the vector with 
one coordinate for each $A \subseteq \Vmst$ where $|A| = k$ such that $(\vec{s}_{k,u,g})_A = s_{A,u}(g)$
\end{definition}
\begin{proposition}
For any function $g: \mathcal{C} \to \mathbb{R}$, $||\vec{s}_{k,u,g}||^2 = \sum_{A:|A| = k}{(s_{A,u}(g))^2}$
\end{proposition}
However, we have to be very careful when choosing the vectors $\vec{s}_{k,u,g}$ because not every collection of vectors 
$\{\vec{s}_{k,u,g}: k,u \geq 0\}$ correspond to an actual function $g$. Here we give equations that 
a collection of vectors $\{\vec{s}_{k,u,g}: k,u \geq 0\}$ will obey if it corresponds to an actual function $g$. 
We also define error terms which show how far a given collection of vectors $\{\vec{s}_{k,u,g}: k,u \geq 0\}$ is from 
corresponding to an actual function.
\begin{definition}
Let $P_k$ be the matrix with rows corresponding to the subsets $\{A \subseteq \Vmst, |A| = k\}$ and 
columns corresponding to the subsets $\{B \subseteq \Vmst, |B| = k+1\}$. Take $(P_k)_{AB} = 1$ if $A \subseteq B$ and $0$ otherwise.
\end{definition}
\begin{proposition}\label{errorvectorprop}
For any function $g$, $\vec{s}_{k,u,g} = \frac{1}{u}P_k\vec{s}_{k+1,u-1,g}$
\end{proposition}
\begin{proof} \noindent 
$(\vec{s}_{k,u,g})_A = s_{A,u}(g) = \sum_{C: |C| = |A| + u}{\hat{g}_C}$\\
$\frac{1}{u}P_k\vec{s}_{k+1,u-1,g} = \frac{1}{u}\sum_{B: |B| = |A|+1, A \subseteq B}{s_{B,u-1}(g)} = 
\frac{1}{u}\sum_{B,C: |C| = |A| + u, \atop |B| = |A| + 1, A \subseteq B \subseteq C}{\hat{g}_C} = 
\sum_{C: |C| = |A| + u}{\hat{g}_C}$
\end{proof}
\begin{definition}
For any collection of vectors $\{\vec{s}_{k,u,g}: k,u \geq 0\}$, define the error vectors\\
$\vec{e}_{k,u,g} = \vec{s}_{k,u,g} - \frac{1}{u}P_k\vec{s}_{k+1,u-1,g}$
\end{definition}
It is relatively easy to choose sets of vectors satisfying these equations. However, we also need to 
ensure that we get an $e$-invariant function for each $e \in E(G)$. $e$-invariance gives us another 
set of equations on the vectors $\{\vec{s}_{k,u,g}: k,u \geq 0\}$.
\begin{proposition} For all $v,w \in \Vmst$,
\begin{enumerate}
\item $(e_{\{\}} + e_{\{w\}})(C) = 2$ if $w \in R(C)$ and $0$ if $w \in L(C)$.
\item $(e_{\{\}} - e_{\{v\}})(C) = 2$ if $v \in L(C)$ and $0$ if $v \in R(C)$.
\item $((e_{\{\}} - e_{\{v\}})(e_{\{\}} + e_{\{w\}}))(C) = 4$ if $v \in L(C)$ and $w \in R(C)$ and $0$ otherwise.
\end{enumerate}
\end{proposition}
\begin{corollary}\label{einvariancecorollary} \ 
\begin{enumerate}
\item If $e = s \to w$ for some $w \in \Vmst$ then $g$ is $e$-invariant if and only if $(e_{\{\}} + e_{\{w\}})g = 0$.
Equivalently, $g$ is $e$-invariant if and only if $\hat{g}_{V \cup \{w\}} = -\hat{g}_V$ whenever 
$w \notin V$.
\item If $e = v \to t$ for some $v \in \Vmst$ then $g$ is $e$-invariant if and only if $(e_{\{\}} - e_{\{v\}})g = 0$.
Equivalently, $g$ is $e$-invariant if and only if $\hat{g}_{V \cup \{v\}} = \hat{g}_V$ whenever 
$v \notin V$.
\item  If $e = v \to w$ for some $v,w \in \Vmst$ then $g$ is $e$-invariant if and only if \\
$(e_{\{\}} - e_{\{v\}})(e_{\{\}} + e_{\{w\}})g = 0$.
Equivalently, $g$ is $e$-invariant if and only if \\
$\hat{g}_{V \cup \{v,w\}} = 
-\hat{g}_{V \cup \{v\}} + \hat{g}_{V \cup \{w\}} + \hat{g}_{V}$ whenever $v,w \notin V$.
\end{enumerate}
\end{corollary}
\begin{lemma}\label{einvariancelemma} If $g$ is an $e$-invariant function for some $e \in E(G)$, 
$A \subseteq \Vmst$, and $|A| = k$ then 
\begin{enumerate} 
\item If $e = s \to w$ for some $w \in \Vmst$ and $w \in A$ then for all $u$, 
$$(\vec{s}_{k,u,g})_A = -(\vec{s}_{k-1,u,g})_{A \setminus \{w\}} + (\vec{s}_{k,u-1,g})_A$$
\item If $e = v \to t$ for some $v \in \Vmst$ and $v \in A$ then for all $u$,
$$(\vec{s}_{k,u,g})_A = (\vec{s}_{k-1,u,g})_{A \setminus \{v\}} - (\vec{s}_{k,u-1,g})_A$$
\item If $e = v \to w$ for some $v,w \in \Vmst$ and $v,w \in A$ then for all $u$,
\begin{align*}
(\vec{s}_{k,u,g})_A &= (\vec{s}_{k-1,u,g})_{A \setminus \{v\}} - (\vec{s}_{k-1,u,g})_{A \setminus \{w\}} + 
(\vec{s}_{k-2,u,g})_{A \setminus \{v,w\}} \\
&- (\vec{s}_{k-1,u-1,g})_{A \setminus \{v\}} - (\vec{s}_{k-1,u-1,g})_{A \setminus \{w\}} + (\vec{s}_{k,u-2,g})_A
\end{align*}
\end{enumerate}
\end{lemma}
\begin{proof}
To show statement 1, note that 
\begin{align*}
(\vec{s}_{k-1,u,g})_{A \setminus \{w\}} &= \sum_{B: A \subseteq B,\atop |B| = |A| + u - 1}{\hat{g}_B} + 
\sum_{B: A \setminus \{w\} \subseteq B,\atop w \notin B, |B| = |A| + u - 1}{\hat{g}_B} \\
&= 
(\vec{s}_{k,u-1,g})_A - \sum_{B: A \setminus \{w\} \subseteq B,\atop w \notin B, |B| = |A| + u - 1}{\hat{g}_{B \cup \{w\}}} \\
&= 
(\vec{s}_{k,u-1,g})_A - \sum_{B: A \subseteq B, \atop |B| = |A| + u}{\hat{g}_{B}}  \\
&= (\vec{s}_{k,u-1,g})_A - (\vec{s}_{k,u,g})_A
\end{align*}
Rearranging now gives the desired statement. Similarly, to show statement 2, note that 
\begin{align*}
(\vec{s}_{k-1,u,g})_{A \setminus \{v\}} &= \sum_{B: A \subseteq B, \atop |B| = |A| + u - 1}{\hat{g}_B} + 
\sum_{B: A \setminus \{v\} \subseteq B, \atop v \notin B, |B| = |A| + u - 1}{\hat{g}_B} \\
&= (\vec{s}_{k,u-1,g})_A + \sum_{B: A \setminus \{v\} \subseteq B, \atop v \notin B, |B| = |A| + u - 1}{\hat{g}_{B \cup \{v\}}} \\
&= (\vec{s}_{k,u-1,g})_A + \sum_{B: A \subseteq B, \atop |B| = |A| + u}{\hat{g}_{B}} \\
&= (\vec{s}_{k,u-1,g})_A + (\vec{s}_{k,u,g})_A
\end{align*}
Rearranging now gives the desired statement. The proof for statement 3 is more complicated 
but uses similar ideas. In particular, note that 
$$(\vec{s}_{k-2,u,g})_{A \setminus \{v,w\}} = \sum_{B: A \subseteq B, \atop |B| = |A| + u - 2}{\hat{g}_B} + 
\sum_{B: A \setminus \{v\} \subseteq B, \atop v \notin B, |B| = |A| + u - 2}{\hat{g}_B} 
+ \sum_{B: A \setminus \{w\} \subseteq B, \atop w \notin B, |B| = |A| + u - 2}{\hat{g}_B} + 
\sum_{B: A \setminus \{v,w\} \subseteq B, \atop v,w \notin B, |B| = |A| + u - 2}{\hat{g}_B}$$
Now let's consider these terms one by one.
\begin{enumerate}
\item $\sum_{B: A \subseteq B, \atop |B| = |A| + u - 2}{\hat{g}_B} = (\vec{s}_{k,u-2,g})_A$
\item $\sum_{B: A \setminus \{v\} \subseteq B, \atop v \notin B, |B| = |A| + u - 2}{\hat{g}_B} = 
\sum_{B: A \setminus \{v\} \subseteq B, \atop |B| = |A| + u - 2}{\hat{g}_B} - 
\sum_{B: A \setminus \{v\} \subseteq B, \atop v \in B, |B| = |A| + u - 2}{\hat{g}_B} = 
(\vec{s}_{k-1,u-1,g})_{A \setminus \{v\}} - (\vec{s}_{k,u-2,g})_A$
\item $\sum_{B: A \setminus \{w\} \subseteq B, \atop w \notin B, |B| = |A| + u - 2}{\hat{g}_B} = 
\sum_{B: A \setminus \{w\} \subseteq B, \atop |B| = |A| + u - 2}{\hat{g}_B} - 
\sum_{B: A \setminus \{w\} \subseteq B, \atop w \in B, |B| = |A| + u - 2}{\hat{g}_B} 
= (\vec{s}_{k-1,u-1,g})_{A \setminus \{w\}} - (\vec{s}_{k,u-2,g})_A$
\item By statement 3 of Corollary \ref{einvariancecorollary},
\begin{align*}
\sum_{B: A \setminus \{v,w\} \subseteq B, \atop v,w \notin B, |B| = |A| + u - 2}{\hat{g}_B} &= 
\sum_{B: A \setminus \{v,w\} \subseteq B, \atop v,w \notin B, |B| = |A| + u - 2}
{(\hat{g}_{B \cup \{v,w\}} + \hat{g}_{B \cup \{v\}} - \hat{g}_{B \cup \{w\}})}\\
&= \sum_{B: A \subseteq B, \atop |B| = |A| + u}{\hat{g}_{B}} + 
\sum_{B: A \setminus \{w\} \subseteq B, \atop w \notin B, |B| = |A| + u - 1}{\hat{g}_{B}} - 
\sum_{B: A \setminus \{v\} \subseteq B, \atop v \notin B, |B| = |A| + u - 1}{\hat{g}_{B}}
\end{align*}
\end{enumerate}
Now we have that 
$(\vec{s}_{k-1,u,g})_{A \setminus \{v\}} - (\vec{s}_{k-1,u,g})_{A \setminus \{w\}} + 
(\vec{s}_{k-2,u,g})_{A \setminus \{v,w\}}$
\begin{enumerate}
\item $\sum_{B: A \subseteq B, \atop |B| = |A| + u}{\hat{g}_{B}} = (\vec{s}_{k,u,g})_A$
\item $\sum_{B: A \setminus \{w\} \subseteq B, \atop w \notin B, |B| = |A| + u - 1}{\hat{g}_{B}} = 
(\vec{s}_{k-1,u,g})_{A \setminus \{w\}} - \sum_{B: A \setminus \{w\} \subseteq B, \atop w \in B, |B| = |A| + u - 1}{\hat{g}_{B}} = 
(\vec{s}_{k-1,u,g})_{A \setminus \{w\}} - (\vec{s}_{k,u-1,g})_A$
\item $\sum_{B: A \setminus \{v\} \subseteq B, \atop v \notin B, |B| = |A| + u - 1}{\hat{g}_{B}} = 
(\vec{s}_{k-1,u,g})_{A \setminus \{v\}} - \sum_{B: A \setminus \{v\} \subseteq B, \atop v \in B, |B| = |A| + u - 1}{\hat{g}_{B}} = 
(\vec{s}_{k-1,u,g})_{A \setminus \{v\}} - (\vec{s}_{k,u-1,g})_A$
\end{enumerate}
Putting these three statements together, 
$$\sum_{B: A \setminus \{v,w\} \subseteq B, \atop v,w \notin B, |B| = |A| + u - 2}{\hat{g}_B} = 
(\vec{s}_{k,u,g})_A + (\vec{s}_{k-1,u,g})_{A \setminus \{w\}} - (\vec{s}_{k-1,u,g})_{A \setminus \{v\}}$$
Putting everything together, 
\begin{align*}
(\vec{s}_{k-2,u,g})_{A \setminus \{v,w\}} &= 
(\vec{s}_{k,u,g})_A + (\vec{s}_{k-1,u,g})_{A \setminus \{w\}} - (\vec{s}_{k-1,u,g})_{A \setminus \{v\}} \\
&+ (\vec{s}_{k-1,u-1,g})_{A \setminus \{v\}} + (\vec{s}_{k-1,u-1,g})_{A \setminus \{w\}} - (\vec{s}_{k,u-2,g})_A
\end{align*}
Rearranging now gives the desired result.
\end{proof}
For each possible edge $v \to w$, we define difference vectors which show far a collection of vectors 
$\{\vec{s}_{k,u,g}: k,u \geq 0\}$ is from representing a $(v \to w)$-invariant function.
\begin{definition} Given a collection of vectors $\{\vec{s}_{k,u,g}: k,u \geq 0\}$, define the vectors \\
$\{\vec{\Delta}_{k,u,g,v \to w}: k,u \geq 0, v,w \in V(G), v \neq w, v \neq t, w \neq s, v \to w \neq s \to t\}$ as follows
\begin{enumerate}
\item If $w \in A$ then  
$(\vec{\Delta}_{k,u,g,s \to w})_A = (\vec{s}_{k,u,g})_A - ((\vec{s}_{k,u-1,g})_A - (\vec{s}_{k-1,u,g})_{A \setminus \{w\}})$\\
Otherwise, $(\vec{\Delta}_{k,u,g,s \to w})_A = 0$
\item If $v \in A$ then 
$(\vec{\Delta}_{k,u,g,v \to t})_A = (\vec{s}_{k,u,g})_A - ((\vec{s}_{k-1,u,g})_{A \setminus \{v\}} - (\vec{s}_{k,u-1,g})_A)$\\
Otherwise, $(\vec{\Delta}_{k,u,g,v \to t})_A = 0$
\item If $v,w \in A$ then 
\begin{align*}
(\vec{\Delta}_{k,u,g,v \to w})_A &= (\vec{s}_{k,u,g})_A - \left((\vec{s}_{k-1,u,g})_{A \setminus \{v\}} - (\vec{s}_{k-1,u,g})_{A \setminus \{w\}} 
+ (\vec{s}_{k-2,u,g})_{A \setminus \{v,w\}} \right. \\
& \left. - (\vec{s}_{k-1,u-1,g})_{A \setminus \{v\}} - (\vec{s}_{k-1,u-1,g})_{A \setminus \{w\}} + 
(\vec{s}_{k,u-2,g})_{A}\right)
\end{align*}
Otherwise, $(\vec{\Delta}_{k,u,g,v \to w})_A = 0$
\end{enumerate}
\end{definition}
\begin{definition}
Call a possible edge $e = v \to w$ degenerate if $\{v,w\} \nsubseteq V(G)$, $v = w$, $v = t$, $w = s$, or 
$v \to w = s \to t$ and non-degenerate otherwise.
\end{definition}
\begin{proposition}
For any non-degenerate $e = v \to w$, 
the collection of vectors $\{\vec{s}_{k,u,g}: k,u \geq 0\}$ 
corresponds to an $e$-invariant function $g$ if and only if the error vectors $\vec{e}_{k,u,g}$ and the 
difference vectors $\vec{\Delta}_{k,u,g,e}$ are all $0$.
\end{proposition}
\begin{proof}
The only if statement follows from Proposition \ref{errorvectorprop} and Lemma \ref{einvariancelemma}. 
For the if statement, we simply take $\hat{g}_A = (\vec{s}_{|A|,0,g})_A$. The fact that 
$\vec{\Delta}_{|A|,0,g,e} = 0$ now corresponds precisely to the criteria for $e$-invariance in 
Corollary \ref{einvariancecorollary}. Using the fact that $\vec{e}_{k,u,g} = 0$ for all $k$ and $u$, it 
is easy to show by induction on $u$ that for all $k,u$ and all $A$ with $|A| = k$, 
$(\vec{s}_{k,u,g})_A = \sum_{B: A \subseteq B, |B| = |A| + u}{\hat{g}_B}$
\end{proof}
\begin{remark}
We do not care about $e$-invariance for degenerate $e$ for the following reasons. Obviously, we only care about vertices in $V(G)$. 
We do not allow $G$ to have loops and even if we did, they would not cross any cuts. Edges of the form $e = v \to s$ or $e = t \to w$ 
do not cross any cuts so every function is $e$-invariant for these $e$. $e = s \to t$ crosses every cut so only the $0$ function 
is $(s \to t)$-invariant. Also, if $s \to t \in E(G)$ then the directed connectivity problem on $G$ is trivial.
\end{remark}
\subsection{Checking error terms and well-definedness}\label{checkingeverythingworks}
In constructing our set of invariant functions $F_G = \{g_e: e \in E(G)\}$, we want the differences 
$g_{e_2} - g_{e_1}$ to be as small as possible. To do this, we will construct a base function 
$g$ and will have that $\hat{g_e}_V = \hat{g}_V$ whenever $|V| < z$ for some $z$. For each $e \in E(G)$ we will 
then choose the Fourier coefficients $\{\hat{g_e}_V: |V| = z\}$ so that $g_e$ is $e$-invariant.

This means that if $e = v \to w \in E(G)$ and we look at the Fourier coefficients 
$\{\hat{g}_V: |V \cup \{v,w\} \setminus \{s,t\}| < z\}$, the equations in subsection \ref{einvarianceconditions} for 
$e$-invariance must hold, so we must be very careful in constructing the collection of sum vectors 
$\{\vec{s}_{k,u,g}\}$ for $g$. We also need to be sure that the error vectors $\{\vec{e}_{k,u,g}\}$ are $0$. We can 
accomplish all of this as follows.
\begin{definition}
If we say that a non-degenerate edge $v \to w$ with $v,w \in A \cup \{s,t\}$ is relevant for a coordinate $(\vec{s}_{k,u,g})_A$ then 
we require that $(\vec{\Delta}_{k,u,g,v \to w})_A = 0$.
\end{definition}
\begin{remark}
It is possible that we could have $(\vec{\Delta}_{k,u,g,v \to w})_A = 0$ by coincidence, but we only say that 
$v \to w$ is relevant for a coordinate $(\vec{s}_{k,u,g})_A$ if we are intentionally making $(\vec{\Delta}_{k,u,g,v \to w})_A = 0$
\end{remark}
\begin{definition}
We say that a coordinate $(\vec{s}_{k,u,g})_A$ is fixed if there is some non-degenerate $v \to w$ with $v,w \in A$ which is relevant 
for $(\vec{s}_{k,u,g})_A$. Otherwise we say that $(\vec{s}_{k,u,g})_A$ is free.
\end{definition}
\begin{theorem}\label{constructionidea}
Given an acyclic input graph $G$ containing a path from $s$ to $t$ but no path from $s$ to $t$ of length at most 
$2^{z}$, if $v \to w$ is relevant for $(\vec{s}_{k,u,g})_A$ whenever $v,w \in A \cup \{s,t\}$, $v \to w$ is non-degenerate, 
and there is a path of length at most $2^{z-k-u-1}$ from $v$ to 
$w$ in $G$, then for any set of values $\{a_V: V \subseteq \Vmst, |V| < z\}$ there is a function $g$ such that 
$(\vec{s}_{|V|,0,g})_V = a_V$ whenever $(\vec{s}_{|V|,0,g})_V$ is a free coordinate.
\end{theorem}
We give two proofs of this theorem. The first proof is relatively short but requires good knowledge of the material 
in \cite{potechin}. The second proof is direct and more general but involves a lot of casework.
\begin{proof}[First proof of Theorem \ref{constructionidea}]
From Section 6 of \cite{potechin}, for all $z_2$ and any non-degenerate possible edge $e = v \to w$, 
if we partition the set of subsets of vertices $\{V: V \subseteq \Vmst, |V| < z_2\} \cup \{t\}$ so that 
\begin{enumerate}
\item $V$ and $V \cup \{w\}$ are in the same component if $v \in V \cup \{s\}$ and $w \notin V \cup \{t\}$
\item $V$ and $\{t\}$ are in the same component if $v \in V \cup \{s\}$ and $w = t$
\end{enumerate}
then if $g$ is a function such that 
$g \cdot K_V = g \cdot K_W$ whenever $V$ and $W$ are in the same connected component (where $K_{\{t\}} = K_{t'}$), 
there is an $e$-invariant function $g_{e}$ such that $\hat{g_{e}}_V = \hat{g}_V$ whenever $V \subseteq \Vmst$ and 
$|V| < z_2$. This implies that $\vec{\Delta}_{k,u,g,v \to w} = 0$ whenever $k + u < z_2$.

Now partition the set of subsets of vertices $\{V: V \subseteq \Vmst, |V| < z\} \cup \{t\}$ as follows. 
\begin{enumerate}
\item If $V \subseteq \Vmst$, $w \in V$, $v \in V \setminus \{w\} \cup \{s\}$, 
and there is a path from $s$ to $w$ of length at most 
$2^{z-1-|V|}$ in $G$ then put $V$ in the same component as $V \setminus \{w\}$
\item If $V \subseteq \Vmst$, $v \in V$ and there is a path from $v$ to $t$ of length at most 
$2^{z-1-|V|}$ in $G$ then put $V$ in the same component as $\{t\}$
\end{enumerate}
From the above, taking $z_2 = z - \lceil{\lg{l}}\rceil$, where $l$ is the length of the path from $v$ to $w$ in $G$, 
if $g$ is a function such that $g \cdot K_V = g \cdot K_W$ whenever $V$ and $W$ are in the 
same connected component then the relevance condition of Theorem \ref{constructionidea} is satisfied. We 
just need to make sure that we can freely choose the free coordinates of each vector $\vec{s}_{k,0,g}$.

\begin{definition}
Call a set of vertices $V$ a representative of its connected component if $V = \{t\}$ or 
$(\vec{s}_{|V|,0,g})_V$ is free.
\end{definition}
\begin{lemma}
Every connected component has exactly one representative.
\end{lemma}
\begin{proof} 
It is sufficient to show that starting at any set of vertices $V \subseteq \Vmst$ no matter which way we choose to 
reduce $V$, we will always end up at the same representative for the connected component. We prove this by showing 
that if there are two possible ways to reduce a set of vertices $V \subseteq \Vmst$, we can get to the same 
representative no matter which one we choose. The result then follows by induction.

If we can use either $v_1 \to v_2$ or $v_3 \to v_4$ to reduce $V$, we have the following cases:
\begin{enumerate}
\item If $v_2,v_4$ are distinct vertices in $V$ and $v_1,v_3 \in V \setminus \{v_2,v_4\} \cup \{s\}$ then using 
$v_1 \to v_2$ and $v_3 \to v_4$ in either order 
we will reduce $V$ to $V \setminus \{v_2,v_4\}$.
\item If $v_2,v_4$ are distinct vertices in $V$, $v_3 = v_2$ and $v_1 \in V \setminus \{v_2,v_4\}$ then if we reduce $V$ with 
$v_3 \to v_4$ first and then with $v_1 \to v_2$ we will obtain $V \setminus \{v_2,v_4\}$. If we reduce $V$ with 
$v_1 \to v_2$ first then $v_3 \notin V \setminus \{v_2\}$. However, there was a path from $v_1$ to $v_2$ of length at most 
$2^{z-1-|V|}$ in $G$ and a path from $v_3$ to $v_4$ of length at most 
$2^{z-1-|V|}$ in $G$ so there is a path from $v_1$ to $v_4$ of length at most $2^{z-1-|V \setminus \{v_2\}|}$ in $G$ so we may now 
use $v_1 \to v_4$ to reduce $V \setminus \{v_2\}$ and obtain $V \setminus \{v_2,v_4\}$.
\item If $v_2 = v_4$ is a vertex in $V$ and $v_1,v_3 \in V \cup \{s\}$ then whether we use $v_1 \to v_2$ or 
$v_3 \to v_4$ we will reduce $V$ to $V \setminus \{v_4\}$
\item If $v_2 = v_4 = t$ and $v_1,v_3 \in V \cup \{s\}$ then whether we use $v_1 \to v_2$ or 
$v_3 \to v_4$ we will reduce $V$ to $\{t\}$
\item If $v_4 = t$ and $v_2,v_3$ are distinct vertices of $V$ then using 
$v_1 \to v_2$ and $v_3 \to v_4$ in either order we will reduce $V$ to $\{t\}$
\item If $v_4 = t$, $v_3 = v_2$ is a vertex in $V$, and $v_1 \in V \setminus \{v_2\}$ then if we reduce $V$ with 
$v_3 \to v_4$ first we will obtain $\{t\}$. If we reduce $V$ with 
$v_1 \to v_2$ first then $v_3 \notin V \setminus \{v_2\}$. However, there was a path from $v_1$ to $v_2$ of length at most 
$2^{z-1-|V|}$ in $G$ and a path from $v_3$ to $t$ of length at most 
$2^{z-1-|V|}$ in $G$ so there is a path from $v_1$ to $t$ of length at most $2^{z-1-|V \setminus \{v_2\}|}$ in $G$ so we may now 
use $v_1 \to t$ to reduce $V \setminus \{v_2\}$ and obtain $\{t\}$.
\end{enumerate}
\end{proof}
Now from Section 6 of \cite{potechin} we may choose a function $g$ with arbitrary values of 
$(\vec{s}_{|V|,0,g})_V = \hat{g}_V$ for all $V \neq \{t\}$ which are the representative of a connected component. 
Equivalently, we may freely choose the values of all free coordinates $(\vec{s}_{|V|,0,g})_V$ whenever $|V| < z$.
We may further take $\hat{g}_V = 0$ whenever $|V| \geq z$ and this completes the proof.
\end{proof}
\begin{proof}[Second proof of Theorem \ref{constructionidea}]
\begin{lemma}\label{checkingwelldefinednesslemma}
Let $G$ be an acyclic input graph containing a path from $s$ to $t$. We may freely choose which non-degenerate edges 
$v \to w$ are relevant for the terms $(\vec{s}_{k,u,g})_A$ so long as the following conditions hold.
\begin{enumerate}
\item If $v \to w$ is relevant for $(\vec{s}_{k,u,g})_A$ then $v \to w$ is relevant for $(\vec{s}_{k_2,u_2,g})_{A_2}$ 
whenever $\{v,w\} \setminus \{s,t\} \subseteq A_2 \subseteq A$, $k_2 = |A_2|$, and $u_2 \leq u$.
\item If $u \to v$ and $v \to w$ are relevant for $(\vec{s}_{k,u,g})_A$ then $u \to w$ is relevant for 
$(\vec{s}_{k-1,u,g})_{A \setminus v}$.
\end{enumerate}
\end{lemma}
\begin{proof}
The proof of this lemma is long and involves a lot of casework, so we put it in Appendix \ref{welldefinedproof}.
\end{proof}
\begin{lemma}\label{errorterms} \ 
\begin{enumerate}
\item If $w \in A$ then 
$$(\vec{e}_{k,u,g})_A = (\vec{\Delta}_{k,u,g,s \to w})_A - 
\frac{1}{u}\sum_{b \notin A}{(\vec{\Delta}_{k+1,u-1,g,s \to w})_{A \cup \{b\}}} 
+ \frac{u-1}{u}(\vec{e}_{k,u-1,g})_A - (\vec{e}_{k-1,u,g})_{A \setminus \{w\}}$$
\item If $v \in A$ then 
$$(\vec{e}_{k,u,g})_A = (\vec{\Delta}_{k,u,g,v \to t})_A - 
\frac{1}{u}\sum_{b \notin A}{(\vec{\Delta}_{k+1,u-1,g,v \to t})_{A \cup \{b\}}} 
- \frac{u-1}{u}(\vec{e}_{k,u-1,g})_A - (\vec{e}_{k-1,u,g})_{A \setminus \{v\}}$$
\item If $v,w \in A$ then 
\begin{align*}
(\vec{e}_{k,u,g})_A &= (\vec{\Delta}_{k,u,g,v \to w})_A - 
\frac{1}{u}\sum_{b \notin A}{(\vec{\Delta}_{k+1,u-1,g,v \to w})_{A \cup \{b\}}} 
+ (\vec{e}_{k-1,u,g})_{A \setminus \{v\}} - (\vec{e}_{k-1,u,g})_{A \setminus \{w\}} \\
&+ (\vec{e}_{k-2,u,g})_{A \setminus \{v,w\}} - \frac{u-1}{u}(\vec{e}_{k-1,u-1,g})_{A \setminus \{v\}} - 
\frac{u-1}{u}(\vec{e}_{k-1,u-1,g})_{A \setminus \{w\}} + \frac{u-2}{u}(\vec{e}_{k,u-2,g})_A
\end{align*}
\end{enumerate}
\end{lemma}
\begin{proof} \ 
\begin{enumerate}
\item If $w \in A$ then 
\begin{align*}
(\vec{e}_{k,u,g})_A &= (\vec{s}_{k,u,g})_A - \frac{1}{u}\sum_{b \notin A}(\vec{s}_{k+1,u-1,g})_{A \cup \{b\}} \\
&= (\vec{\Delta}_{k,u,g,s \to w})_A - 
\frac{1}{u}\sum_{b \notin A}{(\vec{\Delta}_{k+1,u-1,g,s \to w})_{A \cup \{b\}}} \\
&+ \left((\vec{s}_{k,u-1,g})_A - (\vec{s}_{k-1,u,g})_{A \setminus \{w\}}\right) - 
\frac{1}{u}\sum_{b \notin A}{\left((\vec{s}_{k+1,u-2,g})_{A \cup \{b\}} - (\vec{s}_{k,u-1,g})_{A \cup \{b\} \setminus \{w\}}\right)} \\
&= (\vec{\Delta}_{k,u,g,s \to w})_A - 
\frac{1}{u}\sum_{b \notin A}{(\vec{\Delta}_{k+1,u-1,g,s \to w})_{A \cup \{b\}}} \\
&+ \frac{u-1}{u}\left((\vec{s}_{k,u-1,g})_A - \frac{1}{u-1}\sum_{b \notin A}{(\vec{s}_{k+1,u-2,g})_{A \cup \{b\}}}\right) 
+ \frac{1}{u}(\vec{s}_{k,u-1,g})_A \\
&- \left((\vec{s}_{k-1,u,g})_{A \setminus \{w\}} - 
\frac{1}{u}\sum_{b \notin A \setminus \{w\}}{(\vec{s}_{k,u-1,g})_{(A \setminus \{w\}) \cup \{b\}}}\right) 
- \frac{1}{u}(\vec{s}_{k,u-1,g})_A \\
&= (\vec{\Delta}_{k,u,g,s \to w})_A - 
\frac{1}{u}\sum_{b \notin A}{(\vec{\Delta}_{k+1,u-1,g,s \to w})_{A \cup \{b\}}} 
+ \frac{u-1}{u}(\vec{e}_{k,u-1,g})_A - (\vec{e}_{k-1,u,g})_{A \setminus \{w\}}
\end{align*}
\item If $v \in A$ then 
\begin{align*}
(\vec{e}_{k,u,g})_A &= (\vec{s}_{k,u,g})_A - \frac{1}{u}\sum_{b \notin A}(\vec{s}_{k+1,u-1,g})_{A \cup \{b\}} \\
&= (\vec{\Delta}_{k,u,g,v \to t})_A - 
\frac{1}{u}\sum_{b \notin A}{(\vec{\Delta}_{k+1,u-1,g,v \to t})_{A \cup \{b\}}} \\
&+ \left(-(\vec{s}_{k,u-1,g})_A + (\vec{s}_{k-1,u,g})_{A \setminus \{v\}}\right) - 
\frac{1}{u}\sum_{b \notin A}{\left(-(\vec{s}_{k+1,u-2,g})_{A \cup \{b\}} + (\vec{s}_{k,u-1,g})_{A \cup \{b\} \setminus \{v\}}\right)} \\
&= (\vec{\Delta}_{k,u,g,v \to t})_A - 
\frac{1}{u}\sum_{b \notin A}{(\vec{\Delta}_{k+1,u-1,g,v \to t})_{A \cup \{b\}}} \\
&- \frac{u-1}{u}\left((\vec{s}_{k,u-1,g})_A - \frac{1}{u-1}\sum_{b \notin A}{(\vec{s}_{k+1,u-2,g})_{A \cup \{b\}}}\right) 
- \frac{1}{u}(\vec{s}_{k,u-1,g})_A \\
&+ \left((\vec{s}_{k-1,u,g})_{A \setminus \{v\}} - 
\frac{1}{u}\sum_{b \notin A \setminus \{v\}}{(\vec{s}_{k,u-1,g})_{(A \setminus \{v\}) \cup \{b\}}}\right) 
+ \frac{1}{u}(\vec{s}_{k,u-1,g})_A \\
&= (\vec{\Delta}_{k,u,g,v \to t})_A - 
\frac{1}{u}\sum_{b \notin A}{(\vec{\Delta}_{k+1,u-1,g,v \to t})_{A \cup \{b\}}} 
- \frac{u-1}{u}(\vec{e}_{k,u-1,g})_A + (\vec{e}_{k-1,u,g})_{A \setminus \{v\}}
\end{align*}
\item If $v,w \in A$ then 
\begin{align*}
(\vec{e}_{k,u,g})_A &= (\vec{\Delta}_{k,u,g,v \to w})_A - 
\frac{1}{u}\sum_{b \notin A}{(\vec{\Delta}_{k+1,u-1,g,v \to w})_{A \cup \{b\}}} \\
&+ (\vec{s}_{k-1,u,g})_{A \setminus \{v\}} - (\vec{s}_{k-1,u,g})_{A \setminus \{w\}} 
+ (\vec{s}_{k-2,u,g})_{A \setminus \{v,w\}} \\
&- (\vec{s}_{k-1,u-1,g})_{A \setminus \{v\}} - (\vec{s}_{k-1,u-1,g})_{A \setminus \{w\}} + 
(\vec{s}_{k,u-2,g})_{A}\\
&- \frac{1}{u}\sum_{b \notin A}
\left((\vec{s}_{k,u-1,g})_{A \cup \{b\} \setminus \{v\}} - (\vec{s}_{k,u-1,g})_{A \cup \{b\} \setminus \{w\}} 
+ (\vec{s}_{k-1,u-1,g})_{A \cup \{b\} \setminus \{v,w\}} \right. \\
& \left. - (\vec{s}_{k,u-2,g})_{A \cup \{b\} \setminus \{v\}} - (\vec{s}_{k,u-2,g})_{A \cup \{b\} \setminus \{w\}} + 
(\vec{s}_{k+1,u-3,g})_{A \cup \{b\}} \right) \\
&= (\vec{\Delta}_{k,u,g,v \to w})_A - 
\frac{1}{u}\sum_{b \notin A}{(\vec{\Delta}_{k+1,u-1,g,v \to w})_{A \cup \{b\}}} \\
&+ \left((\vec{s}_{k-1,u,g})_{A \setminus \{v\}} - 
\frac{1}{u}\sum_{b \notin A \setminus \{v\}}{(\vec{s}_{k,u-1,g})_{(A \setminus \{v\}) \cup \{b\}}}\right) 
+ \frac{1}{u}(\vec{s}_{k,u-1,g})_{A}\\
&- \left((\vec{s}_{k-1,u,g})_{A \setminus \{w\}} - 
\frac{1}{u}\sum_{b \notin A \setminus \{w\}}{(\vec{s}_{k,u-1,g})_{(A \setminus \{w\}) \cup \{b\}}}\right) 
- \frac{1}{u}(\vec{s}_{k,u-1,g})_{A}\\
&+ \left((\vec{s}_{k-2,u,g})_{A \setminus \{v,w\}} - 
\frac{1}{u}\sum_{b \notin A \setminus \{v,w\}}{(\vec{s}_{k-1,u-1,g})_{(A \setminus \{v,w\}) \cup \{b\}}}\right) \\
&+ \frac{1}{u}(\vec{s}_{k-1,u-1,g})_{A \setminus \{v\}} + \frac{1}{u}(\vec{s}_{k-1,u-1,g})_{A \setminus \{w\}} \\
&- \frac{u-1}{u}\left((\vec{s}_{k-1,u-1,g})_{A \setminus \{v\}} - 
\frac{1}{u-1}\sum_{b \notin A \setminus \{v\}}{(\vec{s}_{k,u-2,g})_{(A \setminus \{v\}) \cup \{b\}}} \right) \\
&- \frac{1}{u}(\vec{s}_{k-1,u-1,g})_{A \setminus \{v\}} - \frac{1}{u}(\vec{s}_{k,u-2,g})_{A}\\
&- \frac{u-1}{u}\left((\vec{s}_{k-1,u-1,g})_{A \setminus \{w\}} - 
\frac{1}{u-1}\sum_{b \notin A \setminus \{w\}}{(\vec{s}_{k,u-2,g})_{(A \setminus \{w\}) \cup \{b\}}} \right) \\
&- \frac{1}{u}(\vec{s}_{k-1,u-1,g})_{A \setminus \{w\}} - \frac{1}{u}(\vec{s}_{k,u-2,g})_{A}\\
&+ \frac{u-2}{u}\left((\vec{s}_{k,u-2,g})_{A} - \frac{1}{u-2}\sum_{b \notin A}{(\vec{s}_{k+1,u-3,g})_{A \cup \{b\}}}\right) 
+ \frac{2}{u}(\vec{s}_{k,u-2,g})_{A}\\
&= (\vec{\Delta}_{k,u,g,v \to w})_A - 
\frac{1}{u}\sum_{b \notin A}{(\vec{\Delta}_{k+1,u-1,g,v \to w})_{A \cup \{b\}}} 
+ (\vec{e}_{k-1,u,g})_{A \setminus \{v\}} - (\vec{e}_{k-1,u,g})_{A \setminus \{w\}} \\
&+ (\vec{e}_{k-2,u,g})_{A \setminus \{v,w\}} - \frac{u-1}{u}(\vec{e}_{k-1,u-1,g})_{A \setminus \{v\}} - 
\frac{u-1}{u}(\vec{e}_{k-1,u-1,g})_{A \setminus \{w\}} + \frac{u-2}{u}(\vec{e}_{k,u-2,g})_A
\end{align*}
\end{enumerate}
\end{proof}
Theorem \ref{constructionidea} now follows easily. First arbitrarily choose all values of $(\vec{s}_{|V|,0,g})_V = \hat{g}_V$ where 
$|V| < z$ and $(\vec{s}_{|V|,0,g})_V$ is a free coordinate and determine the other 
values based on the definition of relevance. Since the conditions of Lemma \ref{checkingwelldefinednesslemma} are satisfied 
this is well-defined. Then take \\
$(\vec{s}_{|V|,u,g})_V = \frac{1}{u}\sum_{w \in \Vmst \setminus V}{(\vec{s}_{|V|+1,u-1,g})_V}$ whenever
$(\vec{s}_{|V|,u,g})_V$ is a free coordinate. Using induction and Lemma \ref{errorterms}, all 
error vectors must be $0$, as needed.
\end{proof}
\begin{corollary}\label{constructionideacorollary}
Given an acyclic input graph $G$ containing a path from $s$ to $t$ but no path from $s$ to $t$ of length at most 
$2^{z}$, if $v \to w$ is relevant for $(\vec{s}_{k,u,g})_A$ whenever $v,w \in A \cup \{s,t\}$, $v \to w$ is non-degenerate, 
and there is a path of length at most $2^{z-k-u-1}$ from $v$ to 
$w$ in $G$, then if we choose the sum vectors $\{\vec{s}_{k,u,g}\}$ in increasing lexicographic order of 
$(k+u,k)$ subject to the constraints that 
\begin{enumerate}
\item For every fixed coordinate $(\vec{s}_{k,u,g})_V$ the corresponding equation for relevance holds.
\item Whenever $k > 0$ and $(\vec{s}_{k-1,u+1,g})_V$ is a free coordinate we choose the coordinates \\
$\{(\vec{s}_{k,u,g})_{W}: V \subseteq W, (\vec{s}_{k,u,g})_{W}$ is a free coordinate$\}$ so that \\
$(\vec{s}_{k-1,u+1,g})_V = \frac{1}{u+1}\sum_{w \in \Vmst \setminus V}{(\vec{s}_{k,u,g})_{V \cup \{w\}}}$
\end{enumerate}
then these sum vectors $\{\vec{s}_{k,u,g}\}$ will correspond to an actual function $g$.
\end{corollary}
\begin{proof}
We use Theorem \ref{constructionidea} and choose the values $\{a_V: V \subseteq \Vmst, |V| < z\}$ in increasing 
order of $|V|$. Assume that we have already chosen the values $\{a_V: V \subseteq \Vmst, |V| < z_2\}$ for some 
$z_2$. From the equations for relevance, this determines all fixed coordinates 
$(\vec{s}_{k,u,g})_V$ where $k+u \leq z_2$. Now if we choose the vectors $\{\vec{s}_{k,z_2-k,g}\}$ in increasing order of 
$k$ so that whenever $k > 0$ and $(\vec{s}_{k-1,z_2-k+1,g})_V$ is a free coordinate the free coordinates 
$\{(\vec{s}_{k,z_2-k,g})_{W}: V \subseteq W, (\vec{s}_{k,z_2-k,g})_{W}$ is a free coordinate$\}$ are chosen so that 
$(\vec{s}_{k-1,z_2-k+1,g})_V = \frac{1}{u+1}\sum_{w \in \Vmst \setminus V}{(\vec{s}_{k,z_2-k,g})_{V \cup \{w\}}}$
then once we reach $k = z_2$ we can take $a_V = (\vec{s}_{z_2,0,g})_{V}$ for all free coordinates $(\vec{s}_{z_2,0,g})_{V}$ and 
arbitrarily choose $a_V$ whenever $(\vec{s}_{z_2,0,g})_{V}$ is a fixed coordinate. The base case $z_2 = 0$ is trivial, 
so by induction we can choose all values $\{a_V: V \subseteq \Vmst, |V| < z$ in this way. 
Theorem \ref{constructionidea} now gives us an actual function $g$ which must have the correct sum vectors beacuse of Proposition \ref{errorvectorprop} and 
the equations for relevance.
\end{proof}
\subsection{Building up a base function}\label{buildinguptools}
Corollary \ref{constructionideacorollary} allows us to construct a base function $g$ while having considerable control over
all of the sum vectors $\{\vec{s}_{k,u,g}\}$. This is extremely useful because we will try to make
$||\vec{s}_{k,u,g}||$ small for every $k$ and $u$. However, in order to use Corollary \ref{constructionidea} 
we need a way to choose each vector $\{\vec{s}_{k+1,u-1,g}\}$ where $u > 0$ so that whenever 
$(\vec{s}_{k,u,g})_V$ is a free coordinate the free coordinates \\
$\{(\vec{s}_{k+1,u-1,g})_{W}: V \subseteq W, (\vec{s}_{k+1,u-1,g})_{W}$ is a free coordinate$\}$ are chosen so that \\
$(\vec{s}_{k,u,g})_V = \frac{1}{u}\sum_{w \in \Vmst \setminus V}{(\vec{s}_{k+1,u-1,g})_{V \cup \{w\}}}$\\
In this subsection, we show how to do this. For this subsection, we are always working with collection of sum vectors 
$\{\vec{s}_{k,u,g}\}$ where we have already decided which coordinates of these vectors are fixed and which coordinates of these 
vectors are free.
\begin{definition}\label{freeandfixeddefinition}
For each $k,u$, 
\begin{enumerate}
\item Define $\pi_{k,u,fixed}$ to be the projection which projects any vector in the same vector space as 
$\vec{s}_{k,u,g}$ onto the fixed coordinates of $\vec{s}_{k,u,g}$. 
\item Define $\pi_{k,u,free}$ to be the projection which projects any vector in the same vector space as 
$\vec{s}_{k,u,g}$ onto the free coordinates of $\vec{s}_{k,u,g}$. 
\item Define $\vec{s}_{k,u,g,fixed} = \pi_{k,u,fixed}(\vec{s}_{k,u,g})$
\item Define $\vec{s}_{k,u,g,free} = \pi_{k,u,free}(\vec{s}_{k,u,g})$
\end{enumerate}
\end{definition}
\begin{proposition}
$\vec{s}_{k,u,g} = \vec{s}_{k,u,g,fixed} + \vec{s}_{k,u,g,free}$
\end{proposition}
\begin{proposition}\label{stepupprop}
Condition 2 of Crolloary \ref{constructionideacorollary} holds if and only if \\
$\vec{s}_{k,u,g,free} = \pi_{k,u,free}(\frac{1}{u}P_{k}\vec{s}_{k+1,u-1,g})$ whenever $u > 0$
\end{proposition}
\begin{definition}
Let $P_{k,u,free}$ be the matrix $P_k$ except that we 
make all of the rows corresponding to a set of vertices $A$ such that $(\vec{s}_{k,u,g})_A$ is fixed $0$ and 
we make all of the columns corresponding to a set of vertices $B$ such that $(\vec{s}_{k+1,u-1,g})_B$ is fixed $0$.
\end{definition}
\begin{definition}
When $u > 0$ define the adjustment vector $\vec{a}_{k,u,g} = -\pi_{k,u,free}(P_k\vec{s}_{k+1,u-1,g,fixed})$
\end{definition}
\begin{proposition}\label{adjustmentvectorprop}
$\pi_{k,u,free}(P_k\vec{s}_{k+1,u-1,g}) = -\vec{a}_{k,u,g} + P_{k,u,free}\vec{s}_{k+1,u-1,g,free}$
\end{proposition}
\begin{corollary}\label{stepupcorollary}
Condition 2 of Corollary \ref{constructionideacorollary} holds if and only if \\
$P_{k,u,free}\vec{s}_{k+1,u-1,g,free} = \vec{a}_{k,u,g} + u\vec{s}_{k,u,g,free}$ whenever $u > 0$
\end{corollary}
\begin{proof}
This follows immediately from Proposition \ref{adjustmentvectorprop} and Proposition \ref{stepupprop}.
\end{proof}
If we let $P$ be the matrix $P_{k,u,free}$ after we delete all the rows and columns corresponding to fixed coordinates 
(which were all $0$ rows and columns by definition), let $\vec{x}$ be the vector $\vec{a}_{k,u,g} + u\vec{s}_{k,u,g,free}$ 
after we delete all fixed coordinates of $\vec{s}_{k,u,g}$ (which were also all $0$) and let $\vec{y}$ be the vector 
$\vec{s}_{k+1,u-1,g,free}$ after deleting all fixed coordinates of $\vec{s}_{k+1,u-1,g}$ (which were also all $0$), then 
we now have the matrix equation $P\vec{y} = \vec{x}$. If $P$ does not have full row rank then it may be impossible 
to find a $\vec{y}$ such that $P\vec{y} = \vec{x}$. We will take care to avoid this case by showing that ${P}P^T$ has no 
zero eigenvalues. If $P$ does have full row rank then we want to find the $\vec{y}$ with smallest norm such that $P\vec{y} = \vec{x}$ 
and then find a bound on 
$||\vec{s}_{k+1,u-1,g,free}||^2 = ||\vec{y}||^2$ in terms of $||\vec{a}_{k,u,g} + u\vec{s}_{k,u,g,free}||^2 = ||\vec{x}||^2$.
The best bound we can get is $||\vec{y}||^2 \leq 
\max_{\vec{x} \neq 0}\{\min_{\vec{y}: P\vec{y} = \vec{x}}{\{\frac{||\vec{y}||^2}{||\vec{x}||^2}\}}\}||\vec{x}||^2$

We have now reduced our problem to a problem of the following form. 
For a given real matrix $P$ with full row rank, obtain bounds on 
$\max_{\vec{x} \neq 0}\{\min_{\vec{y}: P\vec{y} = \vec{x}}{\{\frac{||\vec{y}||^2}{||\vec{x}||^2}\}}\}$
\begin{lemma}
If $P$ is a real matrix with full row rank then letting $\{\lambda_i\}$ be the eigenvalues of $P^{T}P$,
$$\max_{\vec{x} \neq 0}\{\min_{\vec{y}: P\vec{y} = \vec{x}}{\frac{||\vec{y}||^2}{||\vec{x}||^2}}\}\} = 
\max_{i: \lambda_i \neq 0}{\{\frac{1}{\lambda_i}\}}$$
\end{lemma}
\begin{proof}
First note that the null space of $P^T$ is trivial, so 
$P\vec{y} = \vec{x}$ if and only if ${P^T}P\vec{y} = {P^T}\vec{x}$. Further note that $P^{T}P$ is a 
positive semidefinite real symmetric matrix, so if $j$ is the number of columns of $P$ then 
there is an orthonormal basis of $R^{j}$ consisting of eigenvectors 
$\vec{e}_1, \cdots, \vec{e}_j$ of $P^{T}P$ with eigenvalues $\lambda_1, \cdots, \lambda_j$.
Now for any $\vec{x}$, we can write $P^{T}\vec{x} = \sum_{i = 1}^{j}{{c_i}\vec{e}_i}$.
\begin{proposition}\label{zeroeigenvalues}
If $P^{T}\vec{x} = \sum_{i = 1}^{j}{{c_i}\vec{e}_i}$ then $c_i = 0$ for every zero eigenvector $\vec{e}_i$ of ${P^T}P$
\end{proposition}
\begin{proof}
Note that if ${P^T}P\vec{e}_i = 0$ then 
$\vec{e}^T_i{P^T}P\vec{e}_i = P\vec{e}_i \cdot P\vec{e}_i = 0$ so $P\vec{e}_i = 0$. But then 
${\vec{e}^T_i}P^T\vec{x} = c_i = {\vec{x}^T}P\vec{e}_i = 0$
\end{proof}  
Proposition \ref{zeroeigenvalues} implies that for any $\vec{x}$ there is a $\vec{y}$ such that 
${P^T}P\vec{y} = {P^T}\vec{x}$. In particular, the solution $\vec{y}$ with the smallest norm is 
$\vec{y} = \sum_{i: \lambda_i \neq 0}{\frac{c_i}{\lambda_i}\vec{e}_i}$.
Now $|\vec{y}|^2 = \sum_{i: \lambda_i \neq 0}{(\frac{c_i}{\lambda_i})^2}$ and 
$|\vec{x}|^2 = \vec{y}^T{P^T}P\vec{y} = \sum_{i: \lambda_i \neq 0}{\frac{(c_i)^2}{\lambda_i}}$

Comparing term by term, we have that $\sum_{i: \lambda_i \neq 0}{(\frac{c_i}{\lambda_i})^2} \leq 
(\max_{i: \lambda_i \neq 0}{\{\frac{1}{\lambda_i}\}})\sum_{i = 1}^{j}{\frac{(c_i)^2}{\lambda_i}}$ and this 
inequality is an equality when $c_i$ is nonzero for this particular $i$ and $0$ for all other $i$, so \\
$\max_{\vec{x} \neq 0}\{\min_{\vec{y}: P\vec{y} = \vec{x}}{\frac{||\vec{y}||^2}{||\vec{x}||^2}}\}\} = 
\max_{i: \lambda_i \neq 0}{\{\frac{1}{\lambda_i}\}}$
\end{proof}
Thus, we are interested in the smallest nonzero eigenvalue of $P^{T}P$. Since $P$ is real with 
full row rank, the eigenvalues of $P{P^T}$ are precisely the nonzero eigenvalues of $P^{T}P$, 
so we can find the smallest nonzero eigenvalue of $P^{T}P$ by finding the smallest 
eigenvalue of $P{P^T}$. 

We can give an elegant exact answer when $P = P_k$ and can give bounds for many other $P$.
\begin{theorem}
If $k < \frac{n}{2}$ then ${P_k}{P^T_k}$ has eigenvalues $(n-k-i)(k+1-i)$ with multiplicity 
${n \choose {i}} - {n \choose {i-1}}$ for $i \in [0,k]$
\end{theorem}
\begin{proof}
Given $V \subseteq \Vmst$ where $|V| \leq k$, let $\vec{x}_V$ be the vector such that 
for any $A \subseteq V$, if $B \subseteq \Vmst$, $|B| = k$, and $B \cap V = A$ then 
$(\vec{x}_V)_B = \frac{(-1)^{|A|}}{|C \subseteq \Vmst: |C| = k, C \cap V = A|}$.

In other words, we make the coordinates $(\vec{x}_V)_B$ where $B \cap V = A$ have total weight $(-1)^{|A|}$ where the weight 
is spread evenly among these coordinates. Now consider where the weight is in $P_k{P^T_k}\vec{x}_V$. 
Letting $m = |V|$, the coordinates $(\vec{x}_V)_B$ where $|B \cap V| = j$ have total weight $(-1)^j{m \choose j}$. 
Multiplying by $P_kP^T_k$ multiplies and shifts this weight as follows. For a given $B$ with $|B \cap V| = j$ there 
are $(n - k - m + j)j$ ways to add a vertex to $B$ then remove a vertex and obtain a set of vertices $C$ such that 
$|C \cap V| = j-1$. There are $(m-j)(k-j)$ ways to add a vertex to $B$ then remove a vertex and obtain a set of vertices 
$C$ such that $|C \cap V| = j+1$. Finally, since there are $(n-k)(k+1)$ total ways to add a vertex to $B$ and then 
remove a vertex from $B$ there must be $(n-k)(k+1) - (n - k - m + j)j - (m-j)(k-j)$ ways to add a vertex to $B$ then remove 
a vertex and obtain a set of vertices $C$ such that $|C \cap V| = j$. Thus, from the original total weight of 
$(-1)^j{m \choose j}$ on the coordinates $(\vec{x}_V)_B$ where $|B \cap V| = j$, we get 
\begin{enumerate}
\item A total weight of $(-1)^j{m \choose j}(n - k - m + j)j$ on the coordinates $(P_k{P^T_k}\vec{x}_V)_C$ where $|C \cap V| = j-1$
\item A total weight of $(-1)^j{m \choose j}(m-j)(k-j)$ on the coordinates $(P_k{P^T_k}\vec{x}_V)_C$ where $|C \cap V| = j+1$
\item A total weight of 
$(-1)^j{m \choose j}((n-k)(k+1) - (n - k - m + j)j - (m-j)(k-j))$ on the coordinates $(P_k{P^T_k}\vec{x}_V)_C$ where $|C \cap V| = j$
\end{enumerate}
Turing this around, in $P_k{P^T_k}\vec{x}_V$ we have the following contributions to the total weight of the coordinates 
$(P_k{P^T_k}\vec{x}_V)_C$ where $|C \cap V| = j$
\begin{enumerate}
\item A contribution of $(-1)^{j+1}{m \choose {j+1}}(n - k - m + j+1)(j+1)$ from the coordinates $(\vec{x}_V)_B$ where $|B \cap V| = j+1$
\item A contribution of $(-1)^{j-1}{m \choose {j-1}}(m-j+1)(k-j+1)$ from the coordinates $(\vec{x}_V)_B$ where $|B \cap V| = j-1$
\item A contribution of 
$(-1)^j{m \choose j}((n-k)(k+1) - (n - k - m + j)j - (m-j)(k-j))$ from the coordinates $(\vec{x}_V)_B$ where $|B \cap V| = j$
\end{enumerate}
Summing these contributions together, we get the following total weight for the coordinates 
$(P_k{P^T_k}\vec{x}_V)_C$ where $|C \cap V| = j$
\begin{align*}
(-1)^{j+1}&{m \choose {j+1}}(n - k - m + j+1)(j+1) + (-1)^{j-1}{m \choose {j-1}}(m-j+1)(k-j+1) \\
&+ (-1)^j{m \choose j}((n-k)(k+1) - (n - k - m + j)j - (m-j)(k-j)) \\
&= (-1)^j{m \choose j}\Big(-\frac{m-j}{j+1}(n - k - m + j+1)(j+1) - \frac{j}{m-j+1}(m-j+1)(k-j+1) \\
&+ (n-k)(k+1) - (n - k - m + j)j - (m-j)(k-j)\Big) \\
&= (-1)^j{m \choose j}\Big(-(m-j)(n - k - m + j + 1) - (n - k - m + j)j - j(k-j+1) \\
&+ (n-k)(k+1) - (m-j)(k-j)\Big)\\
&= (-1)^j{m \choose j}\Big(-m(n - k - m + j + 1) + j - jk + j^2 - j + nk + n - k^2 - k - mk \\
&+ jk + jm - j^2\Big)\\
&= (-1)^j{m \choose j}\Big(-mn + mk + m^2 - mj - m - jk + j^2 + nk \\
&+ n - k^2 - k - mk + jk + jm - j^2\Big) \\
&= (-1)^j{m \choose j}(-mn + m^2 - m + nk + n - k^2 - k) = (-1)^j{m \choose j}(n-k-m)(k+1-m)
\end{align*}
By symmetry the weight on the coordinates $(P_k{P^T_k}\vec{x}_V)_C$ where $|C \cap V| = j$ will be spread evenly. Thus, 
$\vec{x}_V$ is an eigenvector of $P_k{P^T_k}$ with eigenvalue $(n-k-m)(k+1-m)$ where 
$m = |V|$. Now we need to check that the vectors $\{\vec{x}_V: |V| = m\}$ span the eigenspace of 
$P_k{P^T_k}$ with eigenvalue $(n-k-m)(k+1-m)$ and that these are the only eigenvalues. We also need to 
find the dimension of these eigenspaces.

To do this, for each $V \subseteq \Vmst$ where $|V| \leq k$, consider the vectors 
$\{\sum_{A: A \subseteq V, |A| = j}{\vec{x}_A}: j \in [0,|V|]\}$. These vectors are all nonzero and 
live in different eigenspaces of $P_k{P^T_k}$, so they are all linearly indpendent. Moreover, by symmetry, 
for each of these vectors the value of a coordinate $B$ depends only on $|B \cap V|$. The space of vectors 
with this property has dimension $|V| + 1$ so the vectors $\{\sum_{A: A \subseteq V, |A| = j}{\vec{x}_A}: j \in [0,|V|]\}$ 
must be a basis for this vector space. This implies that if we take the vector $\vec{y}_V$ where 
$(\vec{y}_V)_B = 1$ if $V \subseteq B$ and $0$ otherwise then $\vec{y}_V \in span{\{\vec{x}_A: A \subseteq V\}}$.

Clearly, the vectors $\{\vec{y}_V: V \subseteq \Vmst, |V| = k\}$ are a basis for the vector space 
which $P_k{P^T_k}$ acts on. This implies that the vectors $\{\vec{x}_V: V \subseteq \Vmst, |V| \leq k\}$ span 
this vector space which in turn implies that $\{(n-k-m)(k+1-m): m \in [0,k]\}$ are the only eigenvalues of 
$P_k{P^T_k}$ and the eigenspace with eigenvalue $(n-k-m)(k+1-m)$ is spanned by the vectors 
$\{\vec{x}_V: V \subseteq \Vmst, |V| = m\}$.

Now note that since $k < \frac{n}{2}$, $P_k{P^T_k}$ has no zero eigenvalues. Further note that for all $j$, 
the vectors $\{\vec{y}_V: V \subseteq \Vmst, |V| = j\}$ are linearly independent and 
their span contains all of the vectors $\{\vec{x}_V: V \subseteq \Vmst, |V| \leq j\}$. To see the first part, 
assume that there are coefficients $\{c_V: V \subseteq \Vmst, |V| = j\}$ such that 
$\sum_{V: V \subseteq \Vmst, |V| = j}{c_V\vec{y}_V} = 0$. Then if we let $\vec{c}$ be the vector with 
value $c_V$ in every coordinate $V$, $P^T_{k-1} \cdots P^T_{j}\vec{c} = 0$. However, from the above, this is impoosible. 
To see the second part, for all $V \subseteq \Vmst$ where $|V| \leq j$ consider the sums 
$\{\sum_{A:A \subseteq V, |A \cap V| = i}{\vec{y}_A}: i \in [0,|V|]\}$. By symmetry, the values of these 
sums on a coordinate $B$ will depend only on $|B \cap V|$. Moreover, this sum is 
$0$ on a coordinate $B$ if $|B \cap V| < i$ and nonzero if $|B \cap V| = i$. Together, this implies that 
$\vec{x}_V$ is a linear combination of these sums, as needed.

Putting everything together, for every $j$, the sum of the dimensions of the eigenspaces of $P_k{P^T_k}$ with 
eigenvalues $(n-k-m)(k+1-m)$ where $m \leq j$ is equal to the dimension of the span of the vectors 
$\{\vec{x}_V: V \subseteq \Vmst, |V| \leq j\}$ which is equal to the dimension of the span of the vectors 
$\{\vec{y}_V: V \subseteq \Vmst, |V| = j\}$ which is ${n \choose j}$. The result now follows immediately.
\end{proof}
\begin{theorem}
Let $P$ be a matrix obtained from $P_k$ as follows:\\
Take a collection of bad vertices $V$ and bad pairs of vertices $E$. Delete 
the rows and columns containing bad vertices and pairs of vertices from $P_k$. If
\begin{enumerate}
\item $|V| < \frac{n}{4}$ 
\item There is an $m \leq \frac{n}{2000{k^3}}$ such that for every vertex $v \notin V$, 
there are at most $m$ vertices $w \notin V$ with $\{v,w\} \in E$
\end{enumerate}
then the smallest eigenvalue of $P{P^{T}}$ is at least $\frac{n}{2}$
\end{theorem}
\begin{proof}
The proof idea is as follows. The smallest eigenvalue of $P{P^{T}}$ is 
$\min_{\vec{x}}{\{\frac{{\vec{x}^T}P{P^{T}}\vec{x}}{|\vec{x}|^2}\}}$ This quantity can only 
be reduced when we subtract a symmetric positive semidefinite matrix $M$ from $P{P^T}$. Thus, we will 
choose $M$ so that $M$ is a symmetric positive semidefinite matrix and $P{P^T} - M$ has a large minimum 
eigenvalue, and this will give the claimed bound. We construct $M$ as follows:
\begin{definition}
Let $U, W$ be subsets of $\Vmst$ such that $|U| = k-1$ and $U \cap W = \emptyset$. 
Define $M_{U,W}$ so that $M_{U,W}$ has the same rows and columns as $P{P^T}$ and 
\begin{enumerate}
\item If no vertex or pair of vertices in $U \cup W$ is bad then 
$(M_{U,W})_{{A_1}{A_2}} = 1$ if $U \subseteq A_1 \subseteq U \cup W$ and $U \subseteq A_2 \subseteq U \cup W$ and $0$ otherwise.
\item If a vertex or pair of vertices in $U \cup W$ is bad then $M_{U,W} = 0$
\end{enumerate}
\end{definition}
\begin{lemma}\label{summingcompletegraphs}
If we take $M = \sum_{U,W: U,W \subseteq \Vmst, \atop |U| = k-1, |W| = j, U \cap W = \emptyset}{M_{U,W}}$ and have 
$|A_1| = |A_2| = k$ then 
\begin{enumerate}
\item If $A_1 \cup A_2$ has no bad vertices or pairs of vertices and $|A_1 \cap A_2| = k-1$ then 
$$(1 - \frac{2mj(j+k)}{n}){{n - |V| - k - 1} \choose {j-2}} \leq M_{{A_1}{A_2}} \leq {{n - |V| - k - 1} \choose {j-2}}$$
\item If $A$ has no bad vertices or pairs of vertices then 
$$(1 - \frac{2mj(j+k)}{n})k{{n - |V| - k} \choose {j-1}} \leq M_{AA} \leq k{{n - |V| - k} \choose {j-1}}$$
\item If $A_1$ or $A_2$ has any bad vertices or pairs of vertices, $M_{{A_1}{A_2}}$ is undefined. Otherwise, 
if $A_1 \cup A_2$ has any bad pairs of vertices or $|A_1 \cap A_2| < k-1$ then $M_{{A_1}{A_2}} = 0$
\end{enumerate}
\end{lemma}
\begin{proof}
Statement 3 follows immediately from the definitions and the upper bounds in statements 1 and 2 follow by
noting the these are the maximal possible number of $M_{U,W}$ which have $M_{{A_1}{A_2}}$ or $M_{AA}$ equal to $1$. For the lower bounds, 
for any $A_1,A_2$ such that $A_1 \cup A_2$ has no ``bad'' vertices or pairs of vertices and 
$|A_1 \cap A_2| \geq k-1$, take random $U$, $W$ with $U,W \subseteq \Vmst \setminus V$, $|U| = k-1$, $|W| = j$, 
$U \subseteq A_1 \subseteq U \cup W$, $U \cap W = \emptyset$ and 
$U \subseteq A_2 \subseteq U \cup W$. To do this, first randomly choose $U$. After choosing $U$, start with 
$W_0 = A_1 \cup A_2 \setminus U$ and add vertices to $W$ one at a time. We add at most $j$ vertices and the 
probability that each new vertex adds a bad pair of vertices is at most $\frac{m(j+k)}{n - |V| - (k+j)}$. 
Thus by the union bound the probability that $M_{U,W}$ is $0$ rather 
than $1$ is at most $\frac{mj(j+k)}{n - |V| - (k+j)} \leq \frac{2mj(j+k)}{n}$ as $|V| < \frac{n}{4}$ and $k,j << n$.
\end{proof}
\begin{corollary}\label{summingcompletegraphscorollary}
If we instead take \\ 
$M = \frac{1}{{{n - |V| - k - 1} \choose {j-2}}}
\sum_{U,W: U,W \subseteq \Vmst, |U| = k-1, |W| = j, U \cap W = \emptyset}{M_{U,W}}$ then \\
1. If $A_1 \cup A_2$ has no ``bad'' vertices or pairs of vertices and $|A_1 \cap A_2| = k-1$ then \\
$0 \leq (PP^T - M)_{{A_1}{A_2}} \leq \frac{2mj(j+k)}{n}$\\
2. If $A$ has no ``bad'' vertices or pairs of vertices then 
$(PP^T - M)_{AA} \geq n - k - |V| - km - \frac{k(n-k-|V|)}{j-1}$.\\
3. If $A_1 \cup A_2$ has any bad vertices, $(PP^T - M)_{{A_1}{A_2}}$ is undefined. Otherwise, 
if $A_1 \cup A_2$ has any bad pairs of vertices or $|A_1 \cap A_2| < k-1$ then $(PP^T - M)_{{A_1}{A_2}} = 0$
\end{corollary}
\begin{proof}
This follows immediately from Lemma \ref{summingcompletegraphs}, the fact that 
$(PP^T)_{AA} \geq n - k - |V| - km$ whenever $A$ has no ``bad'' vertices and the fact that 
$$\frac{{{n - |V| - k} \choose {j-1}}}{{{n - |V| - k - 1} \choose {j-2}}} = \frac{n - |V| - k}{j-1}$$
\end{proof}
\noindent The final step is as follows. For each column, we subtract all of the non-diagonal 
elements from the diagonal element and then set all non-diagonal elements to $0$. 
This will not increase the minimum eigenvalue as \\
\[ \left( \begin{array}{cc}
1 & 1 \\
1 & 1 \end{array} \right)\] \\
is a positive semidefinite matrix. Note that each column has at most $k(n-k-|V|)$ nonzero 
non-diagonal elements, so by Corollary \ref{summingcompletegraphscorollary}, all of the 
diagonal elements are still at least \\
$n - k - |V| - km - \frac{k(n-k-|V|)}{j-1} - \frac{2mj(j+k)k(n-k-|V|)}{n}$\\
Taking $j = 8k$, if $m \leq \frac{n}{2000k^{3}}$ then 
$n - k - |V| - km - \frac{k(n-k-|V|)}{j-1} - \frac{2mj(j+k)k(n-k-|V|)}{n} \geq \frac{n}{2}$. Thus, 
the minimal eigenvalue of $P{P^T}$ is at least $\frac{n}{2}$, as claimed.
\end{proof}
\begin{corollary}\label{keybuildingupcorollary}
If for a given $k,u$ we have a set of bad vertices $V$ and bad pairs of vertices $E$ such that \\
1. $|V| < \frac{n}{4}$ and \\
2. There is an $m \leq \frac{n}{2000{k^3}}$ such that for every vertex $v \notin V$, 
there are at most $m$ vertices $w \notin V$ with $\{v,w\} \in E$\\
then if $(\vec{s}_{k,u,g})_A$ is fixed if $A$ has a bad vertex or pair of vertices 
and free otherwise, and $(\vec{s}_{k+1,u-1,g})_B$ is fixed if $B$ has a bad 
vertex or pair of vertices and free otherwise, then for any $\vec{s}_{k,u,g,free}$ and 
$\vec{a}_{k,u,g}$ we can choose $\vec{s}_{k+1,u-1,g,free}$ so that 
\begin{enumerate}
\item $P_k\vec{s}_{k+1,u-1,g,free} = u\vec{s}_{k,u,g,free} + \vec{a}_{k,u,g}$ 
\item $||\vec{s}_{k+1,u-1,g,free}||^2 \leq \frac{2}{n}||u\vec{s}_{k,u,g,free} + \vec{a}_{k,u,g}||^2 
\leq \frac{4u^2}{n}||\vec{s}_{k,u,g,free}||^2 + \frac{4}{n}||\vec{a}_{k,u,g}||^2$
\end{enumerate}
\end{corollary}
\subsection{Constructing a base function}\label{constructingbase}
\begin{theorem}\label{basefunctionexistence}
If $z$ and $m$ are constants such that $m \leq \frac{n}{2000z^{4}}$ and $G$ is a directed acyclic input graph such that 
\begin{enumerate}
\item There is no path of length at most $2^(z-1)$ from $s$ to $t$.
\item For any vertex $v \in V(G)$, there are at most $m$ vertices $w \in V(G)$ such that either there is a path of length at 
most $2^{z-2}$ from $v$ to $w$ in $G$ or there is a path of length at 
most $2^{z-2}$ from $w$ to $v$ in $G$
\end{enumerate}
then there is a function $g$ such that 
\begin{enumerate}
\item $\hat{g}_{\{\}} = 1$
\item $\hat{g}_V = 0$ whenever $|V| \geq z$
\item For every $e \in E(G)$ there is a function $g_e$ such that 
\begin{enumerate}
\item $g_e$ is $e$-invariant.
\item $\hat{g_e}_V = \hat{g}_V$ whenever $|V| < z$
\end{enumerate}
\item $|s_{k,u,g}|^2 \leq (9mn)^{\frac{k+u}{2}}$
\end{enumerate}
\end{theorem}
\begin{proof}
We construct this function $g$ by choosing the vectors $s_{k,u,g}$ in increasing lexicographic order in 
$(k+u,k)$. We choose the vectors $s_{k,u,g}$ with equal values of $k+u$ in order of increasing $k$.
For fixed values of $k+u$, for each $A$, whenever $v,w \in A \cup \{s,t\}$and there is a path of length at most 
$2^{(z-k-u-1)}$ from $v$ to $w$, make $v \to w$ relevant for $(s_{k,u,g})_A$. We now have our fixed and free terms. 

For each $k$ and $u$, we choose the set of bad vertices $V$ to be the set of all 
vertices which are $(z-k-u-1)$-linked to $s$ or $t$. These vertices give a relevant path from $s$ or to $t$ which means 
that terms with these vertices are not free. Similarly, we choose the set of bad pairs of vertices $E$ to be the set of 
vertices $\{v,w\} \subseteq \Vmst$ such that either there is a path of length at 
most $2^{z-2}$ from $v$ to $w$ in $G$ or there is a path of length at 
most $2^{z-2}$ from $w$ to $v$ in $G$. These pairs of vertices give a relevant 
path between them which means that terms containing such a pair are not free. With this setup, we use 
Corollary \ref{keybuildingupcorollary} and the equations for relevance to build up $g$ step by step. Our base cases are
\begin{enumerate}
\item $\vec{s}_{0,0,g} = \hat{g}_{\{\}} = 1$
\item $\vec{s}_{0,u,g} = 0$ if $u > 0$
\end{enumerate}
\begin{lemma}\label{lowconnectivitytargetbounds}
If $k+u \geq 1$ then 
\begin{enumerate}
\item $|\vec{s}_{k,u,g,fixed}|^2 \leq \frac{1}{2}(9mn)^{\frac{k+u}{2}}$
\item $|\vec{s}_{k,u,g,free}|^2 \leq \frac{1}{2}(9mn)^{\frac{k+u}{2}}$
\end{enumerate}
\end{lemma}
\begin{proof}
We prove this by induction. We prove bounds on each term in terms of previous terms and then verify that 
if the bound holds for the previous terms it holds for the current term as well.
\begin{lemma}\label{boundinglemmaone}
For any $k \geq 1$ and any $u$,
$$||\vec{s}_{k,u,g,fixed}||^2 \leq 2||\vec{s}_{k,u-1,g}||^2 + 6||\vec{s}_{k,u-2,g}||^2 + 
(4m+6km)||\vec{s}_{k-1,u,g}||^2 + 6km||\vec{s}_{k-1,u-1,g}||^2 + 3nm||\vec{s}_{k-2,u,g}||^2$$
\end{lemma}
\begin{proof}
By our definitions for every term $(\vec{s}_{k,u,g})_A$ which is fixed one of the following is true
\begin{enumerate}
\item $w \in A$ and there is a relevant path from $s$ to $w$. In this case,
$(\vec{s}_{k,u,g})_A = (\vec{s}_{k,u-1,g})_A - (\vec{s}_{k-1,u,g})_{A \setminus \{w\}}$
This implies that $((\vec{s}_{k,u,g})_A)^2 \leq 2((\vec{s}_{k,u-1,g})_A)^2 + 2((\vec{s}_{k-1,u,g})_{A \setminus \{w\}})^2$
\item $v \in A$ and there is a relevant path from $v$ to $t$. In this case,
$(\vec{s}_{k,u,g})_A = (\vec{s}_{k-1,u,g})_{A \setminus \{v\}} - (\vec{s}_{k,u-1,g})_A$
This implies that $((\vec{s}_{k,u,g})_A)^2 \leq 2((\vec{s}_{k-1,u,g})_{A \setminus \{v\}})^2 + 2((\vec{s}_{k,u-1,g})_A)^2$
\item $v,w \in A$ and there is a relevant path from $v$ to $w$. In this case,
\begin{align*}
(\vec{s}_{k,u,g})_A &=  (\vec{s}_{k-1,u,g})_{A \setminus \{v\}} - (\vec{s}_{k-1,u,g})_{A \setminus \{w\}} 
+ (\vec{s}_{k-2,u,g})_{A \setminus \{v,w\}} \\
&- (\vec{s}_{k-1,u-1,g})_{A \setminus \{v\}} - (\vec{s}_{k-1,u-1,g})_{A \setminus \{w\}} + 
(\vec{s}_{k,u-2,g})_{A}
\end{align*}
This implies that
\begin{align*}
(\vec{s}_{k,u,g})_A &\leq  6((\vec{s}_{k-1,u,g})_{A \setminus \{v\}})^2 + 6((\vec{s}_{k-1,u,g})_{A \setminus \{w\}})^2 
+ 6((\vec{s}_{k-2,u,g})_{A \setminus \{v,w\}})^2 \\
&+ 6((\vec{s}_{k-1,u-1,g})_{A \setminus \{v\}})^2 + 6((\vec{s}_{k-1,u-1,g})_{A \setminus \{w\}})^2 + 
6((\vec{s}_{k,u-2,g})_{A})^2
\end{align*}
\end{enumerate}
To bound $||\vec{s}_{k,u,g,fixed}||^2$ we sum the above inequalities over every $A$ such that 
$(\vec{s}_{k,u,g})_A$ is fixed. Consider how many times each term can appear in this sum.
Terms of the form $((\vec{s}_{k,u-1,g})_A)^2$ and $((\vec{s}_{k,u-2,g})_A)^2$ will only appear once 
(with the appropriate constant in front). Terms of the form $((\vec{s}_{k-1,u,g})_{A})^2$ 
will appear up to $2m$ times from cases 1 and 2 (the vertex added to $A$ must be $(z-k-u-1)$-linked to $s$ or $t$) and up to 
$km$ times from case 3 (the vertex added to $A$ must be $(z-k-u-1)$-linked to a vertex in $A$). 
Terms of the form $((\vec{s}_{k-1,u-1,g})_{A})^2$ will appear up to $km$ times
(the vertex added to $A$ must be $(z-k-u-1)$-linked to a vertex in $V$). Finally, terms of the form 
$((\vec{s}_{k-2,u,g})_{A})^2$ will appear up to $\frac{nm}{2}$ times (the two vertices added to $A$ must be 
$(z-k-u-1)$-linked to each other). Putting everything together, 
$$||\vec{s}_{k,u,g,fixed}||^2 \leq 2||\vec{s}_{k,u-1,g}||^2 + 6||\vec{s}_{k,u-2,g}||^2 + 
(4m+6km)||\vec{s}_{k-1,u,g}||^2 + 6km||\vec{s}_{k-1,u-1,g}||^2 + 3nm||\vec{s}_{k-2,u,g}||^2$$
\end{proof}
\begin{corollary}\label{checkingboundone}
For a given $k \geq 1$ and any $u$, if the bounds hold for earlier terms then 
${||\vec{s}_{k,u,g,fixed}||^2 \leq \frac{1}{2}(9mn)^{\frac{k+u}{2}}}$
\end{corollary}
\begin{proof}
If $k = 0$ then $\vec{s}_{k,u,g,fixed} = 0$. By Lemma \ref{boundinglemmaone} and the bounds of 
Lemma \ref{lowconnectivitytargetbounds}, 
\begin{align*}
|\vec{s}_{k,u,g,fixed}|^2 & \leq 2(9mn)^{\frac{k+u-1}{2}} + 6(9mn)^{\frac{k+u-2}{2}} + 
(4m+6km)(9mn)^{\frac{k+u-1}{2}} + 6km(9mn)^{\frac{k+u-2}{2}} + 3nm(9mn)^{\frac{k+u-2}{2}} \\
& \leq \left(\frac{2}{\sqrt{9mn}} + \frac{6}{9mn} + \sqrt{\frac{100{k^2}{m^2}}{9mn}} + \frac{6km}{9mn} + \frac{1}{3}\right)
(9mn)^{\frac{k+u}{2}}
\end{align*}
Since $n \geq 2000mz^4$ this is less than $\frac{1}{2}(9mn)^{\frac{k+u}{2}}$
\end{proof}
\begin{lemma}\label{boundinglemmatwo}
For any $k \geq 1$ and any $u$, 
$$||\vec{a}_{k-1,u+1,g}||^2 \leq m(k+1)k||\vec{s}_{k,u,g,fixed}||^2$$
\end{lemma}
\begin{proof}
Recall that $\vec{a}_{k-1,u+1,g}$ is the projection of $-P_{k-1}\vec{s}_{k,u,g,fixed}$ onto the free terms of 
$\vec{s}_{k-1,u+1,g}$. Now $(P_{k-1}\vec{s}_{k,u,g,fixed})_A = \sum_{A \subseteq B, |B| = k}{(\vec{s}_{k,u,g,fixed})_B}$. 
For the $A$ such that $(\vec{s}_{k-1,u+1,g})_A$ is free, there are at most $(k - 1 + 2)m$ $B$ such that 
$A \subseteq B$, $|B| = k$, and $(\vec{s}_{k,u,g})_B$ is fixed. For all such $A$, the sum 
$\sum_{A \subseteq B, |B| = k}{(\vec{s}_{k,u,g,fixed})_B}$ has at most $(k - 1 + 2)m$ terms so 
$$((P_{k-1}\vec{s}_{k,u,g,fixed})_A)^2 \leq (k + 1)m\sum_{A \subseteq B, |B| = k}{((\vec{s}_{k,u,g,fixed})_B)^2}$$
Summing over all such $A$, each term $((\vec{s}_{k,u,g,fixed})_B)^2$ appears at most $k$ times so 
$$||\vec{a}_{k-1,u+1,g}||^2 \leq m(k+1)k||\vec{s}_{k,u,g,fixed}||^2$$
\end{proof}
\begin{corollary}\label{boundinglemmatwocorollary}
For any $k \geq 1$ and any $u$, 
$$||\vec{s}_{k,u,g,free}||^2 \leq \frac{4u^2}{n}||\vec{s}_{k-1,u+1,g,free}||^2 + \frac{4m(k+1)k}{n}||\vec{s}_{k,u,g,fixed}||^2$$
\end{corollary}
\begin{proof}
By Corollary \ref{keybuildingupcorollary}, 
$$||\vec{s}_{k,u,g,free}||^2 \leq \frac{4u^2}{n}||\vec{s}_{k-1,u+1,g,free}||^2 + \frac{4}{n}||\vec{a}_{k-1,u+1,g}||^2$$
Now by Lemma \ref{boundinglemmatwo}, 
$$||\vec{s}_{k,u,g,free}||^2 \leq \frac{4u^2}{n}||\vec{s}_{k-1,u+1,g,free}||^2 + \frac{4m(k+1)k}{n}||\vec{s}_{k,u,g,fixed}||^2$$
\end{proof}
\begin{corollary}\label{checkingboundstwo} 
For any $k \geq 1$ and any $u$ if previous bounds hold then 
$|\vec{s}_{k,u,g,free}|^2 \leq \frac{1}{2}(9mn)^{\frac{k+u}{2}}$
\end{corollary}
\begin{proof}
This follows immediately from Corollary \ref{boundinglemmatwocorollary} and the fact that $n \geq 2000mz^4$.
\end{proof}
This completes the proof of Lemma \ref{lowconnectivitytargetbounds}. 
\end{proof}
The final thing to check for Theorem \ref{basefunctionexistence} is that there is indeed a function $g$ 
corresponding to the vectors $\vec{s}_{k,u,g}$ and that $g$ can be extended to an $e$-invariant function 
$g_e$ for every $e \in E(G)$. The fact that the sum vectors $\{\vec{s}_{k,u,g}\}$ do indeed correspond to a function $g$ 
follows from Corollary \ref{constructionideacorollary}. To show that $g$ can be extended to an $e$-invariant 
function $g_e$ for every $e \in E(G)$, we 
contruct such an extension explicitly with the following proposition.
\begin{proposition}\label{constructingge}
Let $z$ be a constant and let $b_e$ be a function such that $\hat{b_e}_V = 0$ whenever $|V| \geq z$.
\begin{enumerate}
\item If $e$ is of the form $s \to w$ and $\hat{b_e}_{V \cup \{w\}} = -\hat{b_e}_{V}$ 
whenever $w \notin V, |V| < z - 1$, if we take $g_e$ so that 
\begin{enumerate}
\item $\hat{g_e}_V = \hat{b_e}_V$ whenever $|V| \neq z$ and 
\item $\hat{g_e}_V = -\hat{b_e}_{V \setminus \{w\}}$ if $w \in V, |V| = z$ and $0$ if $w \notin V, |V| = z$
\end{enumerate}
then $g_e$ is $e$-invariant.
\item If $e$ is of the form $v \to t$ and $\hat{b_e}_{V \cup \{v\}} = \hat{b_e}_{V}$ 
whenever $v \notin V, |V| < z - 1$, if we take $g_e$ so that 
\begin{enumerate}
\item $\hat{g_e}_V = \hat{b_e}_V$ whenever $|V| \neq z$ and 
\item $\hat{g_e}_V = \hat{b_e}_{V \setminus \{v\}}$ if $v \in V, |V| = z$ and $0$ if $v \notin V, |V| = z$
\end{enumerate}
then $g_e$ is $e$-invariant.
\item If $e$ is of the form $v \to w$ and 
$\hat{b_e}_{V \cup \{v,w\}} = -\hat{b_e}_{V \cup \{v\}} + \hat{b_e}_{V \cup \{w\}} + \hat{b_e}_{V}$ whenever 
$v,w \notin V, |V| < z - 2$, if we take $g_e$ so that 
\begin{enumerate}
\item $\hat{g_e}_V = \hat{b_e}_V$ whenever $|V| \neq z$
\item $\hat{g_e}_V = -\hat{b_e}_{V \setminus \{w\}} + \hat{b_e}_{V \setminus \{v\}} + \hat{b_e}_{V \setminus \{v,w\}}$ 
if $v,w \in V, |V| = z$
\item $\hat{g_e}_V = \hat{b_e}_{V \setminus \{v\}}$ if $w \notin V, v \in V, |V| = z$\\
\item $\hat{g_e}_V = 0$ if $v \notin V, |V| = z$
\end{enumerate}
then $g_e$ is $e$-invariant.
\end{enumerate}
\end{proposition}
\begin{proof}
This follows immediately from Corollary \ref{einvariancecorollary}.
\end{proof}
Here we take $b_e = g$ for every $e$. To check that the conditions of this proposition hold for $g$ we note that 
any edge $e \in E(G)$ is always relevant for any $(\vec{s}_{k,u,g})_A$ such that $k+u < z$ and the vertices 
of $e$ are contained in $A \cup \{s,t\}$. We then use Lemma \ref{einvariancelemma} on the terms $(\vec{s}_{k,0,g})_A$.
\end{proof}
\subsection{The cost of extending a base function}\label{cuttingoff}
We now have our base function $g$ and an explicit construction of the $e$-invariant functions $\{g_e\}$. However, 
in constructing the functions $\{g_e\}$ from $g$ we were concerned with $e$-invariance and ensuring that 
$\hat{g_e}_V = 0$ whenever $|V| > z$, we have not yet considered how large the norms $|\vec{s}_{k,u,g_e}|^2$ would be.
We need to check that the norms $|\vec{s}_{k,u,g_e}|^2$ are not too large.
\begin{lemma}\label{cuttingoffcost}
If $b_e$ and $g_e$ satisfy the conditions described in Proposition \ref{constructingge}, then 
\begin{enumerate}
\item If $e$ is of the form $s \to w$ then 
\begin{enumerate}
\item If $w \notin A$ and $|A| + u = z$, $s_{A,u}(g_e) = -s_{A,u-1}(b_e) + s_{A \cup \{w\},u-2}(b_e)$
\item If $w \in A$ and $|A| + u = z$, $s_{A,u}(g_e) = -s_{A \setminus \{w\},u}(b_e) + s_{A,u-1}(b_e)$
\end{enumerate}
\item If $e$ is of the form $v \to t$ then 
\begin{enumerate}
\item If $v \notin A$ and $|A| + u = z$, $s_{A,u}(g_e) = s_{A,u-1}(b_e) - s_{A \cup \{v\},u-2}(b_e)$
\item If $v \in A$ and $|A| + u = z$, $s_{A,u}(g_e) = s_{A \setminus \{v\},u}(b_e) - s_{A,u-1}(b_e)$
\end{enumerate}
\item If $e$ is of the form $v \to w$ then 
\begin{enumerate}
\item If $v,w \notin A$ and $|A| + u = z$,
\begin{align*}
s_{A,u}(g_e) &= s_{A,u - 1}(b_e) - 2s_{A \cup \{v\},u - 2}(b_e) + s_{A \cup \{v,w\},u - 3}(b_e) + s_{A,u - 2}(b_e) \\
&- s_{A \cup \{v\},u - 3}(b_e) - s_{A \cup \{w\},u - 3}(b_e) + s_{A \cup \{v,w\},u - 4}(b_e)
\end{align*}
\item If $v \in A, w \notin A$ and $|A| + u = z$ then 
\begin{align*}
s_{A,u}(g_e) &= s_{A \setminus \{v\},u}(b_e) - 2s_{A,u - 1}(b_e) + s_{A \cup \{w\},u - 2}(b_e) + 
s_{A \setminus \{v\},u - 1}(b_e) \\ 
&- s_{A,u - 2}(b_e) - s_{A \setminus \{v\} \cup \{w\},u - 2}(b_e) + s_{A \cup \{w\},u - 3}(b_e)
\end{align*}
\item If $v \notin A, w \in A$ and $|A| + u = z$,
\begin{align*}
s_{A, u}(g_e) &= 
s_{A, u - 1}(b_e) - s_{A \cup \{v\} \setminus \{w\}, u - 1}(b_e) + s_{A \setminus \{w\},u - 1}(b_e) \\
&- s_{A,u - 2}(b_e) - s_{A \cup \{v\} \setminus \{w\},u - 2}(b_e) + s_{A \cup \{v\},u - 3}(b_e)
\end{align*}
\item If $v,w \in A$ and $|A| + u = z$, 
\begin{align*}
s_{A,u}(g_e) &= s_{A \setminus \{v\}, u}(b_e) - s_{A \setminus \{w\}, u}(b_e) + s_{A \setminus \{v,w\},u}(b_e) \\
&- s_{A \setminus \{v\},u - 1}(b_e) - s_{A \setminus \{w\},u - 1}(b_e) + s_{A,u-2}(b_e)
\end{align*}
\end{enumerate}
\end{enumerate}
\end{lemma}
\noindent Before proving this lemma, we give a corollary which is less exact but much simpler and 
easier to use directly.
\begin{corollary}\label{cuttingoffcostcorollary}
If $b_e$ and $g_e$ satisfy the conditions described in Proposition \ref{constructingge} then 
$$\forall k, \sum_{A:|A| = k}{(s_{A,z-k}(g_e))^2} \leq 
200\max_{k_2}{\left\{\max{\left
\{\sum_{B: |B| = k_2}{(s_{B,z-k_2-1}(b_e))^2}, \sum_{B: |B| = k_2}{(s_{B,z-k_2-2}(b_e))^2}\right\}}\right\}}$$
\end{corollary}
\begin{proof}
Consider the equation in statement 3a of Lemma \ref{cuttingoffcost}:
\begin{align*}
s_{A,u}(g_e) &= s_{A,u - 1}(b_e) - 2s_{A \cup \{v\},u - 2}(b_e) + s_{A \cup \{v,w\},u - 3}(b_e) + s_{A,u - 2}(b_e) \\
&- s_{A \cup \{v\},u - 3}(b_e) - s_{A \cup \{w\},u - 3}(b_e) + s_{A \cup \{v,w\},u - 4}(b_e)
\end{align*}
Applying a Cauchy-Scwarz inequality to this, 
\begin{align*}
(s_{A,u}(g_e))^2 &\leq 10 ((s_{A,u - 1}(b_e))^2 + (s_{A \cup \{v\},u - 2}(b_e))^2 + 
(s_{A \cup \{v,w\},u - 3}(b_e))^2 + (s_{A,u - 2}(b_e))^2 \\
&- (s_{A \cup \{v\},u - 3}(b_e))^2 - (s_{A \cup \{w\},u - 3}(b_e))^2 + (s_{A \cup \{v,w\},u - 4}(b_e))^2)
\end{align*}
Similarly, we can apply a Cauchy-Scwarz inequality to all of the other equations.\\
Now note that the sum $\sum_{A:|A| = k}{(s_{A,z-k}(g_e))^2}$ can be bounded by a sum of terms of the 
form $(s_{B,u}(g_e))^2$ where $|B| + u = z-1$ or $z-2$ and $|(|B| - k)| \leq 2$. Moreover, each 
term $(s_{B,u}(g_e))^2$ must come from an $A$ with $|A| = k, A \bigtriangleup B \subseteq \{v,w\}$ 
where $A \bigtriangleup B$ is the symmetric difference of $A$ and 
$B$ i.e. $A \bigtriangleup B = \{u: u \in A, u\notin B$ or $u \in B, u\notin A\}$. For 
each $B$ there are at most two $A$ which will give a term with that $B$. This implies that 
the coefficient for each term $(s_{B,u}(g_e))^2$ has magnitude at most $2$, so we have the inequality 
$$\forall k, \sum_{A:|A| = k}{(s_{A,z-k}(g_e))^2} \leq 
20\sum_{k_2: |k-k_2| \leq 2}{\left(\sum_{B: |B| = k_2}{(s_{B,z-k_2-1}(b_e))^2} + \sum_{B: |B| = k_2}{(s_{B,z-k_2-2}(b_e))^2}\right)}$$
The result now follows immediately.
\end{proof}
\begin{proof}[Proof of Lemma \ref{cuttingoffcost}]
For statement 1, note that if $w \notin A$ and $|A| + u = z$,
\begin{align*}
s_{A,u}(g_e) &= \sum_{B: A \subseteq B, |B| = z-1, w \notin B}{\hat{g_e}_{B \cup \{w\}}} \\
&= -\sum_{B: A \subseteq B, |B| = z-1, w \notin B}{\hat{b_e}_{B}} \\
&= -\sum_{B: A \subseteq B, |B| = z-1}{\hat{b_e}_{B}} + \sum_{B: A \subseteq B, |B| = z-1, w \in B}{\hat{b_e}_{B}} \\
&= -s_{A,u-1}(b_e) + \sum_{B: A \cup \{w\} \subseteq B, |B| = z-1}{\hat{b_e}_{B}} \\
&= -s_{A,u-1}(b_e) + s_{A \cup \{w\},u-2}(b_e)
\end{align*}
If $w \in A$ and $|A| + u = z$ then by Proposition \ref{constructingge} and Lemma \ref{einvariancelemma}, 
$$s_{A,u}(g_e) = -s_{A \setminus \{w\},u}(g_e) + s_{A,u - 1}(g_e) = -s_{A \setminus \{w\},u}(b_e) + s_{A,u - 1}(b_e)$$
For statement 2, note that if $v \notin A$ and $|A| + u = z$,
\begin{align*}
s_{A,u}(g_e) &= \sum_{B: A \subseteq B, |B| = z-1, v \notin B}{\hat{g_e}_{B \cup \{v\}}} \\
&= \sum_{B: A \subseteq B, |B| = z-1, v \notin B}{\hat{b_e}_{B}} \\ 
&= \sum_{B: A \subseteq B, |B| = z-1}{\hat{b_e}_{B}} - \sum_{B: A \subseteq B, |B| = z-1, v \in B}{\hat{b_e}_{B}} \\
&= s_{A,u-1}(b_e) - \sum_{B: A \cup \{v\} \subseteq B, |B| = z-1}{\hat{b_e}_{B}} \\
&= s_{A,u-1}(b_e) - s_{A \cup \{v\},u-2}(b_e)
\end{align*}
If $v \in A$ and $|A| + u = z$ then by Proposition \ref{constructingge} and Lemma \ref{einvariancelemma}, \\
$$s_{A,u}(g_e) = s_{A \setminus \{v\},u}(g_e) - s_{A,u - 1}(g_e) = s_{A \setminus \{v\},u}(b_e) - s_{A,u - 1}(b_e)$$
Again, the proof for statement 3 is similar but more complicated. Note that if $v,w \notin A$ and $|A| + u = z$, 
$$s_{A,u}(g_e) = \sum_{B: A \subseteq B, |B| = z-1, v,w \notin B}{\hat{g_e}_{B \cup \{v\}}} + 
\sum_{B: A \subseteq B, |B| = z-2, v,w \notin B}{\hat{g_e}_{B \cup \{v,w\}}}$$
Let's consider each term separately. For the first term,
\begin{align*}
\sum_{B: A \subseteq B, |B| = z-1, v,w \notin B}{\hat{g_e}_{B \cup \{v\}}} &= 
\sum_{B: A \subseteq B, |B| = z-1, v,w \notin B}{\hat{b_e}_{B}} \\
&= \sum_{B: A \subseteq B, \atop |B| = z-1}{\hat{b_e}_{B}} - 
\sum_{B: A \subseteq B, \atop |B| = z-1, v \in B}{\hat{b_e}_{B}} - 
\sum_{B: A \subseteq B, \atop |B| = z-1, w \in B}{\hat{b_e}_{B}} + 
\sum_{B: A \subseteq B, \atop |B| = z-1, v,w \in B}{\hat{b_e}_{B}} \\
&= s_{A,u - 1}(b_e) - s_{A \cup \{v\},u - 2}(b_e) - s_{A \cup \{w\},u - 2}(b_e) + s_{A \cup \{v,w\},u - 3}(b_e)
\end{align*}
For the second term,
$$\sum_{B: A \subseteq B, |B| = z-2, v,w \notin B}{\hat{g_e}_{B \cup \{v,w\}}} = 
\sum_{B: A \subseteq B, |B| = z-2, v,w \notin B}{(\hat{b_e}_{B \cup \{w\}} - \hat{b_e}_{B \cup \{v\}} + \hat{b_e}_{B})}$$
Consider each of these terms separately. By the same logic as above, 
$$\sum_{B: A \subseteq B, |B| = z-2, v,w \notin B}{\hat{b_e}_{B}} = 
s_{A,u - 2}(b_e) - s_{A \cup \{v\},u - 3}(b_e) - s_{A \cup \{w\},u - 3}(b_e) + s_{A \cup \{v,w\},u - 4}(b_e)$$
For the first term,
\begin{align*}
\sum_{B: A \subseteq B, |B| = z-2, v,w \notin B}{\hat{b_e}_{B \cup \{w\}}} &= 
\sum_{B: A \cup \{w\} \subseteq B, |B| = z-1, v \notin B}{\hat{b_e}_{B}} \\
&= \sum_{B: A \cup \{w\} \subseteq B, |B| = z-1}{\hat{b_e}_{B}} - 
\sum_{B: A \cup \{w\} \subseteq B, v \in B, |B| = z-1}{\hat{b_e}_{B}} \\
&= s_{A \cup \{w\}, u - 2}(b_e) - s_{A \cup \{v,w\}, u - 3}(b_e)
\end{align*}
Following the same logic, 
$$\sum_{B: A \subseteq B, |B| = z-2, v,w \notin B}{\hat{b_e}_{B \cup \{v\}}} = 
s_{A \cup \{v\}, u - 2}(b_e) - s_{A \cup \{v,w\}, u - 3}(b_e)$$
\noindent Putting everything together, 
\begin{align*}
s_{A,u}(g_e) &= (s_{A,u - 1}(b_e) - s_{A \cup \{v\},u - 2}(b_e) - s_{A \cup \{w\},u - 2}(b_e) + s_{A \cup \{v,w\},u - 3}(b_e)) \\
&+ (s_{A,u - 2}(b_e) - s_{A \cup \{v\},u - 3}(b_e) - s_{A \cup \{w\},u - 3}(b_e) + s_{A \cup \{v,w\},u - 4}(b_e)) \\
&+ (s_{A \cup \{w\}, u - 2}(b_e) - s_{A \cup \{v,w\}, u - 3}(b_e)) - (s_{A \cup \{v\}, u - 2}(b_e) - s_{A \cup \{v,w\}, u - 3}(b_e)) \\
&= s_{A,u - 1}(b_e) - 2s_{A \cup \{v\},u - 2}(b_e) + s_{A \cup \{v,w\},u - 3}(b_e)
+ s_{A,u - 2}(b_e) \\
&- s_{A \cup \{v\},u - 3}(b_e) - s_{A \cup \{w\},u - 3}(b_e) + s_{A \cup \{v,w\},u - 4}(b_e)
\end{align*}
Now that we have done this calculation, note that if $v \in A, w \notin A$ and $|A| + u = z$ then 
\begin{align*}
s_{A,u}(g_e) = s_{A \setminus \{v\},u+1}(g_e) &= 
s_{A \setminus \{v\},u}(b_e) - 2s_{A,u - 1}(b_e) + s_{A \cup \{w\},u - 2}(b_e) + s_{A \setminus \{v\},u - 1}(b_e) \\
&- s_{A,u - 2}(b_e) - s_{A \setminus \{v\} \cup \{w\},u - 2}(b_e) + s_{A \cup \{w\},u - 3}(b_e)
\end{align*}
By Proposition \ref{constructingge} and Lemma \ref{einvariancelemma}, if $v,w \in A$ and $|A| + u = z$ then 
$$s_{A,u}(g_e) = s_{A \setminus \{v\}, u}(b_e) - s_{A \setminus \{w\}, u}(b_e) + s_{A \setminus \{v,w\},u}(b_e) - 
s_{A \setminus \{v\},u - 1}(b_e) - s_{A \setminus \{w\},u - 1}(b_e) + s_{A,u-2}(b_e)$$
Finally, note that if $v \notin A, w \in A$ and $|A| + u = z$ then 
\begin{align*}
s_{A, u}(g_e) = s_{A \cup \{v\},u - 1}(g_e) &= 
s_{A, u - 1}(b_e) - s_{A \cup \{v\} \setminus \{w\}, u - 1}(b_e) + s_{A \setminus \{w\},u - 1}(b_e) \\
&- s_{A,u - 2}(b_e) - s_{A \cup \{v\} \setminus \{w\},u - 2}(b_e) + s_{A \cup \{v\},u - 3}(b_e)
\end{align*}
\end{proof}
\subsection{Putting everything together: A lower bound}\label{puttingeverythingtogether}
We now put everything together and prove a lower bound on $m(G)$
\vskip.1in
\noindent
{\bf Theorem \ref{lowerboundforlowconnectivity}.}
{\it
If $z$ and $m$ are constants such that $m \leq \frac{n}{2000z^{4}}$ and $G$ is a directed acyclic input graph such that 
\begin{enumerate}
\item There is no path of length at most $2^{z-1}$ from $s$ to $t$.
\item For any vertex $v \in V(G)$, there are at most $m$ vertices $w \in V(G)$ such that either there is a path of length at 
most $2^{z-2}$ from $v$ to $w$ in $G$ or there is a path of length at 
most $2^{z-2}$ from $w$ to $v$ in $G$
\end{enumerate}
then $$m(G) \geq \frac{(9mn)^{\frac{1}{4}}}{20|E(G)|(z+1)\sqrt{2^{z}z!}}(\frac{n}{9m})^\frac{z}{4}$$
}
\begin{proof}
We obtain the $e$-invariant functions $\{g_e\}$ from Theorem \ref{basefunctionexistence} and Proposition \ref{constructingge}. 
Now by Corollary \ref{cuttingoffcostcorollary} and Theorem \ref{basefunctionexistence}, for all $e \in E(G)$ if $k + u = z$ 
then  $||s_{k,u,g_{e}}||^2 \leq 200(9mn)^{\frac{z-1}{2}}$. This implies that for all $e_1,e_2 \in E(G)$, 
$||s_{k,u,g_{e_2} - g_{e_1}}||^2 \leq 800(9mn)^{\frac{z-1}{2}}$ if $k+u = z$ and is $0$ otherwise.

By Theorem \ref{lowerboundtheorem},
$$m(G) \geq \frac{2}{|E(G)| - 1}{\left(\max_{e \in E(G) \setminus \{e_0\}}
\left\{2(z+1)\sum_{k,u}{\frac{2^{k}(k+u)!}{n^{k+u}}||s_{k,u,g_e - g_{e_0}}||^2}\right\}\right)}^{-\frac{1}{2}}$$
which implies that 
$$m(G) \geq \frac{(9mn)^{\frac{1}{4}}}{20|E(G)|(z+1)\sqrt{2^{z}z!}}(\frac{n}{9m})^\frac{z}{4}$$
as needed.
\end{proof}
\section{Upper bounds: Solving directed connectivity with parity}\label{upperbounds}
\subsection{The reversible pebble game for directed connectivity}
One tool for upper bounds is the reversible pebble game for directed connectivity. This pebble game was 
introduced by Bennet \cite{cbennet} to study time/space tradeoffs in computation. In this subsection we will explore 
what upper bounds can be proven using just this reversible pebble game. In the next subsection, we will use it as 
a component in a more general lower bound.
\begin{definition}
In the reversible pebble game for directed connectivity on an input graph $G$, we start with a pebble on $s$ and 
use only the following type of move
\begin{enumerate}
\item If there is a pebble on a vertex $v$ and an edge from $v$ to $w$ in $G$ then we may add or remove a pebble from $w$.
\end{enumerate}
We win the reversible pebble game for directed connectivity if we put a pebble on $t$.
\end{definition}
\begin{proposition}
We can win the reversible pebble game for $G$ if and only if there is a path from $s$ to $t$ in $G$.
\end{proposition}
\begin{definition}
Given a set $I$ of input graphs each of which contains a path from $s$ to $t$, define $r(I,k)$ to be the 
smallest size of a set $S$ of states of the reversible pebble game each of which has at most $k$ pebbles on vertices 
in $\Vmst$ such that for any input graph $G$ in $I$, it is possible to win the reversible pebble game on $G$ while only passing 
through game states in $S$ when going from the starting state to a winning state (the starting state and winning state do not 
need to be included in $S$). If there is no such set $S$ then define $r(I,k) = \infty$.
\end{definition}
\begin{proposition}\label{pebblegameupperbounds}
For all $k$, $m(G) \leq r(\{\sigma(G): \sigma \in S_{\Vmst}\},k)$
\end{proposition}
\begin{proof}[Proof sketch]
The idea is to create a switching network where each vertex $v'$ corresponds to a state in $S$ and each edge $e'$ corresponds 
to a move between states. For a more thorough explanation, see Section 3 of \cite{potechin}.
\end{proof}
As noted below, Proposition \ref{pebblegameupperbounds} gives an upper bound corresponding to Savitch's algorithm 
for all input graphs $G$.
\begin{definition}
Given an input graph $G$ containing a path from $s$ to $t$, let the reversible pebbling number $p(G)$ be the minimum number 
of vertices in $\Vmst$ which must be pebbled at one time in order to win the pebbling game on $G$.
\end{definition}
\begin{theorem}\label{pebblingnumber}
If $G$ is an input graph containing a path from $s$ to $t$ of length $l$ then $p(G) \leq \lceil{\lg{l}}\rceil$
\end{theorem}
\begin{proof}
The idea is based on Savitch's algorithm and this result was implicitly noted in Bennet's paper \cite{cbennet} introducing the 
reversible pebble game. We also gave a proof in \cite{potechin}.
\end{proof}
\begin{corollary}
For all input graphs $G$ containing a path from $s$ to $t$, \\
$m(G) \leq r(\{\sigma(G): \sigma \in S_{\Vmst}\},p(G)) \leq n^{\lceil{\lg{l}}\rceil}$
\end{corollary}
We end this subsection by noting that there is a simple graph $G$ where Proposition \ref{pebblegameupperbounds} is 
sufficient to give an even better upper bound than $n^{O(\lg{l})}$ where $l$ is the shortest path from $s$ to $t$ in $G$. 
We analyze this graph $G$ as a warm up for our more general lower bounds.
\begin{theorem}\label{simpleprobabilisticupperbound}
Let $G$ be an input graph with vertex set $V(G) = \{s,t,v_1, \cdots, v_n\}$ and edge set 
$$E(G) = \{s \to v_1,v_k \to t\} \cup \{v_i \to v_{i+1}: i \in [1,k-1]\} \cup \{s \to v_i: i \in [k+1,n]\}$$
Then $m(G) \leq r(\{\sigma(G): \sigma \in S_{\Vmst}\},n) \leq k!kn\lg{n}$
\end{theorem}
\begin{proof}
Randomly choose a permutation $w_1,\cdots,w_n$ of the vertices $v_1,\cdots,v_n$ and then take the set of states 
$\{\{w_1, \cdots, w_j\}: j \in [1,n]\}$ where the state $V$ corresponds to having pebbles on the vertices in $V \cup \{s\}$.
Using this set of states (plus the starting state and winning states), we can win the pebble game on $\sigma(G)$ if and only 
if $\sigma(v_1)$, $\sigma(v_2), \cdots, \sigma(v_k)$ are in order in $w_1,\cdots,w_n$. The probability of this happening 
for a random permutation is $\frac{1}{k!}$. Thus, with this set of states we can win for $\frac{1}{k!}$ of the input graphs we 
are looking at.

If we start with an empty set of states $S$ and then do this repeatedly, adding the new states to $S$ each time, then 
on average each iteration will reduce the number of input graphs for which we cannot win using only the states in $S$ 
(plus the starting state and winning states) by a factor of $\frac{k!}{k!-1}$. This implies that if we instead greedily 
choose the permutations to reduce the number of input graphs for which we cannot win using only the states in $S$ 
(plus the starting state and winning states) by as much as possible, for each iteration we will get a reduction by a 
factor of at least $\frac{k!}{k!-1}$.

There are at most $n^k$ possible input graphs so to eliminate all of them we need at most \\
$\log_{(\frac{k!}{k!-1})}{n^k} \leq k!k\lg{n}$ iterations. Each iteration adds $n$ states to $S$ so the total 
size of $S$ is at most $k!kn\lg{n}$, as needed.
\end{proof}
\begin{remark}
If $k << \lg{n}$ then Theorem \ref{simpleprobabilisticupperbound} gives a better upper bound than $n^{\lceil{\lg{k+1}}\rceil}$
\end{remark}
\subsection{A general upper bound}
In this subsection, we state and prove our general upper bound on $m(G)$.
\begin{definition}
For a given input graph $G$, call a vertex $v \in \Vmst$ a lollipop if $s \to v \in E(G)$ or $v \to t \in E(G)$
\end{definition}
\begin{theorem}\label{generalupperbound}
Let $G$ be an input graph which contains a subgraph $G_0$ such that 
\begin{enumerate}
\item $G_0$ contains a path from $s$ to $t$.
\item All vertices in $V(G) \setminus V(G_0)$ are lollipops.
\item Letting $k = |V(G_0) \setminus \{s,t\}|$, $k \geq 2$ and $n = |\Vmst|$ is divisible by $k$
\end{enumerate}
then for any $z \in [1,k]$, there is a set of sets of sets of vertices 
$\{\mathcal{V}_x\}$ such that 
\begin{enumerate}
\item Each $\mathcal{V}_x = \{V_{x1}, \cdots, V_{xk}\}$ partitions the vertices of $\Vmst$ into $k$ sets 
of $\frac{n}{k}$ vertices.
\item $|\{\mathcal{V}_x\}| \leq 2{(4k)^z}k\lg{n}$
\item $m(G) \leq 2z{2^z}r(G_0,z)|\{\mathcal{V}_x\}|n\lg{n}(2|\{\mathcal{V}_x\}| + 4 + zn)$
\end{enumerate}
\end{theorem}
Before proving this result, we need some preliminary results.
\begin{proposition}\label{functionstoswitchingprop}
Assume that we have a monotone switching network $G'$ and an assignment of a function $h_{v'}: \mathcal{C} \to \mathbb{R}$ to 
each vertex $v' \in V(G')$ such that 
\begin{enumerate}
\item $h_{s'}(C) = -1$ for all $C \in \mathcal{C}$
\item $h_{t'}(C) = 1$ for all $C \in \mathcal{C}$
\item If there is an edge with label $e$ between vertices $v'$ and $w'$ in $G'$ then 
$h_{w'}(C) = h_{v'}(C)$ for any cut $C$ which is not crossed by $e$.
\end{enumerate}
then $G'$ is sound.
\end{proposition}
\begin{proof}
Consider a path $P'$ from $s'$ to $t'$ in $G'$ whose labels are all consistent with an input graph $G$. For 
any $C \in \mathcal{C}$, the value of $h_{v'}(C)$ can only change when we go along an edge whose label crosses $C$. 
$h_{t'}(C) \neq h_{s'}(C)$ for every $C$, so for every $C$ at least one of the edge labels of $P'$ must cross $C$. 
This implies that for every $C \in \mathcal{C}$ there is an edge in $G$ which crosses $C$, so $G$ must contain a path 
from $s$ to $t$, as needed.
\end{proof}
To use Proposition \ref{functionstoswitchingprop} we will carefully choose a set of functions $H$, create a vertex for 
each function $f$ in $H$ and then assign that function to the vertex. We will then add all allowed edges allowed by 
condition 3 of Proposition \ref{functionstoswitchingprop} to obtain our switching network $G'$.
\begin{definition}
If $f$ and $g$ are functions $f,g: \mathcal{C} \to \mathbb{R}$, we say that we can go from 
$f$ to $g$ using an edge $e$ if $(f-g)(C) = 0$ for any cut $C$ which is not crossed by $e$.
\end{definition}
\begin{definition}
If we have a set $H$ of functions from $\mathcal{C}$ to $\mathbb{R}$ then if $f,g \in H$ and $G$ is an input graph 
we say that we can go from $f$ to $g$ in $H$ on input $G$ if there is a sequence of functions $f_0, \cdots f_j$ in $H$ and edges 
$e_1, \cdots, e_{j-1} \in E(G)$ such that 
\begin{enumerate}
\item $f_0 = f$ and $f_j = g$
\item For every $i \in [1,j]$ we can go from $f_{i-1}$ to $f_i$ using the edge $e_i$. 
\end{enumerate}
\end{definition}
\begin{proposition}
If $I$ is a set of input graphs each of which contains a path from $s$ to $t$ and $H$ is a set of functions 
from $\mathcal{C}$ to $\mathbb{R}$ such that 
\begin{enumerate}
\item $s' = -1 \in H$ and $t' = 1 \in H$
\item For any $G \in I$ we can go from $s'$ to $t'$ in $H$ on input $G$.
\end{enumerate}
then $m(I) \leq |H \setminus \{s',t'\}|$
\end{proposition}
\begin{proof}
As described above, we create a vertex for each function $f$ in $H$ and then assign that function to the vertex. We then add all 
edges allowed by condition 3 of Proposition \ref{functionstoswitchingprop} to obtain our switching network 
$G'$. Now by Proposition \ref{functionstoswitchingprop} $G'$ is sound. Moreover, from the construction of $G'$ and the definitions, 
for any $G \in I$ $G'$ must accept $G$. The result now follows immediately.
\end{proof}
Now our challenge is to choose the set of functions $H$ and show what steps we can take on functions in $H$ on input $G$. 
We use the following type of function.
\begin{definition}
Given a mulit-set of functions $F$ from $\mathcal{C}$ to $\{-1,+1\}$, define 
$K_F = 1 - 2^{1-|F|}\prod_{f \in F}{(1-f)}$
\end{definition}
\begin{proposition}
For all $C \in \mathcal{C}$, $K_F(C) = 1$ if $f(C) = 1$ for any $f \in F$ and $-1$ otherwise.
\end{proposition}
\begin{remark}
We have the following correspondence between our functions and knowledge about a possible cut $C$ which might not 
be crossed by an edge in $E(G)$ (based on current knowledge about $G$).
\begin{enumerate}
\item $f = -1$ corresponds to knowing nothing about $C$.
\item $f = 1$ corresponds to knowing that $C$ cannot exist i.e. there is a path from $s$ to $t$ in $G$.
\item $f = e_{V}$ corresponds to knowing that an odd number of the vertices of $V$ are in $L(C)$ and the rest are in 
$R(C)$.
\item $f = -e_{V}$ corresponds to knowing that an even number of the vertices of $V$ are in $L(C)$ and the rest are in 
$R(C)$.
\item $K_F$ corresponds to all of the knowledge from all $f \in F$.
\end{enumerate}
\end{remark}
We can take the following steps on these functions.
\begin{proposition}
For any multi-set of functions $F$, $K_{F \cup \{-1\}} = K_F$
\end{proposition}
\begin{proposition}
For any multi-set of functions $F = \{f_1, \cdots f_j\}$, for any $i \in [1,j],$ 
\begin{enumerate}
\item Using the edge $s \to v$ we can multiply $f_i$ by $-e_{\{v\}}$. In other words, we can 
go from $K_{F_1}$ to $K_{F_2}$ using the edge $e = s \to v$ where 
$F_1 = F = \{f_1, \cdots f_j\}$ and $F_2 = \{f_1, \cdots, f_{i-1}, -{f_i}e_{\{v\}}, f_{i+1}, \cdots, f_j\}$ 
\item Using the edge $v \to t$ we can multiply $f_i$ by $e_{\{v\}}$. In other words, we can 
go from $K_{F_1}$ to $K_{F_2}$ using the edge $e = v \to t$ where 
$F_1 = F = \{f_1, \cdots f_j\}$ and $F_2 = \{f_1, \cdots, f_{i-1}, {f_i}e_{\{v\}}, f_{i+1}, \cdots, f_j\}$ 
\end{enumerate}
\end{proposition}
\begin{proof}
In the first case, $K_{F_2} - K_{F_1}$ is a multiple of $-{f_i}e_{\{v\}} - f_i$ which is a multiple of 
$-1-e_{\{v\}}$. Whenever $s \to v$ does not cross $C$, $v \in L(C)$ so $(-1-e_{\{v\}})(C) = 0$.
In the second case, $K_{F_2} - K_{F_1}$ is a multiple of ${f_i}e_{\{v\}} - f_i$ which is a multiple of 
$e_{\{v\}}-1$. Whenever $v \to t$ does not cross $C$, $v \in R(C)$ so $(e_{\{v\}}-1)(C) = 0$.
\end{proof}
\begin{proposition}
If $F$ is a multi-set of functions containing the function $e_{\{v\}}$ then we go from $K_F$ to $K_{F \cup \{e_{\{w\}}\}}$ using 
the edge $v \to w$.
\end{proposition}
\begin{proof}
$K_{F \cup \{e_{\{w\}}\}} - K_F$ is a multiple of $((1-e_{\{w\}})-2)(1 - e_{\{v\}}) = -(1 + e_{\{w\}})(1 - e_{\{v\}})$. 
Whenever $v \to w$ does not cross a cut $C$, $v \in R(C)$ or $w \in L(C)$. In either case, 
$-(1 + e_{\{w\}})(1 - e_{\{v\}})(C) = 0$.
\end{proof}
\begin{remark}
When all of our $f$ are of the form ${\pm}e_V$ for some $V \subseteq \Vmst$ (which will always be the case for our analysis), 
these steps correspond to the following logical deductions about a possible cut $C$ which is not crossed by any edge in $E(G)$
\begin{enumerate}
\item If we knew that an odd(even) number of the vertices of $V$ are in $L(C)$ and we see the edge $s \to v$, 
then an even(odd) number of the vertices in $V \Delta \{v\}$ are in $L(C)$ where $V \Delta \{v\} = V \cup \{v\}$ if $v \notin V$ and 
$V \Delta \{v\} = V \setminus \{v\}$ if $v \in V$.
\item If we knew that an odd(even) number of the vertices of $V$ are in $L(C)$ and we see the edge $v \to t$, 
then an odd(even) number of the vertices in $V \Delta \{v\}$ are in $L(C)$.
\item If we knew that $v$ is in $L(C)$ and we see the edge $v \to w$ then we know that $w \in L(C)$.
\end{enumerate}
\end{remark}
We now have the following intuition. For the reversible pebble game, we kept track of exactly which vertices 
we knew were in $L(C)$. Here we instead take sets of vertices $V \subseteq \Vmst$ where we really care about one vertex $v \in V$ and all 
other vertices of $V$ are lollipops. When we deduce that $v$ is in $L(C)$, we immediately encode this information 
in the parity of the number of vertices of $V$ which are in $L(C)$ (we can do this easily because all other vertices of V 
except $v$ are lollipops). This allows us to temporarily forget which vertex $v$ we cared about and just remember one bit 
of information. When we need the knowledge that $v$ is in $L(C)$ to make a further deduction, we can decode this information and 
then use it. In this way, we will only ever remember two actual vertices at a time. All of the other information will be parity bits 
and remembering which sets of vertices we know the parity for. With this intuition in mind we are now ready to prove 
Theorem \ref{generalupperbound}. We restate the theorem here for convenience.
\vskip.1in
\noindent
{\bf Theorem \ref{generalupperbound}.}
{\it
Let $G$ be an input graph which contains a subgraph $G_0$ such that 
\begin{enumerate}
\item $G_0$ contains a path from $s$ to $t$.
\item All vertices in $V(G) \setminus V(G_0)$ are lollipops.
\item Letting $k = |V(G_0) \setminus \{s,t\}|$, $k \geq 2$ and $n = |\Vmst|$ is divisible by $k$
\end{enumerate}
then for any $z \in [1,k]$, there is a set of sets of sets of vertices 
$\{\mathcal{V}_x\}$ such that 
\begin{enumerate}
\item Each $\mathcal{V}_x = \{V_{x1}, \cdots, V_{xk}\}$ partitions the vertices of $\Vmst$ into $k$ sets 
of $\frac{n}{k}$ vertices.
\item $|\{\mathcal{V}_x\}| \leq 2{(4k)^z}k\lg{n}$
\item $m(G) \leq 2z{2^z}r(G_0,z)|\{\mathcal{V}_x\}|n\lg{n}(2|\{\mathcal{V}_x\}| + 4 + zn)$
\end{enumerate}
}
\begin{proof}
Let $v_1, \cdots, v_k$ be the vertices of $V(G_0) \setminus \{s,t\}$ 
\begin{definition}
Given a permutation $\sigma \in S_{\Vmst}$, we say that a set $\mathcal{V} = \{V_1, \cdots V_k\}$ of disjoint sets of vertices 
matches the state in the reversible pebbling game on $G_0$ with pebbles on vertices $s,v_{i_1}, \cdots, v_{i_j}$ if 
for all $m \in [1,j]$, $V_{i_m} \cap \sigma(V(G_0) \setminus \{s,t\}) = \{\sigma(v_{i_m})\}$.
\end{definition}
\begin{definition}
If we have an input graph $\sigma(G)$, letting $W = \sigma(V(G_0) \setminus \{s,t\})$, if $V$ is a set of vertices 
which has exactly one vertex in common with $W$ then 
\begin{enumerate}
\item We define the root of $V$ to be the vertex $v$ such that $V \cap W = \{v\}$.
\item We define the parity $p(V)$ to be $p(V) = (-1)^{|v_2 \in V \setminus \{v\}: s \to v_2 \in E(G)|}$.
\end{enumerate}
\end{definition}
\begin{lemma}
Assume that we have a set $\{\mathcal{V}_x\} = \{\{V_{x1}, \cdots V_{xk}\}\}$ of sets of disjoint 
sets of vertices, a set of functions $H$, and a set 
of states $S$ of the reversible pebbling game such that it is possible to win the reversible pebbling game on $G_0$ using 
only states in $S$ (plus the starting state and a winning state). 

If we have that whenever $S$ contains the state with pebbles on vertices $s,v_{i_1}, \cdots, v_{i_j}$, 
\begin{enumerate}
\item For all $x$ and $\pm$ signs, $K_{\{{\pm}e_{V_{xi_1}}, \cdots, {\pm}e_{V_{xi_j}}\}} \in H$
\item For all $\sigma \in S_{\Vmst}$ there is an $x$ such that given $\sigma$, $\mathcal{V}_x$ matches this state. Furthermore, for 
all $x$ such that $\mathcal{V}_x$ matches this state,
\begin{enumerate}
\item If for some $m$ there is an edge from a vertex in $\{v_{i_{m_2}}: m_2 \neq m\} \cup \{s\}$ to $v_{i_{m}}$ in $G_0$ and $S$ contains 
the state with pebbles on vertices $\{s,v_{i_1}, \cdots, v_{i_j}\} \setminus \{v_{i_m}\}$ then it is possible 
to go from $K_{\{p(V_{xi_1})e_{V_{xi_1}}, \cdots, p(V_{xi_j})e_{V_{xi_j}}\}}$ to 
$K_{\{p(V_{xi_1})e_{V_{xi_1}}, \cdots, p(V_{xi_j})e_{V_{xi_j}}\} \setminus \{p(V_{xi_{m}})e_{V_{xi_{m}}}\}}$ in $H$ on input 
$\sigma(G)$.
\item If there is an edge from $v_{i_{m}}$ to $t$ in $G_0$ then it is possible to go from 
$K_{\{p(V_{xi_1})e_{V_{xi_1}}, \cdots, p(V_{xi_j})e_{V_{xi_j}}\}}$ to $t' = 1$ in $H$ on input $\sigma(G)$
\end{enumerate}
\item For all $\sigma \in S_{\Vmst}$, if given $\sigma$ both $\mathcal{V}_x$ and $\mathcal{V}_y$ match this state then 
it is possible to go from 
$K_{\{p(V_{xi_1})e_{V_{xi_1}}, \cdots, p(V_{xi_j})e_{V_{xi_j}}\}}$ to 
$K_{\{p(V_{yi_1})e_{V_{yi_1}}, \cdots, p(V_{yi_j})e_{V_{yi_j}}\}}$ in $H$ on input $\sigma(G)$.
\end{enumerate}
then for all $\sigma \in S_{\Vmst}$, it is possible to go from $s' = -1$ to $t' = 1$ in $H$ on input $\sigma(G)$
\end{lemma}
\begin{proof}
The idea is to follow the steps of the winning sequence of moves in the reversible pebble game for $G_0$ while 
always keeping a $\mathcal{V}_x$ which matches the current state. From the conditions in the lemma, removing a pebble or 
winning the game is no problem. If we ever need to add a pebble and $\mathcal{V}_x$ would no longer match the new state then 
note that there must be a $y$ such that $\mathcal{V}_y$ does match the new state. Now by definition $\mathcal{V}_y$ also 
matches the old state. Thus we can first shift from $\mathcal{V}_x$ to $\mathcal{V}_y$ and then make the needed move 
in the pebbling game.
\end{proof}
Now we just need to count the number of functions we need in $H$ to do all of this.
\begin{lemma}\label{reducinglemma}
By adding $2|V|\lceil{\lg{|V|}}\rceil$ functions to $H$, we can guarantee that whenever $V$ has parity $p(V)$ and root $v$ 
we will be able to go from $p(V)e_V$ to $e_{\{v\}}$ in $H$ regardless of the input graph $\sigma(G)$.
\end{lemma}
\begin{proof}
The proof works by induction. The base case $|V| = 1$ is trivial. If $|V| > 1$ then split $V$ into two halves 
$L = \{v_1,\cdots,v_{j_1}\}$ and $R = \{v_{{j_1}+1},\cdots,v_{j_2}\}$ as evenly as possible. Now add the functions \\
$\{{\pm}e_W: W = L \cup \{v_i: j_1 < i \leq m\}$ for some $m \in [j_1+1,j_2]\}$ and the functions \\
$\{{\pm}e_W: W = R \cup \{v_i: 1 \leq i \leq m\}$ for some $m \in [1,j_1]\}$ to $H$.

If $v \in L$ then note that for all $m \in [j_1+1,j_2]$, $v_m$ is a lollipop so if $W = L \cup \{v_i: j_1 < i \leq m\}$ we can use an allowed function step to go from $p(W)e_W$ to $p(W \setminus \{v_m\})e_{W \setminus \{v_m\}}$. This implies that we can go from 
$p(V)e_V$ to $p(L)e_L$ in $H$ regardless of the input graph $\sigma(G)$. Using similar logic, if $v \in R$ then we can go from 
$p(V)e_V$ to $p(R)e_R$ in $H$ regardless of the input graph $\sigma(G)$. By the inductive hypothesis, 
we can add at most $2|L|\lceil{\lg{|L|}}\rceil + 2|R|\lceil{\lg{|R|}}\rceil$ functions to ensure that we can go from either $p(L)e_L$ 
or $p(R)e_R$ (depending on whether $v$ is in $L$ or $R$) to $e_{\{v\}}$ in $H$ regardless of the input graph $\sigma(G)$. 
The total number of functions is $2|V| + 2|L|\lceil{\lg{|L|}}\rceil + 2|R|\lceil{\lg{|R|}}\rceil \leq 2|V|\lceil{\lg{|V|}}\rceil$
\end{proof}
\begin{corollary}\label{reducingcorollary}
If $|V|$ is a set of vertices of size at most $\lceil{\frac{n}{2}}\rceil$ with parity $p(V)$ and root $v$ 
and $F$ is a multi-set of functions containing the function $f = p(V)e_V$ then by adding at most $2n\lg{n}$ functions 
to $H$ we can ensure that we can go from $K_F$ to $K_{F \setminus \{p(V)e_V\} \cup \{e_{\{v\}}\}}$ in $H$ regardless 
of the input graph $\sigma(G)$.
\end{corollary}
\begin{proof}
This can be proved by using the same logic that was used to prove Lemma \ref{reducinglemma} and 
noting that the function steps apply to individual functions $f$ in the multi-set of functions $F$.  
\end{proof}
We now count how many functions we need overall. To shift from $\mathcal{V}_x$ to $\mathcal{V}_y$ on a given state with 
pebbles on vertices $v_{i_1}, \cdots, v_{i_j}$, for 
each $m$ we change each $p(V_{xi_m})e_{V_{xi_m}}$ to $e_{\{\sigma(v_{i_m})\}}$ and then run this process in reverse to reach 
$p(V_{yi_m})e_{V_{yi_m}}$. Doing this for all $m$ takes a total of at most $4zn\lg{n}$ functions. We need to do this 
for all pairs $x,y$, all parities, and all states of the reversible pebble game in $S$, so 
this gives a total of at most $4z{2^z}|S|{|\{\mathcal{V}_x\}|^2}n\lg{n}$ functions.

To go from the state with pebbles on vertices $v_{i_1}, \cdots, v_{i_j}$ to a winning state using an edge 
$v_{i_m} \to t$ we reduce $p(V_{xi_{m}})e_{V_{xi_{m}}}$ to $e_{\{\sigma(v_{i_{m}})\}}$ and then go directly to $t'$ with the 
function step multiplying $e_{\{\sigma(v_{i_{m}})\}}$ by $e_{\{\sigma(v_{i_{m}})\}}$ using the edge $v_{i_m} \to t$. 
This takes at most $2n\lg{n}$ functions. We may need to do this for all $x$, all parities, all states of the reversible 
pebble game in $S$, and all possible $m$, so this gives a total of at most $2z{2^z}|S|{|\{\mathcal{V}_x\}|}n\lg{n}$ 
functions.

To go from the state with pebbles on vertices $v_{i_1}, \cdots, v_{i_j}$ to the state with pebbles on vertices 
$\{v_{i_1}, \cdots, v_{i_j}\} \setminus \{v_{i_m}\}$ using an edge 
$s \to v_{i_m}$ we reduce $p(V_{xi_{m}})e_{V_{xi_{m}}}$ to $e_{\{\sigma(v_{i_{m}})\}}$ and then do the function step 
multiplying $e_{\{\sigma(v_{i_{m}})\}}$ by $-e_{\{\sigma(v_{i_{m}})\}}$ using the edge $s \to v_{i_m}$. 
This takes at most $2n\lg{n}$ functions. We may need to do this for all $x$, all parities, all states of the reversible 
pebble game in $S$, and all possible $m$, so this gives a total of at most $2z{2^z}|S|{|\{\mathcal{V}_x\}|}n\lg{n}$ 
functions.

Finally, to go from the state with pebbles on vertices $v_{i_1}, \cdots, v_{i_j}$ to the state with pebbles on vertices 
$\{v_{i_1}, \cdots, v_{i_j}\} \setminus \{v_{i_{m_2}}\}$ using an edge 
$v_{i_{m_1}} \to v_{i_{m_2}}$, we first reduce $p(V_{xi_{m_1}})e_{V_{xi_{m_1}}}$ to $e_{\{\sigma(v_{i_{m_1}})\}}$ and 
reduce $p(V_{xi_{m_1}})e_{V_{xi_{m_2}}}$ to $e_{\{\sigma(v_{i_{m_2}})\}}$. We then use the function step deleting 
$e_{\{\sigma(v_{i_{m_2}})\}}$ with the edge $v_{i_{m_1}} \to v_{i_{m_2}}$. Finally, we
restore $e_{\{\sigma(v_{i_{m_1}})\}}$ to $p(V_{xi_{m_1}})e_{V_{xi_{m_1}}}$. The first reduction takes at most $2n\lg{n}$ functions, 
but we may need to do it for all $x$, all parities, all states of the reversible pebble game in $S$, and all possible $m_1$. This gives 
a total of at most $2z{2^z}|S|{|X|}n\lg{n}$ functions. Similarly, the final restoration (which is a reduction in reverse) takes a total of 
at most $2z{2^z}|S|{|X|}n\lg{n}$ functions. The second reduction takes at most $2n\lg{n}$ functions, but we may need to do it for 
all $x$, all parities, all states of the reversible pebble game in $S$, all possible roots of $V_{xi_{m_1}}$, and all possible $m_1,m_2$.
This gives a total of at most $2z^2{2^z}|S|{|\{\mathcal{V}_x\}|}n^2\lg{n}$ functions.

Thus, the total number of functions needed for $H$ is at most 
$2z{2^z}r(G_0,z)|\{\mathcal{V}_x\}|n\lg{n}(2|\{\mathcal{V}_x\}| + 4 + zn)$. Now we just 
need to calculate how large $|\{\mathcal{V}_x\}|$ needs to be. We only need to show that there is a match for any possible state 
with $z$ pebbles, as this implies that there is a match for any state with at most $z$ pebbles. 
If each $\mathcal{V}_x$ partitions the vertices perfectly evenly, 
for any given state with $z$ pebbles and permutation $\sigma \in S_{\Vmst}$ the probability of a match is at least 
$(\frac{1}{k})^z(\frac{k-z}{k})^{k-z}$. To see this, randomly place each vertex $v_{i_m}$ one at a time into a set in 
$\mathcal{V}_x$. Each time, the probability of a correct placement is $\frac{(\frac{n}{k})}{n+1-m} \geq \frac{1}{k}$. Now place all 
of the $k-z$ vertices which are not pebbled. If we have already placed $j$ of these vertices, the probability 
that the next vertex will not be placed in the same set as any $v_{i_m}$ is 
\begin{align*}
\frac{\frac{n}{k}(k-z)-j}{n-j-z} &= \frac{k-z}{k} + \frac{j+z}{n-j-z} \cdot \frac{k-z}{k} - \frac{j}{n-j-z}
= \frac{k-z}{k} + \frac{(j+z)(k-z)-kj}{k(n-j-z)} \\
&= \frac{k-z}{k} + \frac{z(k-j-z)}{k(n-j-z)} \geq \frac{k-z}{k}
\end{align*}
\begin{proposition}
For any $z \in (0,k)$, $((\frac{k-z}{k})^{k-z}) > 4^{-z}$
\end{proposition}
\begin{proof}
We split the proof into two cases depending on whether $z \leq \frac{k}{2}$ or $z \geq \frac{k}{2}$. If 
$z \leq \frac{k}{2}$ then we use the following proposition
\begin{proposition}
For any $x$ such that $0 < x \leq \frac{1}{2}$, $(1-x)^{\frac{1}{2x}} \geq \frac{1}{2}$
\end{proposition}
\begin{proof}
Note that $(1-x)^{\frac{1}{2x}} = (e^{\ln{(1-x)}})^{\frac{1}{2x}} = e^{\frac{\ln{(1-x)}}{2x}}$

The Taylor series for $\ln{(1-x)}$ is $\ln{(1-x)} = -\sum_{j \geq 1}{\frac{x^{j}}{j}}$ so 
$\frac{\ln{(1-x)}}{2x} = -\sum_{j \geq 1}{\frac{x^{j-1}}{2j}}$. Comparing term by term, since 
$x \leq \frac{1}{2}$ we have that $\frac{\ln{(1-x)}}{2x} \geq -\sum_{j \geq 1}{\frac{2^{-j}}{j}} = \ln{\frac{1}{2}} = -\ln{2}$.
Thus, $(1-x)^{\frac{1}{2x}} = e^{\frac{\ln{(1-x)}}{2x}} \geq \frac{1}{2}$, as needed
\end{proof}
Plugging in $x = \frac{z}{k}$ we obtain that 
$$(\frac{k-z}{k})^{k-z} > (\frac{k-z}{k})^{k} = (1-x)^{k} = (1-x)^{\frac{z}{x}} = ((1-x)^{\frac{1}{2x}})^{2z} \geq 2^{-2z} = 4^{-z}$$
For the case when $z > \frac{k}{2}$, note that the derivative of 
$\ln{((\frac{k-z}{k})^{k-z})} = (k-z)(\ln{(k-z)} - \ln{k})$ with respect to $z$ is 
$\ln{k} - \ln{(k-z)} - 1$ which is $0$ if $z = \frac{k}{2}$, less than $0$ if $0 < z < \frac{k}{2}$ and bigger than 
$0$ if $\frac{k}{2} < z < k$. This implies that $(\frac{k-z}{k})^{k-z}$ is minimized at $z = \frac{k}{2}$ so since $z > \frac{k}{2}$,
$(\frac{k-z}{k})^{k-z} > (\frac{k-\frac{k}{2}}{k})^{k-\frac{k}{2}} = 2^{-\frac{k}{2}} > 2^{-z} > 4^{-z}$
\end{proof}

Putting everything together, if $0 < z < k$ the probability of a match is at least $(\frac{1}{k})^z(\frac{k-z}{k})^{k-z}$ 
which is more than $(\frac{1}{4k})^z$. Now note that we do not need to worry about the case $z = 0$ beacuse this corresponds to 
not needing to place any pebbles on any vertices except $s$ and $t$ to win the pebble game on $G_0$ which means that $G_0$ 
has an edge from $s$ to $t$. For the case $z = k$, if $z$ vertices of $G_0$ are pebbled then there 
are no vertices of $G_0$ which are not pebbled so the probability of a match is just the probability that 
the pebbled vertices are placed into the correct setes of vertices by $\sigma$ which from before is at least 
$(\frac{1}{k})^z > (\frac{1}{4k})^z$. Thus in all cases the probability of a match is at least $(\frac{1}{4k})^z$.

There are at most $n^k$ possibilities for $\sigma(V(G_0) \setminus \{s,t\})$ and 
at most $k^z$ different states of the pebble game on $G_0$ with $z$ pebbles, so following the same logic we used in the proof 
of Theorem \ref{simpleprobabilisticupperbound}, we can obtain a set of sets of sets of vertices $\mathcal{V}_x$ where
$|\mathcal{V}_x| \leq {(4k)^z}\lg{({n^k}{k^z})} \leq 2k(4k)^z\lg{n}$
\end{proof}
\subsection{Simplified upper bounds}
In this subsection we simplify the upper bound of Theorem \ref{generalupperbound}.
\begin{corollary}\label{generalupperboundcorollaryone}
Let $G$ be an input graph which contains a subgraph $G_0$ such that 
\begin{enumerate}
\item $G_0$ contains a path from $s$ to $t$.
\item All vertices in $V(G) \setminus V(G_0)$ are lollipops.
\end{enumerate}
then taking $k = |V(G_0) \setminus \{s,t\}|$, and $z = p(G_0)$, $m(G)$ is $n^{O(1)}k^{O(z)}$. In particular, \\
$m(G) \leq 8z{2^z}{k^z}{(4k)^{z+1}}n{(\lg{n})^2}({(4k)^{z+1}}\lg{n} + 4 + zn)$
\end{corollary}
\begin{proof}
If $k = 1$ then $m(G) \leq n$ and if $n < 8$ then $m(G) \leq n^3$ so the bound is trivial unless $n \geq 8$ and 
$k \geq 2$. Now if $n$ is divisible by $k$, note that there are $\sum_{j=1}^{z}{{k \choose j}}$ 
possible states for the reversible pebbling game on $G_0$ 
which are not the starting state, not a winning state, and have at most $z$ pebbles. If $z \leq \lceil\frac{k}{2}\rceil$ then 
$\sum_{j=1}^{z}{{k \choose j}} \leq z{{k \choose z}} \leq k^z$. If $z > \lceil\frac{k}{2}\rceil$ then 
$\sum_{j=1}^{z}{{k \choose j}} \leq z{{k \choose \lceil\frac{k}{2}\rceil}} \leq k(k^{\lceil\frac{k}{2}\rceil}) \leq k^z$. Thus, 
either way we have that $r(G_0,z) \leq k^z$. Now by Theorem \ref{generalupperbound}, 
$m(G) \leq 2z{2^z}{k^z}2{(4k)^z}k{(\lg{n})}n{\lg{n}}(2 \cdot 2{(4k)^z}k\lg{n} + 4 + zn)$. Simplifying this expression 
gives that $m(G) \leq z{2^z}{k^z}{(4k)^{z+1}}n{(\lg{n})^2}({(4k)^{z+1}}\lg{n} + 4 + zn)$

If $n$ is not divisible by $k$ then we use the following proposition
\begin{proposition}
If $G$ is an input graph containing a path from $s$ to $t$ and $G_2$ is the input graph with 
$V(G_2) = V(G) \cup \{w\}$ and $E(G_2) = E(G) \cup \{s \to w\}$ then $m(G_2) \geq m(G)$
\end{proposition}
\begin{proof}
Consider a sound monotone switching network $G'_2$ accepting all of the input graphs \\
$\{\sigma(G_2): \sigma \in S_{\Vmst \cup \{w\}}\}$. If we contract all edges of $G'_2$ labeled $s \to w$ and 
delete all edges of $G'_2$ labeled $w \to v$ for any $v \in V(G)$, we will obtain a sound monotone switching network 
$G'$ accepting all of the input graphs $\{\sigma(G): \sigma \in S_{\Vmst}\}$. Thus, 
$m(G) \leq |V(G') \setminus \{s',t'\}| \leq |V(G'_2) \setminus \{s',t'\}| \leq m(G_2)$, as needed.
\end{proof}
Using this proposition, we can add vertices to $V(G)$ until $|\Vmst|$ is divisible by $k$. Let $n_2 = |\Vmst|$ after 
we have added these vertices. $n \geq 8$ and $n_2 \leq 2n$ so 
$$m(G) \leq z{2^z}{k^z}{(4k)^{z+1}}n_2{(\lg{n_2})^2}({(4k)^{z+1}}\lg{n_2} + 4 + zn_2) \leq 
8z{2^z}{k^z}{(4k)^{z+1}}n{(\lg{n})^2}({(4k)^{z+1}}\lg{n} + 4 + zn)$$
\end{proof}
\begin{corollary}
Let $G$ be an input graph which contains a subgraph $G_0$ such that 
\begin{enumerate}
\item $G_0$ contains a path from $s$ to $t$.
\item All vertices in $V(G) \setminus V(G_0)$ are lollipops.
\end{enumerate}
then letting $k = |V(G_0) \setminus \{s,t\}|$, taking $z = \lceil{\lg(k+1)}\rceil$, 
$m(G) \leq {z^2}{2^{5z+8}}{k^{3z+3}}{n^2}{(\lg{n})^2}$.
More precisely,
\begin{enumerate}
\item If $k \leq 2^{\sqrt{\lg{n} - \lg{\lg{n}}}-2}$ then $m(G) \leq {z^2}{2^{3z+6}}{k^{2z+1}}{n^2}{(\lg{n})^2}$
\item If $k \geq 2^{\sqrt{\lg{n} - \lg{\lg{n}}}-2}$ then $m(G) \leq {z^2}{2^{5z+8}}{k^{3z+3}}{n}{(\lg{n})^3}$
\end{enumerate}
\end{corollary}
\begin{proof}
Note that the bound in Corollary \ref{generalupperboundcorollaryone} holds for this $z$ because the bound holds for any $z \geq p(G_0)$ 
as it is an increasing function of $z$ and by Theorem \ref{pebblingnumber} we have that $p(G_0) \leq z$. Also note that since 
$k$ is an integer, $z \leq \lg{k}+1$ so $2^{z+1} \leq 2^{\lg{k}+2} \leq 4k$. Finally, note that the result is trivial unless 
$n \geq 4$ and $k \geq 2$ so we may assume that $n \geq 4$ and $k \geq 2$.

If $k \leq 2^{\sqrt{\lg{n} - \lg{\lg{n}}}-2}$ then $(4k)^{z+1}\lg{n} \leq 2^{(z+1)\sqrt{\lg{n} - \lg{\lg{n}}}}\lg{n} \leq 2^{\lg{n}-\lg{\lg{n}}}\lg{n} \leq n$. Now $z \geq 2$ and $n \geq 4$ so ${(4k)^{z+1}}\lg{n} + 4 + zn \leq 2nz$. Plugging this 
into the bound given by Corollary \ref{generalupperboundcorollaryone} and simplifying gives the needed bound.

If $k \geq 2^{\sqrt{\lg{n} - \lg{\lg{n}}}-2}$ then 
$$kz(4k)^{z+1}\lg{n} \geq z\lg{n}(2^{\sqrt{\lg{n} - \lg{\lg{n}}}})^{z+2} \geq 
z\lg{n}2^{(\sqrt{\lg{n} - \lg{\lg{n}}})\lg{k}} \geq z\lg{n}2^{\lg{n} - \lg{\lg{n}}} = zn$$
Now since $n \geq 4$ and $k \geq 2$ we have that ${(4k)^{z+1}}\lg{n} + 4 \leq kz(4k)^{z+1}\lg{n}$. Thus, 
${(4k)^{z+1}}\lg{n} + 4 + zn \leq 2kz(4k)^{z+1}\lg{n}$. Plugging this into the bound given by
Corollary \ref{generalupperboundcorollaryone} and simplifying gives the needed bound.

Both of these bounds are less than ${z^2}{2^{5z+8}}{k^{3z+3}}{n^2}{(\lg{n})^2}$ so we always have that \\
$m(G) \leq {z^2}{2^{5z+8}}{k^{3z+3}}{n^2}{(\lg{n})^2}$.
\end{proof}
\section{Conclusion}\label{conclusion}
In this paper, we considered the function $m(G)$, which is a complexity measure on graphs $G$ containing a path from $s$ to $t$. 
Roughly speaking, $\lg{(m(G))}$ is the amount of space needed for a monotone algorithm to find a path from $s$ to $t$ 
on input graphs isomorphic to $G$. As shown by our lower bounds, for many directed acyclic graphs $G$, letting $l$ 
be the length of the shortest path from $s$ to $t$, $m(G)$ is determined primarily by $l$. In particular, if no vertex of $G$ 
has short paths to a lot of other vertices or has short paths from a lot of other vertices to it then $m(G)$ is 
$n^{\Theta(\lg{l})}$. However, as shown by our lower bounds, this is not true for all directed acyclic graphs $G$. 
In particular, if all but $k$ vertices of $G$ are lollipops (vertices $v$ for which there is an edge from $s$ to $v$ or an 
edge from $v$ to $t$) then $m(G)$ is $n^{O(1)}k^{O(\lg{l})}$. 

Both the lower and upper bounds are a significant improvement over our previous bounds in \cite{potechin}. 
However, the question of bounding $m(G)$ for general directed acyclic input graphs $G$ and figuring out under which 
conditions we have that $m(G)$ is $n^{\Theta(\lg{l})}$ is wide open. Bounding $m(G)$ for general input graphs $G$ is 
even more wide open.

In proving our lower and upper bounds, we found new techniques which are more robust then previous techniques 
and thus more likely to be extendable to non-monotone analysis. We see no a priori reason why such an extension would 
be impossible. That said, there is a major obstacle which would almost certainly require many more novel ideas 
and techniques to overcome. For non-monotone analysis, the main difference 
is that we would have to have a function for each non-edge of $G$ as well as for each edge of $G$. Unfortunately, 
this means that we can no longer restrict our attention to maximal NO instances and must instead consider a much 
wider class of NO instances. This in turn means that we would need an alternative way to do Fourier anaylsis or a similar 
analysis.

All in all, we have made significant progress. That said, many questions remain wide open and only time will tell 
how far the switching network approach will take us in analyzing space complexity.

\noindent {\bf Acknowledgement.}
The author would like to thank Professor Jonathan Kelner for his advice in writing up this article.

\newpage
\begin{appendix}
\section{Checking well-definedness}\label{welldefinedproof}
In this section we give a proof of Lemma \ref{checkingwelldefinednesslemma}. We recall 
Lemma \ref{checkingwelldefinednesslemma} below.
\vskip.1in
\noindent
{\bf Lemma \ref{checkingwelldefinednesslemma}.}
{\it
Let $G$ be an acyclic input graph containing a path from $s$ to $t$. We may freely choose which non-degenerate edges 
$v \to w$ are relevant for the terms $(\vec{s}_{k,u,g})_A$ so long as the following conditions hold.
\begin{enumerate}
\item If $v \to w$ is relevant for $(\vec{s}_{k,u,g})_A$ then $v \to w$ is relevant for $(\vec{s}_{k_2,u_2,g})_{A_2}$ 
whenever $\{v,w\} \setminus \{s,t\} \subseteq A_2 \subseteq A$, $k_2 = |A_2|$, and $u_2 \leq u$.
\item If $u \to v$ and $v \to w$ are relevant for $(\vec{s}_{k,u,g})_A$ then $u \to w$ is relevant for 
$(\vec{s}_{k-1,u,g})_{A \setminus v}$.
\end{enumerate}
}
\begin{proof}
For convenience we recall the equations for relevance (which come from Lemma \ref{einvariancelemma}) here.
\begin{enumerate}
\item If $A$ contains a vertex $w$ and $s \to w$ is relevant for $(\vec{s}_{k,u,g})_A$ then
\begin{equation}\label{pathfroms}
(\vec{s}_{k,u,g})_A = -(\vec{s}_{k-1,u,g})_{A \setminus \{w\}} + (\vec{s}_{k,u-1,g})_A
\end{equation}
\item If $A$ contains a vertex $v$ and $v \to t$ is relevant for $(\vec{s}_{k,u,g})_A$ then
\begin{equation}\label{pathtot}
(\vec{s}_{k,u,g})_A = (\vec{s}_{k-1,u,g})_{A \setminus \{v\}} - (\vec{s}_{k,u-1,g})_A
\end{equation}
\item If $A$ contains vertices $v,w$ and $v \to w$ is relevant for $(\vec{s}_{k,u,g})_A$ then
\begin{align}\label{pathvtow}
(\vec{s}_{k,u,g})_A &= (\vec{s}_{k-1,u,g})_{A \setminus \{v\}} - (\vec{s}_{k-1,u,g})_{A \setminus \{w\}} 
+ (\vec{s}_{k-2,u,g})_{A \setminus \{v,w\}} \\
&- (\vec{s}_{k-1,u-1,g})_{A \setminus \{v\}} - (\vec{s}_{k-1,u-1,g})_{A \setminus \{w\}} + 
(\vec{s}_{k,u-2,g})_{A} \nonumber
\end{align}
\end{enumerate}
We need to check that whenever there are multiple equations we can apply they all give the same result. We prove this 
by induction on $k+u$. 

Assume there are two paths $v_1 \to v_2$ and $v_3 \to v_4$ for which we have the corresponding equations. We will 
show that when we apply each equation once it does not matter which one we applied first. We have the following 
cases. Note that we cannot have $v_1 = v_4$ and $v_2 = v_3$ because $G$ is acyclic.
\begin{enumerate}
\item If the multi-set $\{v_1,v_2,v_3,v_4\}$ has no repeated vertices besides $s$ and $t$ then 
the order in which the equations are applied does not matter.
\item If $v_1 = s$, $v_2 = v_4 = w$, and $v_3 = v$ then 
\begin{enumerate}
\item If we apply the equation \eqref{pathfroms} corresponding to $s \to w$ first and then the equation 
\eqref{pathvtow} corresponding to $v \to w$ we obtain
\begin{align*}
(\vec{s}_{k,u,g})_A &= -(\vec{s}_{k-1,u,g})_{A \setminus \{w\}} + (\vec{s}_{k,u-1,g})_A \\
&= -(\vec{s}_{k-1,u,g})_{A \setminus \{w\}} + (\vec{s}_{k-1,u-1,g})_{A \setminus \{v\}} - 
(\vec{s}_{k-1,u-1,g})_{A \setminus \{w\}} + (\vec{s}_{k-2,u-1,g})_{A \setminus \{v,w\}} \\
&- (\vec{s}_{k-1,u-2,g})_{A \setminus \{v\}} - (\vec{s}_{k-1,u-2,g})_{A \setminus \{w\}} + 
(\vec{s}_{k,u-3,g})_{A}
\end{align*}
\item If we apply the equation the equation \eqref{pathvtow} corresponding to $v \to w$ 
first and then the equation \eqref{pathfroms} corresponding to $s \to w$ we obtain
\begin{align*}
(\vec{s}_{k,u,g})_A &= (\vec{s}_{k-1,u,g})_{A \setminus \{v\}} - (\vec{s}_{k-1,u,g})_{A \setminus \{w\}} 
+ (\vec{s}_{k-2,u,g})_{A \setminus \{v,w\}} \\
&- (\vec{s}_{k-1,u-1,g})_{A \setminus \{v\}} - (\vec{s}_{k-1,u-1,g})_{A \setminus \{w\}} + 
(\vec{s}_{k,u-2,g})_{A} \\
&= \left((\vec{s}_{k-1,u-1,g})_{A \setminus \{v\}} - (\vec{s}_{k-2,u,g})_{A \setminus \{v,w\}}\right) 
-(\vec{s}_{k-1,u,g})_{A \setminus \{w\}} \\
&+ (\vec{s}_{k-2,u,g})_{A \setminus \{v,w\}} 
- \left((\vec{s}_{k-1,u-2,g})_{A \setminus \{v\}} - (\vec{s}_{k-2,u-1,g})_{A \setminus \{v,w\}}\right) \\
&- (\vec{s}_{k-1,u-1,g})_{A \setminus \{w\}} + 
\left((\vec{s}_{k,u-3,g})_{A} - (\vec{s}_{k-1,u-2,g})_{A \setminus \{w\}}\right) \\
&= (\vec{s}_{k-1,u-1,g})_{A \setminus \{v\}} - (\vec{s}_{k-1,u,g})_{A \setminus \{w\}} - 
(\vec{s}_{k-1,u-2,g})_{A \setminus \{v\}} + (\vec{s}_{k-2,u-1,g})_{A \setminus \{v,w\}} \\
&- (\vec{s}_{k-1,u-1,g})_{A \setminus \{w\}} + 
(\vec{s}_{k,u-3,g})_{A} - (\vec{s}_{k-1,u-2,g})_{A \setminus \{w\}}
\end{align*}
which after rearranging is the same as above.
\end{enumerate}
\item If $v_1 = v_3 = u$, $v_2 = v$, and $v_4 = w$ then if 
we apply the equation \eqref{pathvtow} corresponding to $u \to v$ first and then the equation 
\eqref{pathvtow} corresponding to $u \to w$ we obtain
\begin{align*}
(\vec{s}_{k,u,g})_A &= (\vec{s}_{k-1,u,g})_{A \setminus \{u\}} - (\vec{s}_{k-1,u,g})_{A \setminus \{v\}} 
+ (\vec{s}_{k-2,u,g})_{A \setminus \{u,v\}} \\
&- (\vec{s}_{k-1,u-1,g})_{A \setminus \{u\}} - (\vec{s}_{k-1,u-1,g})_{A \setminus \{v\}} + 
(\vec{s}_{k,u-2,g})_{A} \\
&= (\vec{s}_{k-1,u,g})_{A \setminus \{u\}} - 
(\vec{s}_{k-2,u,g})_{A \setminus \{u,v\}} + (\vec{s}_{k-2,u,g})_{A \setminus \{v,w\}} 
- (\vec{s}_{k-3,u,g})_{A \setminus \{u,v,w\}} \\
&+ (\vec{s}_{k-2,u-1,g})_{A \setminus \{u,v\}} + (\vec{s}_{k-2,u-1,g})_{A \setminus \{v,w\}} - 
(\vec{s}_{k-1,u-2,g})_{A \setminus \{v\}} \\
&+ (\vec{s}_{k-2,u,g})_{A \setminus \{u,v\}} - (\vec{s}_{k-1,u-1,g})_{A \setminus \{u\}}\\
&- (\vec{s}_{k-2,u-1,g})_{A \setminus \{u,v\}} + (\vec{s}_{k-2,u-1,g})_{A \setminus \{v,w\}} 
- (\vec{s}_{k-3,u-1,g})_{A \setminus \{u,v,w\}} \\
&+ (\vec{s}_{k-2,u-2,g})_{A \setminus \{u,v\}} + (\vec{s}_{k-2,u-2,g})_{A \setminus \{v,w\}} - 
(\vec{s}_{k-1,u-3,g})_{A \setminus \{v\}} \\
&+ (\vec{s}_{k-1,u-2,g})_{A \setminus \{u\}} - (\vec{s}_{k-1,u-2,g})_{A \setminus \{w\}} 
+ (\vec{s}_{k-2,u-2,g})_{A \setminus \{u,w\}} \\
&- (\vec{s}_{k-1,u-3,g})_{A \setminus \{u\}} - (\vec{s}_{k-1,u-3,g})_{A \setminus \{w\}} + 
(\vec{s}_{k,u-4,g})_{A} \\
&= (\vec{s}_{k-1,u,g})_{A \setminus \{u\}} + (\vec{s}_{k-2,u,g})_{A \setminus \{v,w\}} 
- (\vec{s}_{k-3,u,g})_{A \setminus \{u,v,w\}} \\
&+ 2(\vec{s}_{k-2,u-1,g})_{A \setminus \{v,w\}} - 
(\vec{s}_{k-1,u-2,g})_{A \setminus \{v\}} - (\vec{s}_{k-1,u-1,g})_{A \setminus \{u\}}\\
&- (\vec{s}_{k-3,u-1,g})_{A \setminus \{u,v,w\}} \\
&+ (\vec{s}_{k-2,u-2,g})_{A \setminus \{u,v\}} + (\vec{s}_{k-2,u-2,g})_{A \setminus \{v,w\}} - 
(\vec{s}_{k-1,u-3,g})_{A \setminus \{v\}} \\
&+ (\vec{s}_{k-1,u-2,g})_{A \setminus \{u\}} - (\vec{s}_{k-1,u-2,g})_{A \setminus \{w\}} 
+ (\vec{s}_{k-2,u-2,g})_{A \setminus \{u,w\}} \\
&- (\vec{s}_{k-1,u-3,g})_{A \setminus \{u\}} - (\vec{s}_{k-1,u-3,g})_{A \setminus \{w\}} + 
(\vec{s}_{k,u-4,g})_{A}
\end{align*}
This expression is symmetric in $v$ and $w$ so it would have been the same if we had applied the 
equations in the opposite order.

The next two cases are symmetric to the previous two.
\item If $v_1 = v_3 = v$, $v_2 = t$, and $v_4 = w$ then 
\begin{enumerate}
\item If we apply the equation \eqref{pathtot} corresponding to $v \to t$ first and then the equation 
\eqref{pathvtow} corresponding to $v \to w$ we obtain
\begin{align*}
(\vec{s}_{k,u,g})_A &= (\vec{s}_{k-1,u,g})_{A \setminus \{v\}} - (\vec{s}_{k,u-1,g})_A \\
&= (\vec{s}_{k-1,u,g})_{A \setminus \{v\}} - (\vec{s}_{k-1,u-1,g})_{A \setminus \{v\}} + 
(\vec{s}_{k-1,u-1,g})_{A \setminus \{w\}} - (\vec{s}_{k-2,u-1,g})_{A \setminus \{v,w\}} \\
&+ (\vec{s}_{k-1,u-2,g})_{A \setminus \{v\}} + (\vec{s}_{k-1,u-2,g})_{A \setminus \{w\}} - 
(\vec{s}_{k,u-3,g})_{A}
\end{align*}
\item If we apply the equation the equation \eqref{pathvtow} corresponding to $v \to w$ 
first and then the equation \eqref{pathtot} corresponding to $v \to t$ we obtain
\begin{align*}
(\vec{s}_{k,u,g})_A &= (\vec{s}_{k-1,u,g})_{A \setminus \{v\}} - (\vec{s}_{k-1,u,g})_{A \setminus \{w\}} 
+ (\vec{s}_{k-2,u,g})_{A \setminus \{v,w\}} \\
&- (\vec{s}_{k-1,u-1,g})_{A \setminus \{v\}} - (\vec{s}_{k-1,u-1,g})_{A \setminus \{w\}} + 
(\vec{s}_{k,u-2,g})_{A} \\
&= (\vec{s}_{k-1,u,g})_{A \setminus \{v\}} - 
\left(-(\vec{s}_{k-1,u-1,g})_{A \setminus \{w\}} + (\vec{s}_{k-2,u,g})_{A \setminus \{v,w\}}\right) \\
&+ (\vec{s}_{k-2,u,g})_{A \setminus \{v,w\}} 
- \left(-(\vec{s}_{k-1,u-2,g})_{A \setminus \{w\}} + (\vec{s}_{k-2,u-1,g})_{A \setminus \{v,w\}}\right) \\
&- (\vec{s}_{k-1,u-1,g})_{A \setminus \{v\}} + 
\left(-(\vec{s}_{k,u-3,g})_{A} + (\vec{s}_{k-1,u-2,g})_{A \setminus \{v\}}\right) \\
&= (\vec{s}_{k-1,u,g})_{A \setminus \{v\}} + (\vec{s}_{k-1,u-1,g})_{A \setminus \{w\}} + 
(\vec{s}_{k-1,u-2,g})_{A \setminus \{w\}} - (\vec{s}_{k-2,u-1,g})_{A \setminus \{v,w\}} \\
&- (\vec{s}_{k-1,u-1,g})_{A \setminus \{v\}} - 
(\vec{s}_{k,u-3,g})_{A} + (\vec{s}_{k-1,u-2,g})_{A \setminus \{v\}}
\end{align*}
which after rearranging is the same as above.
\end{enumerate}
\item If $v_1 = u$, $v_2 = v_4 = w$, and $v_3 = v$ then if 
we apply the equation \eqref{pathvtow} corresponding to $u \to w$ first and then the equation 
\eqref{pathvtow} corresponding to $v \to w$ we obtain
\begin{align*}
(\vec{s}_{k,u,g})_A &= (\vec{s}_{k-1,u,g})_{A \setminus \{u\}} - (\vec{s}_{k-1,u,g})_{A \setminus \{w\}} 
+ (\vec{s}_{k-2,u,g})_{A \setminus \{u,w\}} \\
&- (\vec{s}_{k-1,u-1,g})_{A \setminus \{u\}} - (\vec{s}_{k-1,u-1,g})_{A \setminus \{w\}} + 
(\vec{s}_{k,u-2,g})_{A} \\
&= (\vec{s}_{k-1,u,g})_{A \setminus \{w\}} + 
(\vec{s}_{k-2,u,g})_{A \setminus \{u,v\}} - (\vec{s}_{k-2,u,g})_{A \setminus \{u,w\}} 
+ (\vec{s}_{k-3,u,g})_{A \setminus \{u,v,w\}} \\
&- (\vec{s}_{k-2,u-1,g})_{A \setminus \{u,v\}} - (\vec{s}_{k-2,u-1,g})_{A \setminus \{u,w\}} + 
(\vec{s}_{k-1,u-2,g})_{A \setminus \{u\}} \\
&+ (\vec{s}_{k-2,u,g})_{A \setminus \{u,w\}} - (\vec{s}_{k-1,u-1,g})_{A \setminus \{w\}}\\
&- (\vec{s}_{k-2,u-1,g})_{A \setminus \{u,v\}} + (\vec{s}_{k-2,u-1,g})_{A \setminus \{u,w\}} 
- (\vec{s}_{k-3,u-1,g})_{A \setminus \{u,v,w\}} \\
&+ (\vec{s}_{k-2,u-2,g})_{A \setminus \{u,v\}} + (\vec{s}_{k-2,u-2,g})_{A \setminus \{u,w\}} - 
(\vec{s}_{k-1,u-3,g})_{A \setminus \{u\}} \\
&+ (\vec{s}_{k-1,u-2,g})_{A \setminus \{v\}} - (\vec{s}_{k-1,u-2,g})_{A \setminus \{w\}} 
+ (\vec{s}_{k-2,u-2,g})_{A \setminus \{v,w\}} \\
&- (\vec{s}_{k-1,u-3,g})_{A \setminus \{v\}} - (\vec{s}_{k-1,u-3,g})_{A \setminus \{w\}} + 
(\vec{s}_{k,u-4,g})_{A} \\
&= (\vec{s}_{k-1,u,g})_{A \setminus \{w\}} + 
(\vec{s}_{k-2,u,g})_{A \setminus \{u,v\}} + (\vec{s}_{k-3,u,g})_{A \setminus \{u,v,w\}} \\
&-2(\vec{s}_{k-2,u-1,g})_{A \setminus \{u,v\}} + (\vec{s}_{k-1,u-2,g})_{A \setminus \{u\}} 
- (\vec{s}_{k-1,u-1,g})_{A \setminus \{w\}}\\
&- (\vec{s}_{k-3,u-1,g})_{A \setminus \{u,v,w\}} \\
&+ (\vec{s}_{k-2,u-2,g})_{A \setminus \{u,v\}} + (\vec{s}_{k-2,u-2,g})_{A \setminus \{u,w\}} - 
(\vec{s}_{k-1,u-3,g})_{A \setminus \{u\}} \\
&+ (\vec{s}_{k-1,u-2,g})_{A \setminus \{v\}} - (\vec{s}_{k-1,u-2,g})_{A \setminus \{w\}} 
+ (\vec{s}_{k-2,u-2,g})_{A \setminus \{v,w\}} \\
&- (\vec{s}_{k-1,u-3,g})_{A \setminus \{v\}} - (\vec{s}_{k-1,u-3,g})_{A \setminus \{w\}} + 
(\vec{s}_{k,u-4,g})_{A}
\end{align*}
This expression is symmetric in $u$ and $v$ so it would have been the same if we had applied the 
equations in the opposite order.
\item If $v_1 = s$, $v_2 = v_3 = v$, and $v_4 = w$ then 
\begin{enumerate}
\item If we apply the equation \eqref{pathfroms} corresponding to $s \to v$ first and then the equation 
\eqref{pathvtow} corresponding to $v \to w$ we obtain
\begin{align*}
(\vec{s}_{k,u,g})_A &= -(\vec{s}_{k-1,u,g})_{A \setminus \{v\}} + (\vec{s}_{k,u-1,g})_A \\
&= -(\vec{s}_{k-1,u,g})_{A \setminus \{v\}} + (\vec{s}_{k-1,u-1,g})_{A \setminus \{v\}} - 
(\vec{s}_{k-1,u-1,g})_{A \setminus \{w\}} + (\vec{s}_{k-2,u-1,g})_{A \setminus \{v,w\}} \\
&- (\vec{s}_{k-1,u-2,g})_{A \setminus \{v\}} - (\vec{s}_{k-1,u-2,g})_{A \setminus \{w\}} + 
(\vec{s}_{k,u-3,g})_{A}
\end{align*}
\item If we apply the equation the equation \eqref{pathvtow} corresponding to $v \to w$ 
first and then the equation \eqref{pathfroms} corresponding to $s \to v$ we obtain
\begin{align*}
(\vec{s}_{k,u,g})_A &= (\vec{s}_{k-1,u,g})_{A \setminus \{v\}} - (\vec{s}_{k-1,u,g})_{A \setminus \{w\}} 
+ (\vec{s}_{k-2,u,g})_{A \setminus \{v,w\}} \\
&- (\vec{s}_{k-1,u-1,g})_{A \setminus \{v\}} - (\vec{s}_{k-1,u-1,g})_{A \setminus \{w\}} + 
(\vec{s}_{k,u-2,g})_{A} \\
&= (\vec{s}_{k-1,u,g})_{A \setminus \{v\}} - 
\left((\vec{s}_{k-1,u-1,g})_{A \setminus \{w\}} - (\vec{s}_{k-2,u,g})_{A \setminus \{v,w\}}\right) \\
&+ (\vec{s}_{k-2,u,g})_{A \setminus \{v,w\}} 
- \left((\vec{s}_{k-1,u-2,g})_{A \setminus \{w\}} - (\vec{s}_{k-2,u-1,g})_{A \setminus \{v,w\}}\right) \\
&- (\vec{s}_{k-1,u-1,g})_{A \setminus \{v\}} + 
\left((\vec{s}_{k,u-3,g})_{A} - (\vec{s}_{k-1,u-2,g})_{A \setminus \{v\}}\right) \\
&= (\vec{s}_{k-1,u,g})_{A \setminus \{v\}} - 
(\vec{s}_{k-1,u-1,g})_{A \setminus \{w\}} + 2(\vec{s}_{k-2,u,g})_{A \setminus \{v,w\}} \\ 
&- (\vec{s}_{k-1,u-2,g})_{A \setminus \{w\}} + (\vec{s}_{k-2,u-1,g})_{A \setminus \{v,w\}} 
- (\vec{s}_{k-1,u-1,g})_{A \setminus \{v\}} \\
&+ (\vec{s}_{k,u-3,g})_{A} - (\vec{s}_{k-1,u-2,g})_{A \setminus \{v\}}
\end{align*}
\end{enumerate}
The difference between these two expressions is 
$$2(\vec{s}_{k-1,u,g})_{A \setminus \{v\}} + 2(\vec{s}_{k-2,u,g})_{A \setminus \{v,w\}} 
-2(\vec{s}_{k-1,u-1,g})_{A \setminus \{v\}}$$
which is $0$ if we have the corresponding equation for the path from $s$ to $w$.

The next case is symmetric to this one.
\item If $v_1 = v$, $v_2 = v_3 = w$, and $v_4 = t$ then 
\begin{enumerate}
\item If we apply the equation \eqref{pathtot} corresponding to $w \to t$ first and then the equation 
\eqref{pathvtow} corresponding to $v \to w$ we obtain
\begin{align*}
(\vec{s}_{k,u,g})_A &= (\vec{s}_{k-1,u,g})_{A \setminus \{w\}} - (\vec{s}_{k,u-1,g})_A \\
&= (\vec{s}_{k-1,u,g})_{A \setminus \{w\}} - (\vec{s}_{k-1,u-1,g})_{A \setminus \{v\}} + 
(\vec{s}_{k-1,u-1,g})_{A \setminus \{w\}} - (\vec{s}_{k-2,u-1,g})_{A \setminus \{v,w\}} \\
&+ (\vec{s}_{k-1,u-2,g})_{A \setminus \{v\}} + (\vec{s}_{k-1,u-2,g})_{A \setminus \{w\}} - 
(\vec{s}_{k,u-3,g})_{A}
\end{align*}
\item If we apply the equation the equation \eqref{pathvtow} corresponding to $v \to w$ 
first and then the equation \eqref{pathtot} corresponding to $w \to t$ we obtain
\begin{align*}
(\vec{s}_{k,u,g})_A &= (\vec{s}_{k-1,u,g})_{A \setminus \{v\}} - (\vec{s}_{k-1,u,g})_{A \setminus \{w\}} 
+ (\vec{s}_{k-2,u,g})_{A \setminus \{v,w\}} \\
&- (\vec{s}_{k-1,u-1,g})_{A \setminus \{v\}} - (\vec{s}_{k-1,u-1,g})_{A \setminus \{w\}} + 
(\vec{s}_{k,u-2,g})_{A} \\
&= -(\vec{s}_{k-1,u,g})_{A \setminus \{w\}} + 
\left(-(\vec{s}_{k-1,u-1,g})_{A \setminus \{v\}} + (\vec{s}_{k-2,u,g})_{A \setminus \{v,w\}}\right) \\
&+ (\vec{s}_{k-2,u,g})_{A \setminus \{v,w\}} 
- \left(-(\vec{s}_{k-1,u-2,g})_{A \setminus \{v\}} + (\vec{s}_{k-2,u-1,g})_{A \setminus \{v,w\}}\right) \\
&- (\vec{s}_{k-1,u-1,g})_{A \setminus \{w\}} + 
\left(-(\vec{s}_{k,u-3,g})_{A} + (\vec{s}_{k-1,u-2,g})_{A \setminus \{w\}}\right) \\
&= -(\vec{s}_{k-1,u,g})_{A \setminus \{w\}} - 
(\vec{s}_{k-1,u-1,g})_{A \setminus \{v\}} + 2(\vec{s}_{k-2,u,g})_{A \setminus \{v,w\}} \\ 
&+ (\vec{s}_{k-1,u-2,g})_{A \setminus \{v\}} - (\vec{s}_{k-2,u-1,g})_{A \setminus \{v,w\}} 
- (\vec{s}_{k-1,u-1,g})_{A \setminus \{w\}} \\
&- (\vec{s}_{k,u-3,g})_{A} + (\vec{s}_{k-1,u-2,g})_{A \setminus \{w\}}
\end{align*}
\end{enumerate}
The difference between these two expressions is 
$$-2(\vec{s}_{k-1,u,g})_{A \setminus \{w\}} + 2(\vec{s}_{k-2,u,g})_{A \setminus \{v,w\}} 
-2(\vec{s}_{k-1,u-1,g})_{A \setminus \{w\}}$$
which is $0$ if we have the corresponding equation for the path from $v$ to $t$.
\item If $v_1 = u$, $v_2 = v_3 = v$, and $v_4 = w$ then 
\begin{enumerate}
\item If we apply the equation \eqref{pathvtow} corresponding to $u \to v$ first and then the equation 
\eqref{pathvtow} corresponding to $v \to w$ we obtain
\begin{align*}
(\vec{s}_{k,u,g})_A &= (\vec{s}_{k-1,u,g})_{A \setminus \{u\}} - (\vec{s}_{k-1,u,g})_{A \setminus \{v\}} 
+ (\vec{s}_{k-2,u,g})_{A \setminus \{u,v\}} \\
&- (\vec{s}_{k-1,u-1,g})_{A \setminus \{u\}} - (\vec{s}_{k-1,u-1,g})_{A \setminus \{v\}} + 
(\vec{s}_{k,u-2,g})_{A} \\
&= (\vec{s}_{k-2,u,g})_{A \setminus \{u,v\}} - 
(\vec{s}_{k-2,u,g})_{A \setminus \{u,w\}} + (\vec{s}_{k-3,u,g})_{A \setminus \{u,v,w\}} \\
&- (\vec{s}_{k-2,u-1,g})_{A \setminus \{u,v\}} - (\vec{s}_{k-2,u-1,g})_{A \setminus \{u,w\}} + 
(\vec{s}_{k-1,u-2,g})_{A \setminus \{u\}} \\
&- (\vec{s}_{k-1,u,g})_{A \setminus \{v\}} + (\vec{s}_{k-2,u,g})_{A \setminus \{u,v\}} 
- (\vec{s}_{k-1,u-1,g})_{A \setminus \{v\}} \\
&- (\vec{s}_{k-2,u-1,g})_{A \setminus \{u,v\}} + 
(\vec{s}_{k-2,u-1,g})_{A \setminus \{u,w\}} - (\vec{s}_{k-3,u-1,g})_{A \setminus \{u,v,w\}} \\
&+ (\vec{s}_{k-2,u-2,g})_{A \setminus \{u,v\}} + (\vec{s}_{k-2,u-2,g})_{A \setminus \{u,w\}} - 
(\vec{s}_{k-1,u-3,g})_{A \setminus \{u\}} \\
&+ (\vec{s}_{k-1,u-2,g})_{A \setminus \{v\}} - (\vec{s}_{k-1,u-2,g})_{A \setminus \{w\}} 
+ (\vec{s}_{k-2,u-2,g})_{A \setminus \{v,w\}} \\
&- (\vec{s}_{k-1,u-3,g})_{A \setminus \{v\}} - (\vec{s}_{k-1,u-3,g})_{A \setminus \{w\}} + 
(\vec{s}_{k,u-4,g})_{A}
\end{align*}
Ordering these terms from largest to smallest $k+u$ we obtain
\begin{align*}
(\vec{s}_{k,u,g})_A &= - (\vec{s}_{k-1,u,g})_{A \setminus \{v\}} + 2(\vec{s}_{k-2,u,g})_{A \setminus \{u,v\}} - 
(\vec{s}_{k-2,u,g})_{A \setminus \{u,w\}} - (\vec{s}_{k-1,u-1,g})_{A \setminus \{v\}} \\
&+ (\vec{s}_{k-3,u,g})_{A \setminus \{u,v,w\}} -2(\vec{s}_{k-2,u-1,g})_{A \setminus \{u,v\}} \\
&+ (\vec{s}_{k-1,u-2,g})_{A \setminus \{u\}} 
+ (\vec{s}_{k-1,u-2,g})_{A \setminus \{v\}} - (\vec{s}_{k-1,u-2,g})_{A \setminus \{w\}} \\
&- (\vec{s}_{k-3,u-1,g})_{A \setminus \{u,v,w\}} 
+ (\vec{s}_{k-2,u-2,g})_{A \setminus \{u,v\}} + (\vec{s}_{k-2,u-2,g})_{A \setminus \{u,w\}} 
+ (\vec{s}_{k-2,u-2,g})_{A \setminus \{v,w\}} \\
&- (\vec{s}_{k-1,u-3,g})_{A \setminus \{u\}} 
- (\vec{s}_{k-1,u-3,g})_{A \setminus \{v\}} - (\vec{s}_{k-1,u-3,g})_{A \setminus \{w\}} + 
(\vec{s}_{k,u-4,g})_{A}
\end{align*}
\item If we apply the equation \eqref{pathvtow} corresponding to $v \to w$ first and then the equation 
\eqref{pathvtow} corresponding to $u \to v$ we obtain
\begin{align*}
(\vec{s}_{k,u,g})_A &= (\vec{s}_{k-1,u,g})_{A \setminus \{v\}} - (\vec{s}_{k-1,u,g})_{A \setminus \{w\}} 
+ (\vec{s}_{k-2,u,g})_{A \setminus \{v,w\}} \\
&- (\vec{s}_{k-1,u-1,g})_{A \setminus \{v\}} - (\vec{s}_{k-1,u-1,g})_{A \setminus \{w\}} + 
(\vec{s}_{k,u-2,g})_{A} \\
&= -(\vec{s}_{k-2,u,g})_{A \setminus \{u,w\}} + 
(\vec{s}_{k-2,u,g})_{A \setminus \{v,w\}} - (\vec{s}_{k-3,u,g})_{A \setminus \{u,v,w\}} \\
&+ (\vec{s}_{k-2,u-1,g})_{A \setminus \{u,w\}} + (\vec{s}_{k-2,u-1,g})_{A \setminus \{v,w\}} - 
(\vec{s}_{k-1,u-2,g})_{A \setminus \{w\}} \\
&+ (\vec{s}_{k-1,u,g})_{A \setminus \{v\}} + (\vec{s}_{k-2,u,g})_{A \setminus \{v,w\}} 
- (\vec{s}_{k-1,u-1,g})_{A \setminus \{v\}} \\
&- (\vec{s}_{k-2,u-1,g})_{A \setminus \{u,w\}} + 
(\vec{s}_{k-2,u-1,g})_{A \setminus \{v,w\}} - (\vec{s}_{k-3,u-1,g})_{A \setminus \{u,v,w\}} \\
&+ (\vec{s}_{k-2,u-2,g})_{A \setminus \{u,w\}} + (\vec{s}_{k-2,u-2,g})_{A \setminus \{v,w\}} - 
(\vec{s}_{k-1,u-3,g})_{A \setminus \{w\}} \\
&+ (\vec{s}_{k-1,u-2,g})_{A \setminus \{u\}} - (\vec{s}_{k-1,u-2,g})_{A \setminus \{v\}} 
+ (\vec{s}_{k-2,u-2,g})_{A \setminus \{u,v\}} \\
&- (\vec{s}_{k-1,u-3,g})_{A \setminus \{u\}} - (\vec{s}_{k-1,u-3,g})_{A \setminus \{v\}} + 
(\vec{s}_{k,u-4,g})_{A}
\end{align*}
Ordering these terms from largest to smallest $k+u$ we obtain
\begin{align*}
(\vec{s}_{k,u,g})_A &= (\vec{s}_{k-1,u,g})_{A \setminus \{v\}} + 2(\vec{s}_{k-2,u,g})_{A \setminus \{v,w\}} - 
(\vec{s}_{k-2,u,g})_{A \setminus \{u,w\}} - (\vec{s}_{k-1,u-1,g})_{A \setminus \{v\}} \\
&- (\vec{s}_{k-3,u,g})_{A \setminus \{u,v,w\}} +2(\vec{s}_{k-2,u-1,g})_{A \setminus \{v,w\}} \\
&+ (\vec{s}_{k-1,u-2,g})_{A \setminus \{u\}} 
- (\vec{s}_{k-1,u-2,g})_{A \setminus \{v\}} - (\vec{s}_{k-1,u-2,g})_{A \setminus \{w\}} \\
&- (\vec{s}_{k-3,u-1,g})_{A \setminus \{u,v,w\}} 
+ (\vec{s}_{k-2,u-2,g})_{A \setminus \{u,v\}} + (\vec{s}_{k-2,u-2,g})_{A \setminus \{u,w\}} 
+ (\vec{s}_{k-2,u-2,g})_{A \setminus \{v,w\}} \\
&- (\vec{s}_{k-1,u-3,g})_{A \setminus \{u\}} 
- (\vec{s}_{k-1,u-3,g})_{A \setminus \{v\}} - (\vec{s}_{k-1,u-3,g})_{A \setminus \{w\}} + 
(\vec{s}_{k,u-4,g})_{A}
\end{align*}
\end{enumerate}
The difference between these expressions is 
\begin{align*}
&2(\vec{s}_{k-1,u,g})_{A \setminus \{v\}} + 2(\vec{s}_{k-2,u,g})_{A \setminus \{v,w\}} - 
2(\vec{s}_{k-2,u,g})_{A \setminus \{u,v\}} - 2(\vec{s}_{k-3,u,g})_{A \setminus \{u,v,w\}} \\
&+ 2(\vec{s}_{k-2,u-1,g})_{A \setminus \{u,v\}} + 
2(\vec{s}_{k-2,u-1,g})_{A \setminus \{v,w\}} - 2(\vec{s}_{k-1,u-2,g})_{A \setminus \{v\}}
\end{align*}
which is $0$ if we have the corresponding equation for the path from $u$ to $w$.
\end{enumerate}
\end{proof}
\end{appendix}

\begin{thebibliography}{1}

\bibitem{randomwalk} R. Aleliunas, R. M. Karp, R. J.Lipton, L. Lov\'{a}sz, and C. Rackoff. 
Random walks, universal traversal sequences, and the complexity of maze problems. FOCS 1979

\bibitem{cbennet} C. Bennet. Time/Space trade-offs for reversible computation. SIAM Journal on Computing
18 no. 4, p. 766-776, 1989 

\bibitem{potechinsiuman} S. Chan, A. Potechin. Tight bounds for monotone switching networks via fourier analysis. STOC 2012

\bibitem{averagecase} Y. Filmus, T. Pitassi, R. Robere, and S. Cook. Average
case lower bounds for monotone switching networks. FOCS 2013

\bibitem{nlconlone} N. Immerman. Nondeterministic Space is Closed Under Complementation. 
SIAM J. Comput. 17, p. 935-938, 1988

\bibitem{potechin} A. Potechin. Bounds on monotone switching networks for directed connectivity. arXiv:0911.0664v6

\bibitem{undirectedgraph} O. Reingold. Undirected ST-connectivity in Log-Space. STOC 2005

\bibitem{savitch} W. J. Savitch. Relationship between nondeterministic and deterministic tape classes. J.CSS 4, 
p. 177-192, 1970

\bibitem{nlconltwo} R. Szelepcs\'{e}nyi. The method of forcing for nondeterministic automata. 
Bull. EATCS 33, p. 96-100, 1987

\bibitem{trifonov} V. Trifonov. An O(log n log log n) space algorithm for undirected st-connectivity.
Proceedings of the thirty-seventh annual ACM symposium on Theory of computing, May 2005
\end{thebibliography}
\end{document}